\newif\ifFULL
\newif\ifACM
\newif\ifIEEE
\newif\ifLNCS
\newif\ifJOUR
\newif\ifCONF
\newif\ifrebuttal\rebuttalfalse
\IEEEtrue
\CONFtrue

\ifCONF
\documentclass[conference]{IEEEtran}
\IEEEoverridecommandlockouts

\fi

\ifJOUR
\documentclass[journal]{IEEEtran}
\IEEEoverridecommandlockouts

\fi

\ifACM

\fi

\ifFULL
\documentclass[a4paper]{article}
\usepackage{fullpage}
\fi

\ifLNCS
\documentclass[runningheads]{llncs}
\fi

\usepackage{framed}
\ifIEEE
\usepackage{amsthm}
\fi
\ifFULL
\usepackage{amsthm}
\fi 
\usepackage{color}
\usepackage{graphicx}
\usepackage{framed}
\usepackage{tikz}
\usepackage{mathtools}
\usepackage{textcomp}
\usepackage{pgf}
\usepackage{pgfplots}
\usepackage{pgfplotstable}
\usepgfplotslibrary{groupplots}
\usepackage{booktabs}
\usepackage{adjustbox}						
\usepackage{pifont}%
 \usepackage{float}
\pgfplotsset{width=8cm,compat=1.9}
\usepackage{footnote}
\usepackage{subfigure}
\makesavenoteenv{tabular}
\makesavenoteenv{table}
\usepackage{hyperref}

\usepackage{listings,xcolor}

\usepackage{amsmath,amssymb,amsfonts}
\usepackage{stmaryrd}
\usepackage{algorithmic}
\usepackage{graphicx}
\usepackage{textcomp}
\usepackage{xcolor}
\usepackage{enumitem}
\def\BibTeX{{\rm B\kern-.05em{\sc  i\kern-.025em b}\kern-.08em
    T\kern-.1667em\lower.7ex\hbox{E}\kern-.125emX}}

\usepackage{forest}
\usetikzlibrary{arrows,automata,calc,decorations,matrix,positioning,shapes}
\usetikzlibrary{decorations.markings}
\usetikzlibrary{decorations.pathmorphing}
\usetikzlibrary{shapes.geometric}
\usetikzlibrary{shapes.misc}
\usetikzlibrary{plotmarks}
\usetikzlibrary{calendar,fpu}

\tikzset{
    moon colour/.style={
        moon fill/.style={
            fill=#1
        }
    },
    sky colour/.style={
        sky draw/.style={
            draw=#1
        },
        sky fill/.style={
            fill=#1
        }
    },
    southern hemisphere/.style={
        rotate=180
    }
}

\makeatletter
\pgfcalendardatetojulian{2010-01-15}{\c@pgf@counta} 
\def\synodicmonth{29.530588853}
\newcommand{\moon}[2][]{%
    \edef\checkfordate{\noexpand\in@{-}{#2}}%
    \checkfordate%
    \ifin@%
        \pgfcalendardatetojulian{#2}{\c@pgf@countb}%
        \pgfkeys{/pgf/fpu=true,/pgf/fpu/output format=fixed}%
        \pgfmathsetmacro\dayssincenewmoon{\the\c@pgf@countb-\the\c@pgf@counta-(7/24+11/(24*60))}%
        \pgfmathsetmacro\lunarage{mod(\dayssincenewmoon,\synodicmonth)}
        \pgfkeys{/pgf/fpu=false}
    \else%
        \def\lunarage{#2}%
    \fi%
    \pgfmathsetmacro\leftside{ifthenelse(\lunarage<=\synodicmonth/2,cos(360*(\lunarage/\synodicmonth)),1)}%
    \pgfmathsetmacro\rightside{ifthenelse(\lunarage<=\synodicmonth/2,-1,-cos(360*(\lunarage/\synodicmonth))}%
    \tikz [moon colour=white,sky colour=black,#1]{
        \draw [moon fill, sky draw] (0,0) circle [radius=1ex];
        \draw [sky draw, sky fill] (0,1ex)
            arc (90:-90:\rightside ex and 1ex)
            arc (-90:90:\leftside ex and 1ex)
            -- cycle;
    }%
}

\pgfplotsset{compat=newest} 
\pgfplotsset{plot coordinates/math parser=false}

\usepackage{varioref}
\usepackage{url}

\ifIEEE
\renewenvironment{description}{\begin{LaTeXdescription}}{\end{LaTeXdescription}}
\newtheorem{definition}{Definition}
\newtheorem{theorem}{Theorem}
\newtheorem{lemma}{Lemma}
\newtheorem*{remark}{Remark}
\fi
\ifFULL
\newtheorem{definition}{Definition}
\newtheorem{theorem}{Theorem}
\newtheorem{lemma}{Lemma}

\fi

\newcommand*\emptycirc[1][0.8ex]{\moon{14}} 
\newcommand*\halfcirc[1][0.8ex]{\moon{8}}
\newcommand*\fullcirc[1][0.8ex]{\moon{0}} 
\usepackage{multirow}

\usepackage[normalem]{ulem}

\newcounter{reviewer}
\newcounter{comment}[reviewer]
\newcounter{shepherd}
\newcounter{scomment}[shepherd]
\renewcommand{\thereviewer}{\Alph{reviewer}}
\renewcommand{\thecomment}{\thereviewer.\arabic{comment}}

\def\reviewer{%
	\refstepcounter{reviewer}%
	\section*{\textsf{Comments from Reviewer }\thereviewer}}
\def\comment{%
	\refstepcounter{comment}%
	\noindent\subsection*{\textbf{Comment}~\thecomment}\ignorespaces}

\def\reply{{\textbf{Reply }}}
\def\reviewnote#1{\ifrebuttal\expandafter\marginpar{\ref{#1}}\fi}
\def\reviewnotemulti#1{\ifrebuttal\expandafter\marginpar{#1}\fi}

\newenvironment{note}%
{\ifrebuttal%
	\begin{color}{red}
		\else  
		\begin{color}{black}
			\fi}
		{\end{color}}

	\def\rem#1{}
\def\add#1{\begin{note}{#1}\end{note}}

\newcommand{\Func}{\mathcal{F}}
\newcommand{\Sim}{\mathsf{S}}
\newcommand{\advZ}{\mathsf{Z}}

\newcommand{\alg}{\mathsf{A}}
\newcommand{\bin}{\{0,1\}}

\newcommand{\num}{n}
\newcommand{\thr}{t}

\newcommand{\secpar}{\lambda}
\newcommand{\Prob}[1]{\mathrm{Pr}\left[#1\right]}
\newcommand{\advD}{\mathsf{D}}
\newcommand{\advA}{\mathsf{A}}
\newcommand{\advS}{\mathsf{S}}
\newcommand{\getsr}{{\:{\leftarrow{\hspace*{-3pt}\raisebox{.75pt}{$\scriptscriptstyle\$$}}}\:}}

\newcommand{\REAL}{\mathbf{REAL}}
\newcommand{\IDEAL}{\mathbf{IDEAL}}
\newcommand{\HYBRID}{\mathbf{HYBRID}}
\newcommand{\cind}{\overset{c}{\approx}}

\newcommand{\drate}{\delta}
\newcommand{\FuncCR}{\mathcal{F}_\mathsf{CR}^*}

\newcommand{\wit}{\ensuremath{w}}

\newcommand{\coins}{\ensuremath{\mathsf{coins}}}

\newcommand{\Commit}{\ensuremath{\mathsf{Commit}}}

\newcommand{\reconstruct}{\ensuremath{\mathsf{Recon}}}
\newcommand{\share}{\ensuremath{\mathsf{Share}}}

\newcommand{\FuncML}{\ensuremath{\Func^*_{\sf ML}}}
\newcommand{\lock}{\ensuremath{\mathtt{lock}}}
\newcommand{\locked}{\ensuremath{\mathtt{locked}}}
\newcommand{\redeem}{\ensuremath{\mathtt{redeem}}}
\newcommand{\payout}{\ensuremath{\mathtt{payout}}}

\newcommand{\abort}{\ensuremath{\mathtt{abort}}}

\newcommand{\cominput}{\ensuremath{\mathtt{input}}}
\newcommand{\comoutput}{\ensuremath{\mathtt{output}}}
\newcommand{\cont}{\ensuremath{\mathtt{continue}}}

	\newcommand{\npv}{\eta}

\newcommand{\hyb}{\ensuremath{\mathbf{Hyb}}}
\newcommand{\Enc}{\ensuremath{\mathsf{Enc}}}
\newcommand{\Dec}{\ensuremath{\mathsf{Dec}}}
\newcommand{\dist}{\ensuremath{\mathcal{D}}}

\newcommand{\payback}{\ensuremath{\mathsf{payback}}}

\newcommand{\NN}{\mathbb{N}}
\newcommand{\abs}[1]{\lvert#1\rvert}
\newcommand{\cX}{\mathcal{X}}
\newcommand{\rv}[1]{\mathbf{#1}}
\newcommand{\aux}{z}
\newcommand{\func}{f}
\newcommand{\gunc}{g}
\newcommand{\party}{\mathsf{P}}

\newcommand{\calI}{\mathcal{I}}
\newcommand{\calC}{\mathcal{C}}
\newcommand{\calH}{\mathcal{H}}

\newcommand{\deposit}{\ensuremath{d}}
\newcommand{\return}{\ensuremath{r}}
\newcommand{\rate}{\ensuremath{\delta}}

\newcommand{\calD}{\mathcal{D}}
\newcommand{\calR}{\mathcal{R}}
\newcommand{\calU}{\mathcal{U}}

\newcommand{\Cost}{\ensuremath{\chi}}
\newcommand{\Tx}{\mathtt{tx}}
\newcommand{\ctx}{c}
\newcommand{\msg}{\mu}
\newcommand{\com}{\mathit{\gamma}}

\newcommand{\rndcom}{\alpha}

\newcommand{\calM}{\mathcal{M}}
\newcommand{\calS}{\mathcal{S}}
\newcommand{\shares}{\sigma}
\newcommand{\calK}{\mathcal{K}}
\newcommand{\key}{\kappa}

\newcommand{\crh}[1]{\mathsf{H}(#1)}

\def\bitcoinB{\leavevmode
	{\setbox0=\hbox{\textsf{B}}%
		\dimen0\ht0 \advance\dimen0 0.2ex
		\ooalign{\hfil \box0\hfil\cr
			\hfil\vrule height \dimen0 depth.2ex\hfil\cr
		}%
	}%
}

\newcommand{\btc}{\bitcoinB}
\newcommand{\btchtx}{\ensuremath{\mathsf{htx}}}
\newcommand{\btctx}{\ensuremath{\mathsf{tx}}}
\newcommand{\btcidx}{\ensuremath{\mathsf{idx}}}
\newcommand{\btctxsimp}{\ensuremath{\mathsf{tx{simp}}}}
\newcommand{\btcscript}{\ensuremath{\phi}}
\newcommand{\btcinscript}{\ensuremath{w}}
\newcommand{\btcs}{\ensuremath{S}}
\newcommand{\btcr}{\ensuremath{R}}

\newcommand{\hk}{\ensuremath{h}}
\newcommand{\kk}{\key}
\newcommand{\kt}{\ensuremath{\tau}}

\newcommand{\pk}{\ensuremath{pk}}
\newcommand{\sk}{\ensuremath{sk}}
\newcommand{\sig}{\ensuremath{\mathsf{Sign}}}
\newcommand{\ver}{\ensuremath{\mathsf{Vrfy}}}

\newcommand{\setN}{\mathbb{N}}

\newcommand{\LMech}{\texttt{L}}
\newcommand{\LLMech}{\texttt{LL}}
\newcommand{\CLMech}{\texttt{CL}}
\newcommand{\PLMech}{\texttt{PL}}
\newcommand{\CPLMech}{\texttt{CPL}}
\newcommand{\ALMech}{\texttt{AL}}
\newcommand{\MLMech}{\texttt{ML}}
\newcommand{\CMLMech}{\texttt{CML}}
\newcommand{\IMPC}{\texttt{IMPC}}

\newcommand{\Ladder}{\texttt{Ladder}}
\newcommand{\LockedLadder}{\texttt{Locked Ladder}}
\newcommand{\CompactLadder}{\texttt{Compact Ladder}}
\newcommand{\PlantedLadder}{\texttt{Planted Ladder}}
\newcommand{\AmortizedLadder}{\texttt{Amortized Ladder}}
\newcommand{\MultiLock}{\texttt{Multi-Lock}}
\newcommand{\CompactML}{\texttt{Compact ML}}
\newcommand{\CompactPL}{\texttt{Compact PL}}
\newcommand{\InsuredMPC}{\texttt{Insured MPC}}

\newcommand{\inp}{x}
\newcommand{\out}{y}
\newcommand{\round}{\ensuremath{t}}
\newcommand{\timeout}{\ensuremath{\tau}}

\newcommand{\ignore}[1]{}

\newcolumntype{R}[2]{%
    >{\adjustbox{angle=#1,lap=\width-(#2)}\bgroup}%
    l%
    <{\egroup}%
}
\newcommand*\rot{\multicolumn{1}{R{30}{1em}}}

\usepackage{color}

\begin{document}

\ifCONF
\title{Cryptographic and Financial Fairness \\ 
\large Proofs, Experimental Data and Supplementary Material}
\fi

\ifCONF
\author{\IEEEauthorblockN{Daniele Friolo}
\IEEEauthorblockA{\emph{Sapienza University of Rome}\\
friolo@di.uniroma1.it}
\and
\IEEEauthorblockN{Fabio Massacci}
\IEEEauthorblockA{\emph{University of Trento} and\\
\emph{Vrije Universiteit Amsterdam}\\
fabio.massacci@ieee.org}
\and
\IEEEauthorblockN{Chan Nam Ngo\textsuperscript{*}\thanks{\textsuperscript{*}The majority of the work presented here was performed while the author was working for the University of Trento, IT.}}
\IEEEauthorblockA{\emph{Kyber Network}\\
nam.ngo@kyber.network}
\and
\IEEEauthorblockN{Daniele Venturi}
\IEEEauthorblockA{\emph{Sapienza University of Rome}\\
venturi@di.uniroma1.it}
}
\fi

\maketitle

\begin{abstract}
A recent trend in  multi-party computation is to achieve 
{\em cryptographic fairness} via monetary penalties, 
{\em i.e.}\ each honest player  either obtains the output 
or receives a compensation in the form of a 
cryptocurrency.
We pioneer another type of fairness, {\em financial 
fairness}, that is closer to the real-world valuation of 
financial transactions. Intuitively, a penalty protocol is 
financially fair if the \emph{net present cost of 
participation} (the total value of cash inflows less cash 
outflows, weighted by the relative discount rate) is the 
same for all honest participants, even when some parties 
cheat.

We formally define the notion, show 
several impossibility results based on game theory,
and analyze the practical effects of (lack of) financial
fairness if one was to run the protocols for
real on Bitcoin using Bloomberg's dark pool trading.

For example, we show that the ladder protocol 
(CRYPTO'14), and its variants (CCS'15 and CCS'16), fail 
to achieve financial fairness both in theory and in 
practice, while the penalty protocols of Kumaresan and Bentov (CCS'14) and Baum, David and Dowsley (FC'20) are financially fair. 
\ifCONF
This version contains formal definitions, detailed security proofs, demos and experimental data in the appendix.
\fi
\end{abstract}

\ifIEEE
\fi

\section{Introduction}\label{sec:intro}
\ifCONF In \fi multi-party computation (MPC), a set of $n$ players wishes to evaluate a 
joint function $f$ of their private inputs in such a way that nothing beyond the function's 
output is revealed~\cite{LindellP08a}. An important property in MPC is the so-called {\em 
cryptographic fairness}, which intuitively says that corrupted parties learn the output 
only if honest parties learn it as well. Unfortunately, without assuming {\em honest majority} ({\em e.g.},\ for two parties)
there are concrete examples of functions for which 
cryptographic fairness is impossible to achieve~\cite{Cleve86}.

To circumvent this impossibility, several solutions have been proposed:
restricted functionalities~\cite{GordonHKL11}, partial 
fairness~\cite{GordonK12,MoranNS16}, gradual release 
protocols~\cite{GarayMPY11}, 
optimistic models~\cite{CachinC00}, and incentivized 
computation~\cite{AsharovLZ13}.
A recent trend
is to guarantee  
cryptographic fairness 
via
monetary compensation (a.k.a.\ {\em cryptographic fairness with penalties}).

This approach gained momentum as decentralized payment systems ({\em e.g.},\ Bitcoin and Ethereum)
offer a convenient way to realize such {\em penalty 
protocols}~\cite{AndrychowiczDMM14,BentovK14,KumaresanB14,kumaresan2015use,kumaresan2016improvements,KumaresanB16,KiayiasZZ16,BentovKM17,david2017kaleidoscope}. The main idea is 
that each party can publish a transaction containing a time-locked deposit 
which can be redeemed by honest players in case of malicious aborts during a protocol run. On the other hand, if no 
abort happens, a deposit owner can redeem the corresponding transaction by showing 
evidence of having completed the protocol.




A concern that is \emph{not} discussed in penalty 
protocols  is the amount of money that should be put into 
escrow nor the time it should stay there. The assumption is 
that it does not matter because all parties would {eventually get their money back.}
While true when the deposit $\deposit$ is a symbol in a 
crypto paper, things differ when $\deposit$ is a noticeable 
amount in a bank account. 

In fact, empirical studies show that
people have a strong preference for immediate payments, 
and receiving the same amount of money later than others 
is often not acceptable~\cite{brown2015empirical}. Even 
for the wealthy, there is the opportunity cost of not 
investing it in better 
endeavors~\cite{lee2016myopic}. For example, in a 
classical experimental 
study~\cite{benzion1989discount}, individuals asked to 
choose between immediate 
delivery of money and a deferred payment (for amounts 
ranging from \$40 to \$5000 )
exhibited a discount rate close to the official borrowing rate. 
These results are consistent across countries ({\em e.g.},\ \cite{benzion1989discount} in the US 
and~\cite{ahlbrecht1997empirical} in Germany). 
Individuals and companies exhibit 
varying degree of risk 
aversion~\cite{brown2015empirical,lee2016myopic}, but 
they all agree that money paid or received ``now'' has a 
greater value than the same 
amount received or paid 
``later''~\cite{angeletos2001hyperbolic}, and that 
small deposits are always preferable to large deposits.
In FinTech, where the base chip $\deposit$ to 
play is a million US\$~\cite{hatch2009reforming,massacci2018fmex}, 
the timings and amounts of deposits can make a huge difference in practice.

Indeed, as shown later in our experimental evaluation, a party could suffer up to 0.49\% loss as devaluation by locking deposit in a long protocol execution time.
Given the observations above, the following research question arises:
\begin{quote}
\emph{How ``fair'' is a cryptographically fair protocol with penalties from a ``financial'' point of view?}
\end{quote}

\add{Remark that, the problem of judging whether an MPC protocol has never surfaced in the past until the recent trend of achieving cryptographic fairness of a protocol via monetary penalty protocols that require locking and then releasing the fund of the MPC participants.}

\textbf{Our Contribution.}
As a useful guide to the extant and future literature, in this 
paper we recall two traditional criteria and provide a new one
for characterizing and evaluating penalty protocols: security models and assumptions, protocol efficiency and, for the first time, \emph{financial fairness}, defined and motivated below.

\emph{Traditional Criterion \#1: Security Models and Assumptions.}
Following the principles of modern cryptography, a secure protocol should be accompanied with a formal proof of security in a well defined framework. The standard definitions for MPC (with and without penalties) follow the simulation-based paradigm and are reviewed in \S\ref{sec:background-security}
, along with the main assumptions required for proving security.

\emph{Traditional Criterion \#2: Protocol Efficiency.}
The efficiency of penalty protocols over blockchains is typically measured w.r.t.: 
(i) the number of transactions sent to the public ledger (relative to the total transaction fees);
(ii) number of interaction rounds with the public ledger; and
(iii) the script complexity, that intuitively corresponds to the public ledger miners' verification load and the space occupied on the public ledger.\footnote{Recent advances in off-chain execution such as RollUp~\cite{kalodner2018arbitrum} can indeed circumvent the efficiency issue. However, the opportunity cost upon locking the deposits into the penalty protocol on-chain when boostraping the off-chain protocol is still an issue.}
We elaborate more on these efficiency criteria in \S\ref{sec:eff_cri}.

\emph{New Criterion: Financial Fairness.} In a nutshell, a penalty protocol is financially fair if the difference between the total \emph{discounted} value of cash inflows and the total \emph{discounted} value of cash outflows of honest parties at the end of the protocol is the same (even when some parties cheat). 
In \S\ref{sec:ffair}, we discuss the principle of financial fairness at a level of abstraction that captures a large class of penalty protocols implementing monetary compensation via any kind of currency (with or without smart contracts).\footnote{
From a technical perspective all participants to a protocol are bound to it, otherwise it is technically impossible. We do not discuss utilities across different protocols, but utilities within runs of the same protocol. The same apply to crypto-fairness which is a property of a protocol. 
    As such, we are not comparing rewards of different protocols, e.g. we are not comparing participant A using protocol 1 and participant B using protocol 2, because they can \emph{not} do it. All participants in a protocol do have the same protocol! A protocol must be agreed by all participants and therefore it must be fair for all participants. 
}
\footnote{Remark that financial fairness is a criteria and not a methodology. \add{Our aim is to show the alternatives between financial fairness, efficiency and security.}}

\textbf{Implications for Practice.}
We argue that the lack of attention to financial aspects is also one of the causes behind lack of adoption of the aforementioned penalty protocols. Our focus is at first on designers.\footnote{
For example our ``illustrative protocol'' (c.f.~\ref{subsec:newprot}) guides designers into the alternative choices where parallel runs of protocols could increase the speed of wealth accumulation. If the complexity of locally computed share is negligible, and the tasks are in abundance (an infinite number of tasks can be run in parallel) then all financially fair MPC protocols would be equally attractive as they might  be constrained by the funds (for deposits) available to each participant. However it might not be possible, as running the illustrative protocol in parallel might be insecure and vulnerable to attacks.} Users can use the model and the software we have used for simulations and our software to check whether the protocol meets their interests.\footnote{
Available at \url{https://github.com/namnc/financial_fairness}.}
Our scenario with Bitcoin and Bloomberg Tradebook is an illustrative example, specific timing determines absolute differences between expected utilities in different protocols. While the protocol preference order may not change when time-discounting varies, the effect might also vary. We chose those because Bitcoin is the most popular decentralized network and Bloomberg because it is a practical industrial application where security is essential (as cryptographic and financial fairness are both important: who would use a protocol that one can scott-free abort or one has to pay more to achieve the same purpose).
In the fully general scenario the value of coins might also change during time (e.g. fluctuations in currencies rate), and this can be considered by in the evaluation function.  
There are both financial (e.g. currency options) or technical (e.g. stable coins) instruments to shield oneself from fluctuations.

\textbf{Relation to Expected Utility Theory.}
Expected Utility Theory (EUT) and Mean-Variance (MV) analysis are not relevant to Financial Fairness and our protocols analysis as they have to do with a portfolio evaluation in presence of (typically) stochastic uncertainty of returns. 
For example one uses MV or EUT to establish that an investment in assets X is ``better'' than an investment in assets Y because X might return 10\% of the initial investments and do so with very low volatility (i.e. variance) while Y might return 20\% but do so with a very high volatility or even generate a loss. Depending on the risk aversion of the user,  A and B might therefore be better than X and Y.
In penalty protocols there is no stochastic uncertainty in the returns of the deposits or the withdrawals are entirely deterministic. There is only an epistemic uncertainty as a player might decide to stop the protocol beforehand.
The only possible stochastic uncertainty is the final evaluation of the overall ``outerprotcol'' if such outcome have ones (e.g. the result of the actual poker hand). However this is immaterial to the penalty protocol to arrive to the final state.\footnote{\add{Obviosuly if $q$ are traded for other different currencies (e.g. one intially $q$s are euro and later one would like to cash withdrawals in US dollars, there might be stochastics uncertainties. But this is due to the fact that we change the notational money mid-way and this would apply to any system in presence of exchange rates among goods. This issue is immaterial to the present paper as for every given penalty protocol, once a deposit is done in q-units, the same q-units are withdrawn, albeit at a later amount of time.}}\footnote{
One can of course generalize the section on the Utility to use a utility function and use the exponential discount with risk neutral evaluation as an example for the calculation. However for the sake of simplicity we use the simplest discount function that is sufficient for illustrative purpose.
}

\textbf{Summary of Results.}
We first review the main state-of-the-art penalty protocols in \S\ref{sec:protocols}, namely \Ladder~\cite{BentovK14},
\MultiLock~\cite{KumaresanB14}, \InsuredMPC~\cite{BaumDD18},  \CompactLadder~\cite{kumaresan2016improvements}, \LockedLadder~\cite{kumaresan2015use},  \AmortizedLadder~\cite{KumaresanB16},  \PlantedLadder~\cite{kumaresan2016improvements}, and \CompactPL~\cite{kumaresan2016improvements}, then we do an exhausting comparison using the above criteria. In \S\ref{subsec:newprot} we introduce a new protocol, namely \CompactML, to illustrate \add{the alternative design that meet } financial fairness but does not achieve UC-security.
A comprehensive summary of the comparisons done throughout the paper is given in Tab.~\ref{fig:summary}.

In \S\ref{sec:compsec}, we evaluate and compare penalty protocols in terms of security model and assumptions (\S\ref{subsec:assumptions}) and asymptotic protocol efficiency (\S\ref{sec:eff_sim}).
In \S\ref{sec:thfin} we introduce a theoretical analysis of financial fairness in terms of possibility and impossibility results.
In \S\ref{sec:negative}, we  study financial fairness via empirical simulations on the differences in deposits and net present values as the number of parties increases to a level that would be needed for
realistic FinTech applications like the Bloomberg Tradebook.

\begin{table*}[!t] 
	\centering
	\caption{Comparing Penalty Protocols}
	\label{fig:summary}
	\begin{minipage}{1.8\columnwidth}
	We compare state-of-the-art penalty protocols under the criteria of security models and assumptions, efficiency, and financial fairness. 
	The last four protocols are for multi-stage functionalities, i.e. functionalities maintaining a persistent state for multiple rounds of interaction, while the first five protocol realize non-reactive functionalities, i.e. functionalities who are reset after the function has been computed. Regarding efficiency, a more detailed discussion with a concrete analysis tailored to Bitcoin implementations is given in \S\ref{sec:compsec}. 
\end{minipage}

	\begin{tabular}{llcc}
			\centering
		\textbf{Criteria} & \textbf{Description} & 
		\begin{tabular}{lllll}
		\rot{\Ladder~\cite{BentovK14}}  &
	\rot{\MultiLock~\cite{KumaresanB14}} &
	\rot{\InsuredMPC~\cite{BaumDD18}} &
	\rot{\CompactLadder~\cite{kumaresan2016improvements}} &
	\rot{\CompactML\ (\S\ref{subsec:newprot})}
	\end{tabular}
					&
							\begin{tabular}{llll}
		\rot{\LockedLadder~\cite{kumaresan2015use}}  &
	\rot{\AmortizedLadder~\cite{KumaresanB16}} &
	\rot{\PlantedLadder~\cite{kumaresan2016improvements}} &
	\rot{\CompactPL~\cite{kumaresan2016improvements}} 
	\end{tabular} 
		\\
		\hline
		\#1 \emph{Security (model)} &
		
		\begin{tabular}{ll}
		 \fullcirc & Universal Composability \\ 
		 \halfcirc & Sequential Composability \\
		 \emptycirc & Not Provably Secure
		 
			\end{tabular}
			&
			 	\begin{tabular}{lllll}
					\fullcirc\  &
					\fullcirc &
					\fullcirc &
					\emptycirc &
					\halfcirc
				\end{tabular}	
				& 
				\begin{tabular}{llll}
					\fullcirc\  &
					\fullcirc &
					\fullcirc &
					\emptycirc
				\end{tabular}
				\\
				\hline
				\#1 \emph{Security (assumptions)} &
					\begin{tabular}{ll}
		 \fullcirc & Plain Model \\ 
		 \halfcirc\ & Random Oracle Model \\
		 \emptycirc & Not Provably Secure
			\end{tabular} &
				\begin{tabular}{lllll}
					\fullcirc  &
					\fullcirc &
					\moon{8} &
					\emptycirc &
					\fullcirc
				\end{tabular}
						& 
				\begin{tabular}{llll}
					\fullcirc  &
					\fullcirc &
					\fullcirc &
					\emptycirc
				\end{tabular}
			\\ \hline
		\#2 \emph{Efficiency} (rounds) &
			\begin{tabular}{ll}
			\fullcirc & Constant \\
			\halfcirc & Linear in the num. of parties \\
			\emptycirc & Quadratic in the num. of parties
			\end{tabular}
		& 
		\begin{tabular}{lllll}
					\halfcirc  &
					\fullcirc &
					\fullcirc &
					\halfcirc &
					\fullcirc
				\end{tabular}
					& 
				\begin{tabular}{llll}
					\emptycirc  &
					\emptycirc &
					\halfcirc &
					\halfcirc
				\end{tabular}
		\\

			\hline
					\#2 \emph{Efficiency} (transactions) &
			\begin{tabular}{ll}
			\fullcirc & Linear in the num. of parties \\
			\emptycirc & Quadratic in the num. of parties
			\end{tabular}
		& 
		\begin{tabular}{lllll}
					\fullcirc  &
					\emptycirc &
					\fullcirc &
					\fullcirc &
					\emptycirc
				\end{tabular}
					& 
				\begin{tabular}{llll}
					\emptycirc  &
					\emptycirc &
					\fullcirc &
					\fullcirc
				\end{tabular}
		\\

			\hline
			\#2 \emph{Efficiency} (complexity) & 
						\begin{tabular}{ll}
			\fullcirc & Output independent \\
			\emptycirc & Output dependent
			\end{tabular}
			&
					\begin{tabular}{lllll}
					\emptycirc  &
					\emptycirc &
					\emptycirc &
					\fullcirc &
					\fullcirc
				\end{tabular}
							& 
				\begin{tabular}{llll}
					\emptycirc  &
					\emptycirc &
					\emptycirc & 
					\fullcirc
				\end{tabular}
			\\ \hline
		\#3 \emph{Financial Fairness} &
		\begin{tabular}{lll}
			\fullcirc & Financially fair \\
			\emptycirc & Financially unfair
			\end{tabular} &
				\begin{tabular}{lllll}
					\emptycirc  &
					\fullcirc &
					\fullcirc &
					\emptycirc &
					\fullcirc
				\end{tabular}
							& 
				\begin{tabular}{llll}
					\emptycirc  &
					\emptycirc &
					\emptycirc &
					\emptycirc
				\end{tabular}
					\\\hline

	\end{tabular}
\end{table*}

In particular, in \S\ref{sec:optunfair} we show  that the protocols 
in~\cite{BentovK14,kumaresan2015use,KumaresanB16,kumaresan2016improvements} 
are only financially 
viable for the ``big guy'' with deep pockets beyond the 
money at stake in the 
protocol. ``Small guys'' must rush to be first, or 
participating would be out of reach.   
Furthermore, the latter happens even in the practical case 
of \emph{optimistic computation} for honest 
parties~\cite{KumaresanB14}: Playing first or last can yield 
a gap of several basis points (the units for
discount rates of financial institutions).

Another surprising finding is that the CCS'15 
\LockedLadder~\cite{kumaresan2015use} is better than its ``improved'' version, the 
CCS'16 \PlantedLadder~\cite{kumaresan2016improvements}, in terms of financial fairness.
Importantly, these negative results hold regardless of which 
technology is used in order to implement the penalty 
protocol (be it simple transactions, or smart contracts). 

One may wonder whether financial fairness of the above 
protocols can be saved by using a small collateral, or 
by running a financially unfair protocol in a round-robin 
fashion. Interestingly, in \S\ref{sec:colla}, we show using 
game theory that these approaches are deemed to fail for 
all practical purposes. 

Finally, in \S\ref{sec:conclusions}, we conclude the paper with lessons learned and by listing open problems for further research.

Additionally, in Fig. \ref{fig:symbols} we provided a table of the mathematical symbols used throughout the paper.

\section{Security and efficiency}\label{sec:background}
\begin{figure}
\begin{center}
	\begin{tabular}{lr}
			\hline
			\textbf{Symbol} & \textbf{Meaning} \\
			\hline
			$\Func_\func$ & Ideal Functionality for computing $\func$\\ 
			$\Func_\func^*$ & Ideal Functionality for $\func$  with coins\\
			$\pi_f$ & Protocol for computing  $f$ \\
			$x_i$ & Input of $\party_i$ to $\Func_\func$ \\
			$y$ & Output of the function $f$\\

			$\Cost_i(\tau)$ & Net Present Cost of party $i$ after $\tau$ rounds\\
			$\deposit_{i,\round}$ & Deposit of party $\party_i$ at round $\round$ \\
	       	$r_{i,\round}$ & Refund to party $\party_i$ at round $\round$\\
			$\com_i$ & Commitment of party $\party_i$ \\
			$\shares_i$ & Output share of party $\party_i$\\
			$\rndcom_i$ & Commitment opening of $\party_i$\\
			  	$\phi_{i,j}$ & Predicate bounding $\party_i$ and $\party_j$ (\LMech)\\
			$\npv(t)$ & Discount rate at round $t$ \\
			$\kappa_i$ & Secret key of party $\party_i$ (\CLMech, \CMLMech)\\
			\hline
		\end{tabular}
				\end{center}
		\caption{Symbols used throughout the paper}
		\label{fig:symbols}
		\end{figure}
\subsection{Security Models and the Real-Ideal Paradigm.}\label{sec:background-security}
To define security of an $n$-party protocol $\pi$ for computing a function $\func$, we compare an execution of $\pi$ in the real world with an ideal process where the parties simply send their inputs to an ideal functionality $\Func_\func$ that evaluates $\func$ on behalf of the players.
$\Func_\func$ (acting as a trusted party) takes all the parties' inputs $(x_i)_{i \in [n]}$  privately, and outputs the value $f_i(x_1,\ldots,x_n)$ to each party $i \in [n]$.\footnote{Output privacy must be ensured even when the functions $f_i$ are different since they can depend on other parties' inputs.}
In the real world, where parties directly exchange messages between themselves, such trusted party does not exist. Therefore it cannot be used by the real honest parties to privately evaluate the function $f_i$. A protocol is said to be secure if the two worlds are (computationally) indistinguishable.

An important feature of simulation-based security is {\em composability}. Intuitively, this property refers to the guarantee that an MPC protocol securely realizing an ideal functionality, continues to do so even if used as a sub-protocol in a larger protocol, which greatly simplifies the design and security analysis of MPC protocols.
The most basic form of composition is known as {\em sequential composability}, which roughly corresponds to the assumption that each sub-protocol is run sequentially and in isolation. 
A much stronger flavor of composition is the so-called {\em universal composability} (UC)~\cite{Canetti20,Canetti01,canetti2002universally}, which instead corresponds to the more realistic scenario where many secure protocols are executed together. In UC, both the real and ideal world are coordinated by an environment $\mathcal{Z}$ that can run multiple interleaved executions of different protocols. We say that $\pi$ $\thr$-securely computes $\Func_\func$ if for all PPT (Probabilistic Polynomial-Time) adversaries $\advA$ corrupting at most $t$ parties in the real execution, there exists an efficient simulator $\advS$ in the ideal execution, such that no efficient environment $\advZ$ interacting with the adversary in both worlds can tell apart the output of $\advA$ in the real world from the output of $\advA$ when its view is simulated by $\advS$. 

In the case of sequential composition, $\advZ$ is replaced by a distinguisher $\advD$ handling inputs to the parties and waiting to receive the output and an arbitrary function of $\advA$'s view at some point.\footnote{$\advZ$ actually defines an environment of protocols running altogether, whilst D can distinguish only single executions.} The latter allows the simulator to internally control the adversary, {\em e.g.} by rewinding it.
Since in the UC setting the interaction between $\advZ$ and $\advA$ can be arbitrary, eventual rewinds of the adversary from the simulator can be spotted by the environment, and thus input extraction techniques adopted by the simulator cannot be based on rewinding (this is usually called \emph{straight-line simulation}). 

\subsection{Hybrid Model, Cryptographic Fairness with Penalties.}\label{subsec:hybrid}
Composability can be formalized using the so-called {\em hybrid model}. In the $\Func_\gunc$-hybrid model a
protocol $\pi$ is augmented with an ideal functionality for securely computing 
$\gunc:(\bin^*)^n \rightarrow (\bin^*)^n$. 
The trusted party for $\Func_\gunc$ may be used a number of times 
throughout the execution of $\pi$; in the case of universal composability, each functionality uses different session ids $sid$ and sub-session ids $ssid$ to keep track of concurrent executions.
The UC theorem states that if $\pi_g$ securely computes $\Func_\gunc$, and $\pi_f$ securely computes some functionality $\Func_\func$ in
the $\Func_\gunc$-hybrid model, then $\pi_f^{\pi_g}$ ({\em i.e.},\ the protocol where each call to $\Func_\gunc$ is replaced with an independent run of protocol $\pi_g$) securely computes $\Func_\func$.
In the case of sequential composability, the latter only holds under the assumption that honest parties send their inputs to the trusted party corresponding to the hybrid ideal functionality in the same 
round, and do not send other messages until they receive the output. 
Since those functionalities can be accessed globally, $\advZ$, as well as $\advS$ and $\advA$ in the ideal world and the honest parties and $\advA$ in the real world can interact with it by only sending queries.

As shown by Cleve~\cite{Cleve86}, the standard MPC definition can not be achieved 
for certain functionalities. 
Indeed a party may abort the computation,
and, in case there is no honest majority, it might 
irreversibly block the protocol. 
In particular, an attacker could violate cryptographic fairness by learning the output whilst no honest party does.
Following~\cite{BentovK14}, we extend MPC to 
the setting of ``MPC with coins''\footnote{In \cite{BentovK14}, the authors did not formally prove the composition theorem of their augmented model. However, as the authors point out, in principle the UC composition theorem should be allowed analogously to  \cite{Canetti01}.}, where each party is provided with his
own wallet and safe\footnote{To ensure indistinguishaibility between real and ideal world it is crucial that the environment is only allowed to access and modify the wallet of each party, but not the safe. Precise details about the model can be found in \cite{BentovK14}. }. We use $\coins(x)$ to represent a coin of value $x$, and
denote special functionalities dealing with coins with the apex $^*$. If a party owns
$\coins(x)$ and deposits (resp.\ receives) $\coins(\deposit)$ (resp.\ $\coins(\return)$), it will own $\coins(x-\deposit)$ (resp.\ $\coins(x+\return)$). 
To define fairness with penalties, we modify the ideal world using the following ideal
functionality $\Func_\func^*$: 
	(i) At the outset, $\Func_\func^*$ receives 
	the inputs and a deposit from each party; the coins deposited by the malicious 
	parties must be enough to compensate all honest players in case of abort. 
	(ii) Then, in 
	the output phase, the functionality returns the deposit to the honest parties; if the 
	adversary deposited enough coins, it is given the chance to look at the output, and finally decides whether to continue delivering the output to the 
	honest players, or to abort, in which case all honest players are compensated using 
	the penalty amount deposited during the input phase.
In this setting, we say that a protocol $\pi$ $\thr$-securely computes $\Func_\func^*$ with penalties (in the hybrid model).
\footnote{Since the blockchain environment admits interleaved execution of protocols, UC security is a strict requirement when part of a protocol realizing an MPC with coins functionality is run on-chain.}

\subsection{Security Assumptions.} \label{subsec:additional}
In this paper, when we mention security we implicitly mean that security holds in the plain model ({\em i.e.}, without assuming trusted setup or idealized primitives) and under standard assumptions ({\emph{e.g.}, DDH). 
In the Random Oracle Model (ROM), all the parties  have access to a truly-random hash function. In particular, when a value $v$ is given as an input from a party to the RO (Random Oracle), the latter samples a random answer, stores the pair $(v,r)$, and outputs $r$ to the party. If the RO is queried on the same value $v$ multiple times, the same answer $r$ is output. In the ROM, the simulator needs to further simulate the interaction between the parties and the RO. While doing so, the simulator may program the output of the RO at specific inputs to particularly convenient random-looking values. This powerful feature is known as random-oracle programmability.
In the setting of generalized UC, the RO is defined as a global ideal functionality $\mathcal{G}_\mathsf{RO}$. In this case, the simulator can only interact with the RO by sending queries to it, which severely limits random-oracle programming.
Security proofs in the ROM only guarantee heuristic security. This is because ROs do not exist in the real world, and thus a security proof in the ROM only guarantees that the protocol remains secure so long as the hash function is close enough to behave as a truly-random function. 
Even worse, there exist (albeit contrived) cryptoschemes that are secure in the ROM but become always insecure for any possible instantiation of the RO with a real-world hash function~\cite{DCanettiGH98,CanettiGH04}.
While the above may look controversial, security proofs in the ROM are generally considered useful as they  guarantee that any security vulnerability can only depend on the hash function. 
Imagine now a cryptographic primitive, {\em e.g.} an encryption scheme, using a hash function modeled as a RO.
Consider an MPC protocol for evaluating the encryption function in a distributed setting where the secret key and the message is somehow shared between the players. The latter typically requires to implement the encryption function as a circuit; however, this is not possible because the encryption algorithm needs to invoke the random oracle which cannot be implemented as a circuit.
Unfortunately, some of the state-of-the-art protocols, for example \cite{kumaresan2016improvements}, rely on the assumption that a random oracle can be implemented as a circuit\footnote{The efficient instantiation of \LLMech\ \cite{kumaresan2015use} (which we decided to not analyze in this work) has the same issue, since the underlying statistical binding commitment scheme that will be used inside the MPC circuit is instantiated with a random oracle. However, since statistical binding commitments can be instantiated also from standard assumptions, \LLMech\ still retains UC security.}. Thus, they fail to achieve any kind of provable security.

\subsection{On-chain and Off-chain Efficiency.}\label{sec:eff_cri}
The efficiency of a penalty protocol can be broken down into two parts: off-chain and on-chain efficiency.
The former refers to traditional MPC efficiency in terms of: the number of communication rounds, the required bandwidth, and the computational complexity;
the latter refers to efficiency in terms of the interaction between the blockchain and the miners in terms of: the number of transactions, the number of round executed on-chain, and the script complexity.
On-chain efficiency in a penalty protocol is much more important compared to off-chain efficiency, as:
	i) the number of transactions determine the transaction fee that a penalty protocol incurs;
	ii) the number of rounds executed on-chain determine how long the protocol runs, as a round executed on-chain requires for a transaction to be confirmed which corresponds to, {\em e.g.},  6 blocks ({\em i.e.,} 1 hour) in Bitcoin;
	iii) the script complexity needs to be multiplied with the number of miners, which could be more than 100K.\footnote{As of Dec 2020, the Bitcoin mining pool called Slushpool (\url{https://slushpool.com/stats/?c=btc}) has 116157 active miners.}
	and iv) off-chain complexity is not dependent of blockchain's block generation rate and transaction throughput.
 Transaction, round and script complexity can asymptotically depend on the security parameter $\secpar$, the number of players $n$, the size of the output of the function $m = |f|$, and the number of stages (for multistage protocols).

\section{Financial Fairness}\label{sec:ffair}

\subsection{Economics Principles.}\label{sec:escrow_overview}
To capture financial fairness, economists introduced the concept of \emph{net 
	present value} and 
\emph{discount rate}. The former tells us how much an amount of money 
received (or paid) later (at time $\round$) is discounted w.r.t.\ the same amount of
money received (or paid) 
now (at time $\round=0$). The difference in value between two adjacent instances  
is captured by the discount rate.\footnote{In general, the discount rate may depend on the 
	risk aversions of the 
	players~\cite{lee2016myopic}, or the confidence in the certainty of future 
	payments~\cite{brown2015empirical}. 
	The net present value may also have different functional forms 
	({\em e.g.}\ exponential, hyperbolic, etc.) or different values for
	borrowing or receiving money~\cite{angeletos2001hyperbolic}.}

\emph{The Cost of Participation.}
Let $\npv_i(\round)$
be the function representing the net
present value at the beginning of the protocol ({\em i.e.},\ at time $0$) of a unit coin that is 
transacted at a later round\footnote{For simplicity, we think of a round as a single time unit.} ({\em i.e.},\ at round \round), according to the $i$-th party's own discount rate.
Let $\deposit_{i,\round}$ be the 
coins put into escrow by player $i$ during 
round $\round$, and let $\return_{i,\round}$ be the coins that the same player 
receives  at round $\round$ (possibly including  
compensating penalties extracted from 
misbehaving parties). 
Given a sequence of deposits $\deposit_{i,\round}$ 
and refunds $\return_{i,\round}$ made by $\party_i$ at rounds $\round \in [0,\timeout]$ of a 
protocol running up to time $\timeout$,
the \emph{net present cost of participation} for $\party_i$ is then
\begin{equation}\label{eq:npv}
\Cost_i(\timeout):= \sum\nolimits_{\round\in[0,\timeout]} (\deposit_{i,\round} - \return_{i,\round})\cdot\npv_i(\round).
\end{equation}

	
	The intuition behind the net present value calculated using Eq.~\eqref{eq:npv} is that 
	money received at a later round $t'$ is \emph{less valued} than the money received at the current round $t$.
	A real world analogy would be the money that is needed to buy a property, e.g a car or a house, now,
	is always not enough, to buy the very same property in 2 years, e.g. due to inflation.

	Our model is agnostic to the rates of the parties. We use a global $\eta$ for sake of simplicity. Furthermore, most economic models assume market conditions in equilibrium as otherwise arbitration will happen. Indeed, the parties might have different utility functions and one can generalize it but it is immaterial for our discussion as we focus on the opportunity cost upon locking deposits. Some specific users might be risk averse, others on the other hand might be myiopic and it is unknown to the protocol designer. As such one cannot base a design on the assumption that the protocol only works if there is a risk-seeker participant or else.
	}
	
	Let us consider a toy example as follows. 
	Suppose the discount rate is 50\% (hourly rate);
	at the beginning of the protocol, a party deposits 100\$ into the blockchain,
	after 1 hour, the party withdraws 50\$, and after 2 hours, the party withdraws the remaining 50\$.
	The net present value for the party in this case, would be 
	\begin{enumerate}
		\item 50/(1+0.5) = 33.3\$ (first withdrawal, 
		$\return(1) = -50$ and $\npv(1) = \frac{1}{(1+0.5)^1}$
	)
		\item 50/[(1+0.5)*(1+0.5)] = 22.2\$ (second withdrawal,  $\return(2) = -50$ and 
		$\npv(2) = \frac{1}{(1+0.5)^2}$
	)
	\end{enumerate}
	which equals 55.5\$.
	The cost of participation will thus be
	\begin{enumerate}
		\item 100\$ (deposit, $\deposit(0) = 100$ and
		$\npv(0) = 1$
		)
		\item -55.5\$ (net present value at the end)
	\end{enumerate}
	which equals 44.5\$.

\emph{The Payment Interest.} The basic fixed interest rate model (used for home mortgages) is sufficient to 
show the marked financial 
unfairness of some protocols. 
Tab.~\ref{tab:overnight}
reports the December 2019 rates used by US depository institutions, measured in \emph{basis points} (bps, 1/100th of 1\%). Those are the 
rates at which depository institutions ({\em e.g.}\ commercial banks) can deposit 
money in each other (in the US) to adjust 
their capital requirements. 
A quantity of money paid or 
received by a party $\party_i$ after a time $\round$ is cumulatively 
discounted at a constant discount rate $\rate$, {\em i.e.} $\npv_i(\round)$ can be computed
using a standard algorithm that we report in the appendix.

\begin{table}[t]
	\centering 
	\caption{US depository inst. rates in 2019 (per annum).}\label{tab:overnight}
	\begin{footnotesize}
		\begin{tabular}{lrrrrr}
			\hline
			\textbf{Fund} & \textbf{Max} & \textbf{Min} & \textbf{Median} & \textbf{Median Hour} & \textbf{Median Minute} \\
			\hline
			EFFR & 245 & 155 & 238 & 0.0272 & 0.0005 \\ 
			SOFR & 525 & 152 & 238 & 0.0272 & 0.0005 \\
			\hline
		\end{tabular}
	\end{footnotesize}
	\begin{minipage}{0.9\columnwidth}\footnotesize
		A basis point, bps, is 1/100th of 1\%. It is the standard unit of 
		measure for interest rates at which depository institutions can deposit 
		money in each other to adjust their capital requirements.
		Median Hour/Minute are those of Hourly and Minute Rate converted from Median Yearly Interest Rate.
	\end{minipage}
	\vspace*{-\baselineskip}
	\label{table:usdep}
\end{table}


Notice that we deliberately do not consider the value that parties might give to protocol's 
outputs ({\em i.e.},\ obtaining the output may be significantly more valuable to party $i$ 
than party $j$). This issue is definitely relevant  from the viewpoint of 
\emph{protocol participants} to decide whether the whole MPC hassle (with or without penalties) is worth the bother.
However, the outcome's valuation should be at least fair from the viewpoint of a \emph{protocol designer}: all 
parties being equal the construction should be fair for them all, and they should not 
be discriminated by going first, last, or third. 
In a formal model this could be simply achieved using an utility 
function so that instead of $(-\deposit_{i,\round} + \return_{i,\round})\cdot\npv_i(\round)$ 
we have ${\mathcal U}((-\deposit_{i,\round} + \return_{i,\round})\cdot\npv_i(\round))$.


\subsection{The Escrow Functionality.}\label{sec:fescrow}

The functionality $\Func_{\sf escrow}^*$  (Fig.~\ref{fig:escrow}) captures inter-temporal economic choices ({\em i.e.}\ a 
party can abort or continue the protocol), and formalizes a notion of  
fairness grounded in the economic literature.
The experimental evidence about inter-temporal economic choices 
\cite{benzion1989discount,ahlbrecht1997empirical,brown2015empirical,lee2016myopic} 
is that money 
paid or received ``now'' has a greater value than the same amount of money 
received or paid  ``later''.  
At the end all parties could still {receive their money back}, 
but whoever was forced by the protocol to pay into escrow at \emph{noticeably 
	different 
	times}, or held deposits of \emph{noticeably different sizes}, would clearly {feel that is not playing a fair game.}
\footnote{\label{fn:landlord}One could argue that these deposits are comparable to security deposits, as those used in the U.S.\ for interest-bearing accounts, with the interest accrued to the depositor's benefit. That is not true for the deposits used in penalty protocols based on cryptocurrencies: once the deposits are locked, they cannot be used, and therefore no interest is accrued to the depositor.}

Deposits/refunds can appear in an arbitrary order; we only 
keep track of the round in which those are made. 
Apart from these commands, and the impossibility of 
creating money from nothing, the behavior of $\Func_{\sf 
escrow}^*$ is unspecified.

Our functionality  is meant to capture any $n$-party protocol in 
the hybrid model with a so-called {\em escrow} ideal functionality in which: (i) Each player $\party_i$ can 
deposit a certain number of coins $\deposit_{i,t}$ in the escrow (possibly multiple times and on different rounds $\round$);
(ii) At some point the functionality might pay $\party_i$ with some coins from the escrow. In concrete 
instantiations, case (ii) can happen either because $\party_i$ claims back a previous deposit, or as a refund 
corresponding to some event is triggered by another party ({\em e.g.},\ in case of aborts). 

Fix an execution of any protocol $\pi$ in the $\Func_{\sf escrow}^*$-hybrid 
model. For each message $(\mathtt{deposit},sid,ssid,\allowbreak \coins(\deposit_{i,\round}),*)$ sent by $
\party_i$ during the execution of the subsession $ssid$
 we add an entry $(\round,\deposit_{i,
	\round})$ into an array $\calD_i$ to book keep all deposits $\party_i$ made.
For the commands $(\mathtt{refund},sid,ssid,\round,\return_{i,\round},*)$ received 
by $
\party_i$, we maintain an array $\calR_i$ of entries $(\round,\return_{i,
	\round})$ 
keeping track of all claims/refunds $\party_i$ received.

\begin{figure}[!t]
	\begin{framed}\footnotesize
			\vspace*{-0.5\baselineskip}
		\noindent The \textbf{Escrow Functionality $\Func_{\sf escrow}^*$} runs with security parameter $1^\secpar$, parties $\party_1,\ldots,\party_n$, and adversary $\Sim$ corrupting a subset of parties. {Let $\deposit_{i,\round}$ be the 
coins put into escrow by player $i$ during 
round $\round$, and let $\return_{i,\round}$ be the coins that the same player 
receives  at round $\round$.}
		The behavior of $\Sim$ is unspecified, except for the following:
		\begin{itemize}
			\item Upon input $(\mathtt{deposit},sid,ssid,\coins(\deposit_{i,\round}),*)$ from an honest party $\party_i$ at round $\round$, record $(\mathtt{deposit},sid,ssid,\allowbreak i,\allowbreak \round,\deposit_{i,\round},*)$ and send it to all parties.
			\item During round $\round$ the functionality might send $(\mathtt{refund},sid,ssid,\allowbreak \round,\allowbreak\coins(\return_{i,\round}),\allowbreak*)$ to party $\party_i$ and $(\mathtt{refund},sid,ssid,\allowbreak \round,\return_{i,\round},*)$ to all parties.
			\item The functionality is not allowed to create coins, {\em i.e.} at any round $\round$ the following invariant is maintained:
			\[
			\sum\nolimits_{i \in [n], \round' \le \round} \return_{i,\round'} \leq \sum\nolimits_{i \in [n], \round' \le \round} \deposit_{i,\round'}
			\]
		\end{itemize}
		\vspace*{-\baselineskip}
	\end{framed}
	\vspace*{-1\baselineskip}
	\caption{The family of escrow functionalities.}
			\vspace*{-\baselineskip}
			\label{fig:escrow}
\end{figure}

Common ideal functionalities used in cryptographically fair 
MPC with penalties are of the escrow type. Two instances of $\Func_{\sf escrow}^*$ are commonly used for designing state-of-the-art penalty protocols:
the Claim-Or-Refund functionality $\FuncCR$ and 
the Multi-Lock functionality $\FuncML$. Both functionalities can be implemented using both the Bitcoin network and Ethereum smart contracts. {We report the detailed functionality in the appendix}.

For example, $\FuncCR$~\cite{BentovK14}, 
allows a sender $\party_i$ to \emph{conditionally} send 
coins to a receiver $\party_j$, where the condition is formalized as a circuit $
\phi_{i,j}$ with time-lock $\tau$: $\party_j$ can obtain $\party_i$'s 
deposit by providing a satisfying assignment $w$ within time $\tau$, otherwise $
\party_i$ can have his deposit refunded at time $\tau+1$.
$\FuncML$~\cite{KumaresanB14}, instead,  
allows $n$ parties to atomically agree on a timeout $
\timeout$, circuits $\phi_1,\ldots,\phi_n$, and a deposit $\deposit$. Hence, if $
\party_i$ within round $\timeout$ reveals to everyone a valid witness $\wit_i$ for $
\phi_i$, it can claim its deposit back; otherwise, at round $\timeout + 1$, the 
deposit of $\party_i$ is split among all other players.


\subsection{Financial Fairness.}
Financial fairness then says that, even in a 
run of $\pi$ with possibly corrupted parties, the net present cost of 
participation $\Cost_i$ associated to each {\em honest player} 
is the same. 
Here, we make no assumption on $\npv_i$, but one may limit fairness to specific, 
empirical, forms of $\npv_i$ ({\em e.g.},\ known to hold for poker players).
\begin{definition}[Financial fairness] \label{def:financial_fairness}
	Consider an $n$-party protocol $\pi$
	in the $\Func^*_{\sf escrow}$-hybrid model, 
	and let $(\calD_i,\calR_i)_{i\in[n]}$ be as described above.
	We say that $\pi$ is {\em financially fair} if for every possible 
	discount rate function $\npv(\round)\in[0,1]$, 
	for all transcripts resulting from an arbitrary execution of $\pi$ 
	(with 
	possibly corrupted parties), and for all $i,j\in[n]$ such that $\party_i$ and $\party_j$ 
	are honest, it holds that $\Cost_i = \Cost_j$ where the {\em net present cost of participation} $\Cost_i$ is defined in Eq.~\eqref{eq:npv}.
\end{definition}
Financial fairness may be trivial to achieve {\em in isolation}. However, the end goal is to design protocols that are both cryptographically and financially fair.
Also, observe that Def.~\ref{def:financial_fairness} could be weakened by considering  specific discount rates $\npv(\round)$ ({\em e.g.}\ financial fairness with hyperbolic discount).


\section{Penalty Protocols}\label{sec:protocols}

%


\subsection{Protocols Descriptions.}\label{subsec:protdesc}

The idea of guaranteeing cryptographic fairness through monetary compensation  was originally studied in the setting of e-cash or central bank 
systems~\cite{BelenkiyCEJKLR07,Lindell09,KupcuL12}, and implemented  using Bitcoin 
by Andrychowicz {\em et al.}~\cite{AndrychowiczDMM14}.
Other penalty protocols also exist for the concrete case of cryptographic 
lotteries~\cite{AndrychowiczDMM14,BentovK14,miller2017zero}.
A different type of penalty protocol is the one introduced in Hawk~\cite[Appendix G, \S
B]{kosba2016hawk}. This construction follows the blueprint of \MLMech\ except that
it employs a \emph{semi-trusted manager} in order to enforce a correct cash distribution.
For further discussions 
see~\cite{AzouviHM18}.

 A detailed description of the \MultiLock\ (\MLMech)~\cite{KumaresanB14} and \Ladder\ (\LMech)~\cite{BentovK14} protocols are summarized in the following. In Table \ref{table:protocols} we provided a brief description of all the protocols compared in this paper.
 
Let $\func$ be the function being 
computed, and 
$\inp_i$ be the private input of party $\party_i$. At the beginning, the players run a 
cryptographically 
{\em unfair}, off-chain, MPC protocol for a derived function $\tilde\func$ that: 
(i) computes the output $\out = \func(\inp_1,\ldots,\inp_n)$;
(ii) divides $\out$ into $n$ shares\footnote{The secret sharing scheme ensures that an attacker corrupting up to 
	$n-1$ players 
	obtains no information on the output $\out$ at the end of this phase.} $\shares_1,\ldots,\shares_n$;
(iii) computes a commitment $\com_i$ (with opening $\rndcom_i$) to each share 
$\shares_i$, and 
gives $(\shares_i,\rndcom_i)$ to the $i$-th party and $(\com_j)_{j\in[n]}$ to every 
player.

From this point, the functioning of \LMech\ and \MLMech\ differs.
In \MLMech\ all the players that are willing to complete the protocol engage in the $\FuncML$ functionality. 
During the Lock Phase, each party conditionally sends (locks) the same amount of coins to $\FuncML$. Each transaction is parametrized by the values $\com_i$ for all $i \in [n]$. Each party $\party_i$ shall then reveal the opening $\rndcom_i$ together with its share $\shares_i$ to $\FuncML$ before a fixed timeout (Redeem Phase). If not done, his deposit will be redistributed to the honest parties during the Compensation Phase.

More in details, during the Lock Phase, each party $\party_i$ sends to $\FuncML$ all the commitments $\com_i$ for all $i \in [n]$ received as an output of $\tilde\func$, together with a deposit $d$. If all the players sent the same commitment values and the same deposit amount before a fixed timeout $\timeout_1$, then the functionality moves to the Redeem Phase.

During the Redeem Phase, each party $\party_i$ shall send his share $\shares_i$ together with the opening $\rndcom_i$ before a fixed timeout $\timeout_2$. If all the parties provided a pair $(\shares_i,\rndcom_i)$ that is also a valid opening of $\com_i$ (i.e. such that $\Commit(\shares_i;\rndcom_i)=\com_i$), then all the deposits are given back to their owners. Now each party can use $(\shares_1,\ldots,\shares_n)$ to compute the  output of $\func$.

If at least one party $\party_i$ did not send  $(\shares_i,\rndcom_i)$ to $\FuncML$ before $\timeout_2$, his deposit $d$ will be used to compensate the other parties during the Compensation Phase. 

The \InsuredMPC\ (\IMPC)~\cite{BaumDD18} protocol can be see as a more efficient version of \MLMech\ achieving lower script complexity.

\Ladder\ works differently. In \MLMech\ the players engage in a sequence of ``claim-or-refund'' 
transactions divided into two phases.
During the Deposit Phase, each player conditionally sends some coins to another party via $\FuncCR$. These transactions are 
parameterized by the 
values $\com_i$, and require the receiving player to reveal the pair $(\shares_i,\rndcom_i)$ before a fixed timeout, during the Claim Phase, in order to ``claim'' the reward (thus compensating honest players),\footnote{More precisely, in  $\LMech$ the condition requires the recipient $i$ to publish the pair $(\shares_{j},\rndcom_j)$ such that $\Commit(\shares_j;\rndcom_j) = \com_j$ for each $j \leq i$.} which otherwise will be refunded to the sender who will lose the coins sent to the honest parties without being able to redeem the coins received from them ({\em i.e.},\ a penalty to the 
dishonest 
player). 
Finally, every party either reconstructs 
the output or receives a monetary 
compensation.

More in details, the Deposit Phase of Protocol $\LMech$ consists of Roof/Ladder Deposits, as illustrated below for $n=4$:
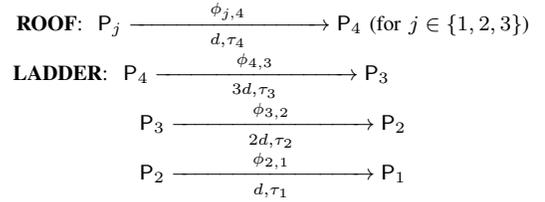
\begin{center}	\footnotesize
	\textbf{ROOF}:
	~$\party_j \xrightarrow[d,\tau_4]{\hspace*{1cm} \phi_{j,4} \hspace*{1cm}} \party_4$ (for $j \in \{1,2,3\}$) \\
	\hspace*{-1.9cm}\textbf{LADDER}:
	~$\party_4 \xrightarrow[3 \deposit,\tau_3]{\hspace*{1cm} \phi_{4,3}
		\hspace*{1cm}} \party_3$ \\
	$\party_3 \xrightarrow[2 \deposit,\tau_2]{\hspace*{1cm} \phi_{3,2} \hspace*{1cm}} \party_2$ \\
	$\party_2 \xrightarrow[\deposit,\tau_1]{\hspace*{1cm} \phi_{2,1} \hspace*{1cm}} \party_1$ 
\end{center}
where $\tiny \party_i \xrightarrow[\deposit,\tau]{\hspace*{1cm} \phi_{i,j} \hspace*{1cm}} \party_j$ indicates that $\party_i$ deposits $\deposit$ coins that can be claimed by $\party_j$ before time $\tau$, as long as $\party_j$ sends to $\FuncCR$ a valid witness $w$ for the predicate $\phi_{i,j}$. 
Importantly, the protocol requires that the claims happen in reverse order w.r.t.\ the deposits.
Assume that $\party_3$ is malicious and aborts the protocol during the Claim Phase. In such a case, $\party_1$ would claim $\deposit$ coins from $\party_2$ at round $\tau_1$, whereas $\party_2$ would claim $2\deposit$ coins from $\party_3$ at round $\tau_2$. If $\party_3$ aborts (and thus it does not provide a valid witness), $\party_4$ is refunded $3\deposit$ coins at round $\tau_4$. After that, at round $\tau_5>\tau_4$, each $\party_{i\le 3}$ is refunded $\deposit$ coins (from the roof deposits).
Thus, $\party_3$ loses $2\deposit$ coins, while each $\party_{i\le 2}$ is compensated with $\deposit$ coins. Player $\party_4$ has not moved, and thus it is not compensated.

%


\subsection{Illustration of Financial Unfairness.}
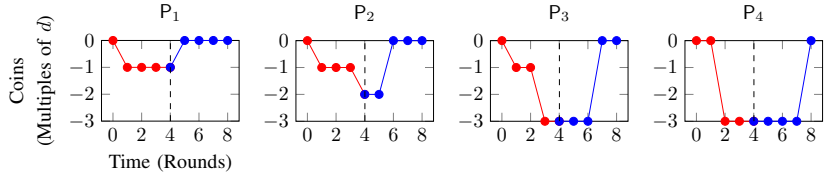
\begin{figure*}[t]
	\centering
	\resizebox*{11cm}{!}{\begin{tikzpicture}
\begin{groupplot}[
group style={group size=4 by 1,vertical sep = 1cm},
height=3cm,width=4cm, ymin = -3, 
ymax = 0]
\nextgroupplot[title=$\party_1$,xlabel=Time (\add{Rounds}),ylabel=(Multiples of $\deposit$),legend to name = curves2]
\addplot[mark = *,red] coordinates { 
	(0,0)
	(1,-1)
	(2,-1)
	(3,-1)
	(4,-1)
};
\addlegendentry{Deposits}
\addplot[mark = *,blue] coordinates { 
	(4,-1)
	(5,0)
	(6,0)
	(7,0)
	(8,0)
};
\addlegendentry{Withdrawals}
\coordinate (top1) at (axis cs:4,\pgfkeysvalueof{/pgfplots/ymax});
\coordinate (bot1) at (axis cs:4,\pgfkeysvalueof{/pgfplots/ymin});
\nextgroupplot[title=$\party_2$]
\addplot[mark = *,red] coordinates{ 
	(0,0)
	(1,-1)
	(2,-1)
	(3,-1)
	(4,-2)
};
\addplot[mark = *,blue] coordinates{
	(4,-2)
	(5,-2)
	(6,0)
	(7,0)
	(8,0)
};
\coordinate (top) at (rel axis cs:0,1);
\coordinate (top2) at (axis cs:4,\pgfkeysvalueof{/pgfplots/ymax});
\coordinate (bot2) at (axis cs:4,\pgfkeysvalueof{/pgfplots/ymin});

\nextgroupplot[title=$\party_3$]
\addplot[mark = *,red] coordinates { 
	(0,0)
	(1,-1)
	(2,-1)
	(3,-3)
	(4,-3)
};
\addplot[mark = *,blue] coordinates {
	(4,-3)
	(5,-3)
	(6,-3)
	(7,0)
	(8,0)
};

\coordinate (top3) at (axis cs:4,\pgfkeysvalueof{/pgfplots/ymax});
\coordinate (bot3) at (axis cs:4,\pgfkeysvalueof{/pgfplots/ymin});
\nextgroupplot[title=$\party_4$](b)
\addplot[mark = *,red] coordinates { 
	(0,0)
	(1,0)
	(2,-3)
	(3,-3)
	(4,-3)
};
\addplot[mark = *,blue] coordinates {
	(4,-3)
	(5,-3)
	(6,-3)
	(7,-3)
	(8,0)
};
\coordinate (bot) at (rel axis cs:1,0);
\coordinate (top4) at (axis cs:4,\pgfkeysvalueof{/pgfplots/ymax});
\coordinate (bot4) at (axis cs:4,\pgfkeysvalueof{/pgfplots/ymin});
\end{groupplot}
\path (top)--(bot) coordinate[midway] (group center);
\node[above,rotate=90] at (group center -| current bounding box.west) {Coins};
\draw [dashed] (bot1) -- (top1);
\draw [dashed] (bot2) -- (top2);
\draw [dashed] (bot3) -- (top3);
\draw [dashed] (bot4) -- (top4);
\end{tikzpicture}}
	\vspace*{-\baselineskip}
	\caption{Coins locked in a run of the 4-party Ladder Protocol during the Deposit Phase (in red) and the Claim Phase (in blue).}\label{fig:ladder4}
\end{figure*}
The amount of deposited coins for each player in the Ladder Protocol is illustrated in Fig.~\ref{fig:ladder4} (for the 4-party case).
Observe that $\party_1$ has to deposit only $\deposit$ coins, while $\party_4$ needs to deposit $3\deposit$ coins.
Furthermore, $\party_4$ has to lock its coins very early ({\em i.e.},\ at the 2nd round), but can only claim its coins very late ({\em i.e.},\ at the last round).
Hence, this protocol is financially \emph{un}fair in the following sense: 
(i) The amount of deposits are different for each player ({\em e.g.},\ $\party_1$ deposits $\deposit$ coins while $\party_n$ deposits $(n-1)\deposit$ coins);
and (ii) some players deposit early but can only claim late in the protocol ({\em e.g.},\ $\party_4$ in Fig.~\ref{fig:ladder4}).

While financial \emph{un}fairness is easy to notice (by pure observation) in simple protocols such as Ladder,
it tends to be more difficult to judge whether a penalty protocol is financially fair or not when it yields a more complicated sequence of deposit and claim transactions~\cite{kumaresan2015use,KumaresanB16,kumaresan2016improvements}.

\begin{table*}
\caption{Penalty Protocols}\label{table:protocols}
\footnotesize
\begin{tabular}{p{0.15\textwidth}p{0.8\textwidth}}
	\hline
\emph{\Ladder\ (\LMech)~\cite{BentovK14}} & Parties run an off-chain cryptographically unfair MPC in which each party learns a share $\shares_i$ of the output and all the commitments $(\com_j)_{j \in n]}$ of the other players' shares. After that, they deposit an amount $d$ of coins inside a smart contract (Deposit Phase) together with the description of a circuit, that they can reclaim later (Claim Phase) if and only if they provide the correct witness for the circuit (e.g. the opening $\shares_i$ of their commitment $\com_i$). The name \emph{Ladder} stems from the fact that parties make their deposit and reclaims in an asymmetric ladder-style fashion. 
\\
\hline
\emph{\LockedLadder\ (\LLMech)~\cite{kumaresan2015use}} & This protocol is specifically tailored for playing distributed poker. To support multiple stages, a locking mechanism is designed to penalize aborting players in each phase of computation (in \LMech\ players are penalized only during the Claim Phase).
Moreover,  \LLMech\ yields a more complicated deposit sequence, as it requires additional deposits (called Bootstrap and Chain Deposits) to force the first party to start the new stage of the poker ideal functionality. 
\\ \hline
\emph{\CompactLadder\ (\CLMech)~\cite{kumaresan2016improvements}} &
To prevent the explosion in script complexity (in \LMech\ the witness in the last round is $n$ times larger than the witness in the first round), protocol \CLMech\ uses a trick that makes the size of the witness independent of the number of players. 
The basic idea is to replace $(\shares_1,\ldots,\shares_n)$ with a secret sharing $(\key_1,\ldots,\key_n)$ of the 
secret key $\key$ for a symmetric encryption scheme, and to reveal to every party 
an encryption $\ctx$ of the output $\out$.
\\ \hline
\emph{\PlantedLadder\ (\PLMech)~\cite{kumaresan2016improvements}.}	 &
This protocol extends \LMech\ to reactive functionalities by stacking multiple instances of \LMech, {\em i.e.}\ an $n$-party $r$-stage functionality is handled as a run of $\LMech$ with $r \cdot n$ parties, using additional deposits to force the next stage to start (the so-called Underground Deposits).
As a result, \PLMech\ requires more transactions and very high deposits from each player. 
To improve efficiency, one can replace \CLMech\ with \LMech; we denote the resulting protocol as \CPLMech.
\\ \hline
\emph{\AmortizedLadder\ (\ALMech)~\cite{KumaresanB16}} & This protocol aims at performing multiple MPCs using  a single instance of a penalty protocol. The sequence of deposits/claims is the same as in \LLMech, except that all the deposits/claims happen in parallel. 
\\ \hline
\emph{\MultiLock\ (\MLMech)~\cite{KumaresanB14}} & This protocol relies on the ideal functionality $\FuncML$, instead of $\FuncCR$, in order to realize the ``claim-or-refund'' transactions. The latter allows to manage multiple deposits/claims in an atomic fashion, thus resulting
in an improved round complexity. 
\\ \hline
\emph{\InsuredMPC\ (\IMPC) \cite{BaumDD18}} & This protocol follows the same blueprint of \MLMech, {\em i.e.} \IMPC\ manages multiple deposits/claims in an atomic fashion. However, the protocol further improves the efficiency in the evaluation of the commitments in the off-chain MPC and the on-chain reconstruction phase using    publicly-verifiable  additively homomorphic commitments.
\\ \hline
\end{tabular}
\end{table*}

	

	\subsection{An Illustrative New Protocol.} \label{subsec:newprot}
	\add{To illustrate alternative design choices } 
 when looking into an efficient and financially fair penalty protocol, we provide a simple example of 
	another protocol that is provably secure and achieves the same efficiency of the \CLMech\ protocol by Kumaresan and Bentov~\cite{KumaresanB14}, but still fails to achieve UC security.
	Namely, we can design a new penalty protocol that combines ideas from~\cite{GordonIMOS10} and~\cite{KumaresanB14,KumaresanB16}, to obtain a constant-round penalty protocol with $O(n)$ transactions and script complexity $O(n\secpar)$. 
	This protocol is both cryptographically and financially fair; however, security only holds in the sense of sequential composition (although in the plain model and under standard assumptions).
	
\emph{Compact Multi-Lock (\CMLMech).}
We rely on some standard cryptographic primitives:
(i) an $n$-party secret sharing scheme $(\share,\reconstruct)$ with message space $\bin^\secpar$ and share space $\bin^k$;
(ii) a secret-key encryption scheme $(\Enc,\Dec)$ with secret keys in $\bin^\secpar$ and message space $\bin^{m}$, where $m$ is the output size of the function $\func$; and
(iii) a non-interactive commitment $\Commit$ with message space $\bin^k$.

Let $\func$ be the function to be computed.
The protocol proceeds in two phases.
In the first phase, the parties run a cryptographically unfair MPC for a derived function $\tilde\func$ that samples a random key $\key$, secret shares $\key$ into shares $\key_1,\ldots,\key_n$, commits to each share $\key_i$ individually obtaining a commitment $\com_i$, and finally encrypts the output $\func(\inp_1,\ldots,\inp_n)$ using the key $\key$ yielding a ciphertext $\ctx$. The output of $\party_i$ is $((\com_j)_{j\in[n]},\ctx,\allowbreak\key_i,\rndcom_i)$, where $\rndcom_i$ is the randomness used to generate  $\com_i$. The $\tilde\func$ can be safely run by an unfair protocol, since an adversary can not learn the final output $y$ at the end of its execution.
During the second phase, the players use $\FuncML$ to reveal the shares of the key $\key$ in a fair manner, thus ensuring that every player can reconstruct the key and decrypt the ciphertext $\ctx$ to obtain the output of the function.

\begin{note}
\emph{Analysis of Cryptographic Fairness.}\label{sec:security}
The theorem below states that our protocol also achieves cryptographic fairness in the sequential composability setting.
As we will argue in Sec. \ref{sec:compsec}, we need
\emph{non-committing encryption} to make the proof go through in the UC setting. Unfortunately, this would remove the compactness requirement on the script complexity unless unrealizable assumptions like the one of \MLMech\ is used.
It may be possible that compactness in UC might be achieved by removing the financial fairness requirement. However, we identify financial fairness as a main requirement for penalty protocols.
We show the proof sketch below and the detailed proof can be found in
\ifJOUR\else Appendix~\ref{app:multilock_security}\fi.
\begin{theorem}\label{thm:multilock_security}
	Let $\func$ be any $n$-party function, and $\tilde\func$ be as above. 
	Assume $(\share,\reconstruct)$ is an $(n,n)$-threshold secret sharing scheme, $(\Enc,\Dec)$ is semantically secure, and $\Commit$ is perfectly binding and computationally hiding.
	Then, the protocol described above
	$(n-1)$-securely computes $\Func_\func^*$ with penalties in the $(\Func_{\tilde f},\FuncML)$-hybrid model.
\end{theorem}
\begin{proof}[Proof sketch.]
	We prove the above theorem by following the Real/Ideal world paradigm. We start by describing the simulator $\Sim$ in the 
ideal world. Let $\calI$ be the set of corrupted parties
and $\calH = [n] \setminus \calI$ be the set of honest 
parties (with $h = |\calH|$).
\begin{enumerate}
	\item\label{step:simsetup} Acting as $\Func_{\tilde\func}$, wait to receive inputs $\{\inp_i\}_{i \in \calI}$ from $\advA$.
	Hence, sample $\key,\tilde\key\getsr\bin^\secpar$, let $(\key_1,\allowbreak\ldots,\key_n) \getsr \share(\tilde\key)$, $\ctx \getsr \Enc(\key,0^m)$ and $\com_i \allowbreak\getsr\allowbreak \Commit(0^k)$ for all $i \in \calH$, and send $((\com_j)_{j\in[n]},\allowbreak\ctx,\allowbreak \key_i,\rndcom_i)_{i\in\calI}$ to $\advA$.
	\item\label{step:unfair} Acting as $\FuncML$, wait to receive the lock message together with the circuits $\phi(\com_i;\cdot)_{i \in  [n]}$ from $\advA$ on behalf of the corrupted parties.
 If for some $i\in\calI$ the message was not received, or the adversarial circuits do not match, send $\coins(d)$ back to each corrupted party, abort and terminate the simulation. 
	Else, notify $\advA$ that each corrupted party correctly locked its coins.

	\item Send $\{\inp_i\}_{i \in\calI}$ together with $\coins(hq)$ to $\Func_f^*$, receiving $\out$ back. Hence, rewind the execution of $\advA$ to the beginning of the lock phase, change the distribution of
	$((\com_j)_{j\in[n]},\allowbreak\ctx,\key_i,\rndcom_i)_{i\in\calI}$ to that of the real protocol $\pi$, and repeat step~\ref{step:unfair} of the simulation, except that, in case $\advA$ now aborts, the rewinding is repeated with fresh randomness and step~\ref{step:unfair} is run again.
	\item At round $\timeout$, acting as $\FuncML$, send the send the redeem message to $\advA$ for each $i\in\calH$. Set $\ell = 0$. Hence,
	upon receiving the redeem message from $\advA$ containing $\key_i'$ (on behalf of each corrupted $\party_i$): if $\ell < n-h$, check that $\phi(\com_i;(\key'_i,\rndcom'_i))=1$ and $\key'_i = \key_i$; if the check passes send $\coins(d)$ back to $\party_i$ and set $\ell \gets \ell + 1$; If $\ell = n-h$, receive $\coins(hq)$ from $\Func^*_f$, and terminate.
	\item At round $\timeout+1$, if $\ell < n-h$, send
	$\coins(\tfrac{\deposit}{n-1}))$ to each corrupted $\party_j \ne \party_i$ on behalf of each 
	corrupted player $\party_i$ that did not redeem its witness in the previous step of 
	the simulation. Hence, abort and terminate the 
	simulation.
\end{enumerate}

To conclude the proof, we consider a sequence of hybrid experiments and show that each pair of hybrid distributions derived from the experiments are computationally close.
\begin{description}
	\item[$\hyb_0(\secpar)$:] Identical to the real world experiment, except during the fair 
	reconstruction phase, in case the attacker does not provoke an abort during the lock phase, we rewind the adversary to the 
	share distribution phase and re-run the entire protocol.
	\item[$\hyb_1(\secpar)$:] As above except during 
	the {\em first run} of the share reconstruction phase (before the rewinding) we 
	switch the distribution of the commitments to $\com_i \getsr \Commit(0^k)$ for each $i \in \calH$. 
	During the {\em second run} (after the rewind) of the share distribution phase, if any, 
	the honest commitments are reset to the original distribution $\com_i \getsr \Commit(\key_i)$.
	\item[$\hyb_2(\secpar):$] As above, except during the {\em first run} of the share reconstruction (before the rewinding) we switch the distribution of the shares to $(\key_1,\ldots,\key_n) \getsr \allowbreak \share(\tilde\key)$ for random $\tilde\key$ 
	independent from $\key$.
	\item[$\hyb_3(\secpar):$] As above, except during the {\em first run} of the share reconstruction (before the rewinding) we switch the distribution of the ciphertext to $\ctx \getsr \allowbreak \Enc(\key,0^m)$.
	\item[$\hyb_4(\secpar):$] Identical to the ideal world for the above defined simulator $\Sim$.
\end{description}
In particular we argue that $\{\hyb_0(\secpar)\}_{\secpar \in \setN} \cind \{\hyb_1(\secpar)\}_{\secpar \in \setN}$ for the hiding property of the underlying commitment scheme, 	$\{\hyb_1(\secpar)\}_{\secpar \in \setN} \cind \{\hyb_2(\secpar)\}_{\secpar \in \setN}$ by privacy of the underlying threshold secret sharing scheme, $\{\hyb_2(\secpar)\}_{\secpar \in \setN} \cind \{\hyb_3(\secpar)\}_{\secpar \in \setN}$ by semantic security of the underlying secret key encryption 
scheme and $\{\hyb_3(\secpar)\}_{\secpar \in \setN} \equiv \{\hyb_4(\secpar)\}_{\secpar \in \setN}$ by perfect binding of the commitments and perfect correctness of the encryption scheme. The theorem  follows by combining the above lemmas.
\end{proof}
\end{note}
\emph{Further considerations.} If the complexity of locally computed share is negligible, and the tasks are in abundance (an infinite number of tasks can be run in parallel) then all financially fair MPC protocols would be equally attractive as they might  be constrained by the funds (for deposits) available to each participant. However, running our illustrative protocol might be insecure and thus vulnerable to attacks violating cryptographic fairness or the privacy of the parties' inputs.

\section{Comparison over Security and Efficiency} \label{sec:compsec}


\subsection{Security Assumptions.} \label{subsec:assumptions}
Protocols \LMech, \MLMech, \LLMech, \ALMech\ and \PLMech\ satisfy UC security, and \IMPC\ satisfies (G)UC security with the global RO functionality.
The situation is different for $\CLMech$ (and $\CPLMech$ as well). Recall that the players start by engaging in an off-chain, unfair, MPC protocol whose output for party $\party_i$ includes $(\ctx,\allowbreak\key_i)$, where $\ctx$ is an encryption of the output $y$ under a symmetric key $\key$, and $\key_i$ is a share of the key.

Unfortunately, the encryption must be ``non-committing''~\cite{CanettiFGN96} for the security proof to go through: the simulator first 
must send the adversary a bogus ciphertext $\ctx$ (say an encryption of the all-zero string), 
and when it learns the correct output $\out$, if the adversary did not abort, must explain 
$\ctx$ as an 
encryption of $\out$ instead.
In the plain model, such encryption  inherently 
requires keys  as long as the 
plaintext~\cite{Nielsen02}, which would void any efficiency improvement 
w.r.t.\ the 
original \LMech\ Protocol. To circumvent this 
problem,~\cite{kumaresan2016improvements} 
builds the encryption $\ctx$ in the ROM, 
essentially setting 
$\ctx = \mathsf{Hash}(\key_1 \oplus \ldots \oplus \key_n) \oplus \out$.

A considerable drawback of the hash-based \CLMech\ Protocol is that  
it is {\em not} provably secure in the ROM because one cannot assume that $\mathsf{Hash}$ is a 
random oracle: to run an MPC protocol that computes $\ctx$ we must represent the very hashing algorithm as a circuit.

Protocol $\CMLMech$ follows the same blueprint as \MLMech, but intuitively replaces the ideal functionality $\Func^*_{\sf CR}$ with $\Func^*_{\sf ML}$ in order to improve the round complexity and achieve financial fairness. As a consequence, its security analysis faces a similar issue as the one discussed above.
Here, we propose an alternative solution that allows to obtain provable security in the plain model by focusing on standalone, rather than UC, security (which in turn implies security under sequential composition).
This weakening allows us to replace the non-committing encryption scheme with any semantically secure one, and to solve the issue of 
equivocating the ciphertext in the security proof by rewinding the adversary (which is not 
allowed in the UC setting). A similar solution was considered in~\cite{GordonIMOS10} 
for fair MPC without coins.

Rewinding in our setting essentially means that the simulator have the ability to reverse transactions on the blockchain, whereas distributed ledgers are typically immutable. However, there already exist certain blockchains where blocks can be redacted given a secret trapdoor~\cite{AtenieseMVA17} (and are immutable otherwise). In our case, such a trapdoor would not exist in the real world, but rather it would be sampled by the simulator in the security proof, in a way very similar to standard proofs of security in the common reference string model~\cite{BlumSMP91,CanettiF01}.
We further note that previous work also used limited forms of rewinding in the setting of MPC protocols with blockchains. For instance, Choudhuri {\em et al.}~\cite{ChoudhuriGJ19} construct black-box zero-knowledge protocols in a blockchain-hybrid model where the simulator is allowed to rewind only during certain slots.

\vspace{0.2\baselineskip}

\subsection{Asymptotic Efficiency.}\label{sec:eff_sim}
\begin{table}[!tp]
	\centering
	\caption{Efficiency of State-of-the-art penalty protocols}
	\label{fig:eff}
	\begin{minipage}{0.9\columnwidth} \footnotesize
		Comparing penalty protocols in terms of round complexity, number of transactions, script complexity/capability, and fairness. We denote by $n$ the number of parties, by $m$ the output size of the function being computed, by $r$ the number of stages in the reactive/multistage setting, and by $\secpar$ the security parameter.
	\end{minipage}
	\begin{tabular}{lccc}
		\hline
		\textbf{Protocol} & \textbf{\#Rounds} & \textbf{\#TXs} & \textbf{Script Complexity} 
		\\
		\hline
		\LMech &\small $O(n)$ &\small $O(n)$ & \small $O(n^2m)$ 
		\\
		\CLMech  &\small $O(n)$ &\small $O(n)$ & \small $O(n\secpar)$
		\\
		\MLMech~ &\small $O(1)$ &\small $O(n^2)$ & \small $O(n^2m)$
				\\
		\CMLMech~ &\small $O(1)$ &\small $O(n^2)$ & \small $O(n\secpar)$
		\\

		\IMPC &\small $O(1)$ &\small $O(n)$ & \small $O(nm)$
		\\ \hline
		\LLMech & \small $O(n)$ & \small $O(n^2)$ & $O(n^2mr)$ \\
		\ALMech & \small $O(n)$ & \small $O(n^2)$ & $O(n^2m\secpar)$ \\
		\PLMech & \small $O(n)$ & \small $O(n)$ & $O(n^2mr)$ \\
		\CPLMech & \small $O(n)$ & \small $O(n)$ & \small $O(nr\secpar )$ \\
		\bottomrule
	\end{tabular}
		\vspace*{-1\baselineskip}
\end{table}
In this section we compare the penalty protocols w.r.t their on-chain efficiency in both an asymptotic and an empirical way (assuming their execution is on a Bitcoin network). 
In Table \ref{fig:eff} we can notice that the script complexity of \CLMech, \CPLMech\ and \CMLMech\  does not depend on the size of the output function, but only on the number of parties $n$ and the security parameter $\secpar$, thus leading to a significant efficiency speed-up. As it can be noticed also in Table~\ref{fig:summary}, \CLMech\ and \CPLMech\ are not provably secure, whilst our \CMLMech\ is secure only under standard (rather than universal) composition. 
Because of this limitation, \IMPC\ can be considered the best protocol among the presented ones.

\section{Theoretical Analysis of Financial Fairness.}\label{sec:thfin}
\subsection{(Un)fairness of penalty protocols.}
We formally prove that the family of Ladder Protocols does not meet financial fairness as per our definition. The latter is achieved by interpreting \LMech, \LLMech, \CLMech, \PLMech\ and \ALMech\ as MPC protocols in the $\Func^*_{\sf escrow}$-hybrid model, and by carefully analyzing the sequences of deposits/refunds made/received by each participant.
\begin{theorem}\label{thm:kumaresan}
	For any $n\ge 2$, and penalty amount $q>0$, the following holds for
	the $n$-party Ladder protocol $\pi_{\sf L}$ from~\cite{BentovK14}: 
	\begin{itemize}
		\item If $\npv = 1$, the protocol is financially fair.
		\item If $\npv \ne 1$, the protocol is {\em not} financially fair.
	\end{itemize}
\end{theorem}
\begin{proof}
	Consider the hybrid ideal functionality $\Func_{\sf CR}^*$ 
	; 
	intuitively this functionality allows a sender $\party_i$ to \emph{conditionally} send 
	coins to a receiver $\party_j$, where the condition is formalized as a circuit $
	\phi_{i,j}$ with time-lock $\tau$: The receiver $\party_j$ can obtain $\party_i$'s 
	deposit by providing a satisfying assignment $w$ within time $\tau$, otherwise $
	\party_i$ can have his deposit refunded at time $\tau+1$.
	
	$\Func_{\sf CR}^*$ clearly belongs to the family $\Func_{\sf escrow}
	^*$ described in \S\ref{sec:fescrow}: the Deposit Phase 
	corresponds to the $(\mathtt{deposit},*,*)$ commands in $\Func_{\sf escrow}$, 
	whereas the Claim/Refund Phase corresponds to the $(\mathtt{refund},*,*)$ 
	commands in $\Func_{\sf escrow}$.
	
	Given $\Func_{\sf CR}^*$, protocol $\pi_{\sf L}$ proceeds 
	as follows.
	First, each player $\party_i$ (except for $\party_n$) uses $\Func_{\sf 
		CR}^*$ to make a deposit of $\coins(q)$ to $\party_n$, with predicate\footnote{We 
		do not specify the predicates, as those are immaterial for characterizing the financial 
		fairness of the protocol.} $\phi_n$.
	After all these deposits are made, each party $\party_i$ with $i = n,\ldots, 2$ in 
	sequence uses $\Func_{\sf CR}^*$ to make a deposit of $\coins((i-1)q)$ to $\party_{i-1}$ and with predicate $\phi_{i-1}$. Let us call these deposits $\Tx_{i-1}$, 
	and denote by $\Tx_{n,i}$ the initial deposits to $\party_n$.
	Finally, the deposits are claimed in reverse order: First, $\party_1$ claims $\Tx_1$, 
	then $\party_2$ claims $\Tx_2$, until $\party_n$ claims $\Tx_{n,i}$ for each $i\ne 
	n$.
	Aborts are handled as follows: If $\party_{i+1}$ does not make $\Tx_{i}$, each party 
	$\party_{j\le i}$ does not make $\Tx_{j-1}$ and waits to receive the refund from $
	\Tx_{n,j}$, whereas each $\party_{j >i}$ keeps claiming $\Tx_j$ as described above.
	
	For simplicity we consider all parties are honest. The loss 
	of party $\party_i$ is:
	\begin{small}
		\begin{align*}
		\Cost_{i\ne n} &= q\cdot\npv(1) + (i-1)\cdot q\cdot\npv(n-i+2) - i\cdot 
		q\cdot\npv(n+i). \\
		\Cost_{n} &= (n-1)\cdot q\cdot\npv(2) - (n-1)\cdot q\cdot\npv(2n). 
		\end{align*}
	\end{small}
	When $\npv = 1$, we have $\Cost_1 = \ldots = \Cost_n = 0$.
	However, since, {\em e.g.},\ $\Cost_1 \ne \Cost_2$ for any choice of $\npv < 1$, we 
	conclude that $\pi_{\sf L}$ is not financially fair whenever $\npv\ne 1$.
\end{proof}

In contrast, the $\MLMech$ protocol is financially fair.

\begin{theorem}
	For any $n \ge 2$, \MLMech~\cite{KumaresanB14} is financially fair.
\end{theorem}
\begin{proof}
	It is easy to see that $\FuncML$ belongs to the family $\Func_{\sf escrow}^*$.
	Then, financial fairness immediately follows by the fact that the loss of the $i$-th player can be computed as follows:

	$\Cost_i = (n-1)q\cdot\npv(1) - (n-1)q\cdot\npv(\round) - s\cdot 
	q\cdot\npv(\round+1)$,
	where $s \le n-1$ is the number of corrupted parties that did not redeem a valid witness in the fair reconstruction phase.
\end{proof}
\IMPC~\cite{BaumDD18} can similarly be shown to be financially fair.

\subsection{Round Robin, Small Collateral and Repeated Games.}\label{sec:colla}
A natural idea to overcome the negative results on financial fairness shown above would be to simply let parties rotate their roles in different executions, or to select the roles randomly in each execution, with the hope of achieving financial fairness on expectation.
Unfortunately, we show here that these approaches are also deemed to fail, except for a finite, very small, numbers of discount rates. 

\textbf{Round Robin.}
	The ``round robin'' approach considers a global protocol 
	which consists of multiple repetitions in a round robin fashion
	of a financially \emph{un}fair penalty-based protocol 
	(such as the Ladder protocol \LMech~\cite{BentovK14}).
	This hopes to fix the unfairness in the penalty-based protocol
	if the same set of parties are to run the protocol more than once: 
	by shifting the party index in each run, 
	{\em e.g.}\ the last party becomes the first party, the first party becomes the second party, and so on.
	This is different from the penalty protocols 
	that support a multi-stage reactive functionality (such as Locked Ladder \LLMech~\cite{kumaresan2015use}),
	as even though those protocols seem to be based on repeated instances of a non-reactive protocol,
	one cannot shift the party index without affecting security.
Unfortunately this solution doesn't work in general even for 
the simple case in which the reward is the same for all parties.
In the following theorem, we will show that to achieve fairness using the round robin approach,
	the parties must be able to obtain a specific limited number of discount rates ({\em e.g.}\ from the banks), 
	depending on the deposit schedule,
	which is not practical (as the discount rates are given by the banks, not asked by the parties).
\begin{theorem} Let an unfair protocol be identified by deposits $\deposit_{i,\round}$ for each party $i\in[n]$ which may be rotated to different parties at round robin step $\rho\in[k]$ thus determining a schedule $\deposit_{i,\round+\rho}$. There are at most $\timeout\cdot k$ specific rates $\rate$ that 
	admit a fair round robin global protocol.
\end{theorem}
\begin{proof}We observe that for the round-robin protocol to be fair we need to satisfy the following equation for all pair of parties $i$ and $j$:
	$\sum_{\rho}\Cost_{i,\rho} e^{-\rate \rho} = \sum_{\rho}\Cost_{j,\rho} e^{-\rate \rho}$

	In the simple case where each party have the same reward at the end of the protocol, it can be transformed as follows:
	\[
	\sum_\rho\sum_{(\round,\deposit_{i,\round})} \deposit_{i,\round+\rho} \cdot e^{-\rate (\rho+\round)}
	=\sum_\rho\sum_{(\round,\deposit_{j,\round})} \deposit_{i,\round+\rho} \cdot e^{-\rate (\rho+\round)}
	\]
	By setting $e^{-\rate(\rho+\round)}$ equal to $x^{\rho+\round}$ we obtain a polynomial 
	equation of degree at most $k\timeout$ with integer coefficients equal to $\deposit_{i,
		\round+\rho}-\deposit_{j,\round+\rho}$. We can repeat the procedure for all pairs and we 
	obtain at most $n-1$ independent polynomial equations. It is enough that we consider the 
	pairs $1,i$ for all $i\geq 2$. The remaining equations can be derived from those ones.
	Since the original protocol is unfair, there must be at least one polynomial where at least one of such coefficient for each value of $\rho$ is 
	not zero, hence each polynomial is not identically zero and of rank at least $(k-1)\timeout$ 
	and maximum of $k\timeout$. Hence each polynomial has at most $k\timeout$ zeroes i.e. 
	values of \rate\ that admit a fair round schedule. If the other polynomials are also not identically zero, such values \rate\ must also be zeroes for the other $n-1$ 
	polynomials.
\end{proof}
\textbf{Small Collateral and Repeated Games.}
Another approach to allow the use of financially unfair protocols
in practice might be the use a of a small collateral.
Then all parties will not worry on a small interested rate on
the collateral if they have a choice of a significant 
reward at the end of the protocol. If the game is repeated several thousands
times, e.g. in financial trading, a small collateral might quickly accrue to a significant
values and therefore such solution might only hold for 
games that are not played often.

Unfortunately,  game theoretic considerations makes such proposition (make a 
mall collateral w.r.t. stakes and rewards so that interest is negligible)
less practical than it seems. It only works if \emph{all} parties have a final large reward 
with certainty.
In  cases where a party may win a lot and other 
parties may lose everything, such as poker or financial trading, this is not longer true. It is a variant of the prisoner dilemma \cite[Chap.2]{binmore2007playing}. 

\begin{theorem}
In a game where the reward of the winning parties is significantly larger than the losing party and the collateral is negligible in comparison to the stake,
the strategy `Playing last and abort if unsatisfied' is a strictly dominating strategy.
\end{theorem}
\begin{proof} The first 
player (Alice) can decide to (1a) abort and retrieve the initial stake minus the collateral, or
cooperate and --- if the last player cooperate --- (2a) retrieve nothing for herself 
or (3a) grab the reward depending 
on the random outcome of the computation. If the last player aborts  (4a) she will retrieve the initial stake plus the collateral. In contrast
the last player (Bob) can (2a) retrieve the reward if he cooperates and the outcome is
positive or (4a) abort and retrieve the initial stakes minus the collateral if the outcome is zero for him. 
The option (3a) of cooperating --- when the outcome is negligible for him --- is dominated by the option (4a) of retrieving the initial stake minus the collateral. Hence
the Nash equilibrium is first player cooperates, last 
player aborts if he doesn't win.
\end{proof}
In a repeated games with discount rates for later moves (See Section 8.3.3 in \cite{binmore2007playing}) both players may cooperate 
if the discount rate is large enough even if the individual game 
would have a dominating strategy for defecting (i.e. in our case 
going last and abort if unsatisfied).
Unfortunately, this case is not applicable in our scenario
as it requires players to have strategies that are contingent 
on the previous behavior of the game, i.e. players needs to know how
the other players played in the previous instance of the game. Since
players might join with a new pseudonym, one cannot hold them accountable for 
repeated aborts. Therefore the repeated game collapses into a sequence of
independent games where our result holds.

In summary, \emph{participating to MPCs with a 
small collateral is essentially a waste of time as the last partner is likely to defect.} For example, Poker with MPC will require “robust" collateral. Hence, for the remainder
of the paper we will consider serious games where 
the collateral is large. 

\section{Fairness and Concrete Efficiency on Bitcoins}\label{sec:concsim}
In the following, we will analyze the penalty protocols' efficiency in a   more concrete and empirical way. 
\subsection{Setting the Scene: Bloomberg Transactions}

We consider a realistic scenarion on Bloomberg Tradebook which was also used as a benchmark in a decentralized trading system \cite{NgoMKW21}:\footnote{
		Bloomberg Tradebook is an agency broker that serves financial services providers in the US. The brokerage pool is a darkpool (quoting is restricted to the market participants) and only trades (market orders) are reported to the market. There is no central limit order book, but an electronic blind matching algorithm that utilizes a direct market access tool for individual positions. The traders then can match against each other's price-demand schedules (hence a very suitable setup for peer-to-peer trading). The number of trades within this specific market are low, around 55 per day. The number of messages is quite high, at around 6500 messages per day at an average of around 130 messages quoted per trade executed. This is quite typical. Futures markets, such as the Eurodollar, Crude Oil and agricultural commodities, have very similar limit-order (public quotes) to traded securities ratios of between 50 to 300 quotes per trade. Hence, this benchmark provides a realistic test case for our algorithm for a single security.
}
\begin{itemize}
\item The number of parties could be 55 in an average trading day; as such, a two round	protocol would last 110 rounds
\item The number of messages (protocol executions) could reach 6000 in an average trading day
\end{itemize}

We consider the 55-party realistic case with a minimum penalty chip of $q$ to be consistent with the cited papers. The unfairness phenomena are amplified with $n$ parties: {\em e.g.} a large futures trading venue would comprise up to 500 parties.

We simulate the on-chain efficiency for 2-55 parties using the practical case from the Bloomberg dark pool, just to execute one contract in the case of non-reactive functionality, or two in the case of reactive functionality.
If the protocol supports 
reactive functionalities ({\em i.e.}, \LLMech\ and \PLMech), we limit the number of stages 
to 2. 
For each protocol \LMech, \LLMech, \PLMech, \ALMech, \MLMech, we first simulate 
the sequence of deposits $\{\deposit_{i,t}\}$ (with $q$ being the base unit used for 
penalization) and withdraws $\{\return_{i,t}\}$ of each $\party_i$ from an honest 
execution of the protocol.  

For transaction fees we assume the Bitcoin network's commonly used minimum transaction fee of 546 satoshis (1 satoshi = $10^{-8}$ BTC) and a BTC costs approximately 48k USD (by May 2021). For the execution time, we use the standard assumption that 1 BTC round is one hour. For script complexity, we assume 80 bits security with input size (shares) of 128 bits and the commitment scheme is SHA-256 (pre-images of 512 bits and outputs of 256 bits).
\LMech, \MLMech, and \ALMech\ are non-reactive protocols. We assume the reactive \LLMech\ and \PLMech\ are with 2 stages. For the \PLMech\ we evaluate both the naive version and the compact version \CPLMech.

\subsection{Efficiency Analysis}\label{subsec:concr_eff}
Empirical efficiency is measured in terms of transaction fees (based on the number of transactions) in Fig.~\ref{fig:eff_tx}, execution time (based on the number of rounds) in Fig.~\ref{fig:eff_r}, and script complexity (based on input size in bits) in Fig.~\ref{fig:eff_c}. 
In our empirical analysis we do not take into account \IMPC, \CMLMech, and \CPLMech, since concrete efficiency numbers are not provided in the protocol description of \IMPC\footnote{$\MLMech$ has higher script complexity than $\IMPC$, but provides better security guarantees. 
$\IMPC$ is asymptotically better ($O(nm)$), but we do not have enough data to to compare it with \MLMech\ on concrete numbers. Moreover, its security degrades because of the RO assumption.}, while \CMLMech\ is not universally composable and \CPLMech\ has no provable security.

All of the protocols show acceptable transaction fees, even the most expensive protocol \LLMech, costs only 3277\$ for the 55 parties case (which means approximately only 60\$ for each party). However, all protocols except \MLMech\ and \ALMech\ show unacceptable execution time: \MLMech\ concludes in  the simplest (non-reactive) protocol \LMech\ takes 5 days to finish for 55 parties to execute 1 contract while the most complicated (reactive) protocol \LLMech\ requires 23 days to execute 2 contracts; the improved (reactive) protocol \PLMech\ does reduce the requirement to 14 days but it is still too slow. \MLMech\ is the protocol with the lowest script complexity in the case of non-reactive while Compact \PLMech\ is the protocol with the lowest script complexity in the case of reactive. The non-compact version of \PLMech\ yields the highest script complexity.

\begin{figure}[t]
	\centering
			\begin{tikzpicture}
		\begin{axis}[
		height=6cm,
		legend pos=north west,
		ylabel= Transaction Fee in USD,
		xlabel= Number of Parties,
		ymode=log,
		log basis y={10}
		]
		\addplot[color=black,mark=*]
		coordinates{
			(2,2)
			(3,3)
			(4,4)
			(5,5)
			(10,10)
			(15,15)
			(20,20)
			(25,26)
			(50,52)
			(55,57)
		};
		\addlegendentry{\LMech}
		\addplot[color=black, mark=square*]
		coordinates {
			(2,2)
			(3,2)
			(4,3)
			(5,3)
			(10,6)
			(15,8)
			(20,11)
			(25,14)
			(50,27)
			(55,29)
			};
		\addlegendentry{\MLMech}
		\addplot[color=black, mark=otimes*]
		coordinates {
			(2,3)
			(3,4)
			(4,6)
			(5,7)
			(10,15)
			(15,23)
			(20,31)
			(25,39)
			(50,78)
			(55,86)
			};
		\addlegendentry{\ALMech}
		\addplot[color=black, mark=triangle*]
		coordinates {
			(2,5)
			(3,11)
			(4,19)
			(5,30)
			(10,114)
			(15,250)
			(20,439)
			(25,680)
			(50,2672)
			(55,3227)
		};
		\addlegendentry{\LLMech}
		\addplot[color=black, mark=diamond*]
		coordinates {
			(2,4)
			(3,6)
			(4,8)
			(5,10)
			(10,20)
			(15,31)
			(20,41)
			(25,52)
			(50,104)
			(55,115)
		};
		\addlegendentry{\PLMech}
		\end{axis}
		\end{tikzpicture}
	\begin{minipage}{0.9\columnwidth} \footnotesize
		The most expensive \LLMech\ requires 12312 transactions, which costs around 3277\$ for the 55 parties case, approximately 60\$ per party.	\end{minipage}
	\caption{Transaction fees (based on the number of BTC transactions, where 1 transaction costs 546 satoshis, 1 BTC = 48k USD, by May 2021).}\label{fig:eff_tx}
		\vspace*{-\baselineskip}
\end{figure}

\begin{figure}[t]
	\centering
			\begin{tikzpicture}
		\begin{axis}[
		height=6cm,
		legend pos=north west,
		ylabel= Execution Time (Days),
		xlabel= Number of Parties,
		log basis y={10}
		]
		\addplot[color=black,mark=*]
		coordinates{
			(2,1)
			(3,1)
			(4,1)
			(5,1)
			(10,1)
			(15,2)
			(20,2)
			(25,3)
			(50,5)
			(55,5)
		};
		\addlegendentry{\LMech}
		\addplot[color=black, mark=square*]
		coordinates {
			(2,1)
			(3,1)
			(4,1)
			(5,1)
			(10,1)
			(15,1)
			(20,1)
			(25,1)
			(50,1)
			(55,1)
			};
		\addlegendentry{\MLMech\ and \ALMech}
		\addplot[color=black, mark=triangle*]
		coordinates {
			(2,1)
			(3,2)
			(4,2)
			(5,2)
			(10,4)
			(15,7)
			(20,9)
			(25,11)
			(50,21)
			(55,23)
		};
		\addlegendentry{\LLMech}
		\addplot[color=black, mark=diamond*]
		coordinates {
			(2,1)
			(3,1)
			(4,1)
			(5,2)
			(10,3)
			(15,4)
			(20,5)
			(25,7)
			(50,13)
			(55,14)
		};
		\addlegendentry{\PLMech} 
		\end{axis}
		\end{tikzpicture}
	\begin{minipage}{0.9\columnwidth} \footnotesize
		The simplest protocol \LMech\ takes 216 rounds, which is 5 days to finish, while the most complicated protocol \LLMech\ takes 23 days to finish. The improved protocol \PLMech\ takes 14 days to finish.
	\end{minipage}
	\caption{Execution time (based on number of rounds, where 1 round = 1 hour).}\label{fig:eff_r}
\end{figure}

\begin{figure}[t]
	\centering
			\begin{tikzpicture}
		\begin{axis}[
		height=6cm,
		legend pos=north west,
		ylabel= Input Size in Bits,
		xlabel= Number of Parties,
		ymode=log,
		log basis y={10}
		]
		\addplot[color=black,mark=*]
		coordinates{
			(2,1152)
			(3,3456)
			(4,6912)
			(5,11520)
			(10,51840)
			(15,120960)
			(20,218880)
			(25,345600)
			(50,1411200)
			(55,1710720)
		};
		\addlegendentry{\LMech}
		\addplot[color=black, mark=square*]
		coordinates {
			(2,768)
			(3,1152)
			(4,1536)
			(5,1920)
			(10,3840)
			(15,5760)
			(20,7680)
			(25,9600)
			(50,19200)
			(55,21120)
			};
		\addlegendentry{\MLMech}
		\addplot[color=black, mark=otimes*]
		coordinates {
			(2,1920)
			(3,14080)
			(4,37760)
			(5,76800)
			(10,636800)
			(15,2156800)
			(20,5116800)
			(25,9996800)
			(50,79996800)
			(55,106476800)
			};
		\addlegendentry{\ALMech}
		\addplot[color=black, mark=triangle*]
		coordinates {
			(2,2304)
			(3,16896)
			(4,45312)
			(5,92160)
			(10,764160)
			(15,2588160)
			(20,6140160)
			(25,11996160)
			(50,95996160)
			(55,127772160)
		};
		\addlegendentry{\LLMech}
		\addplot[color=black, mark=diamond*]
		coordinates {
			(2,5376)
			(3,12672)
			(4,23040)
			(5,36480)
			(10,149760)
			(15,339840)
			(20,606720)
			(25,950400)
			(50,3820800)
			(55,4625280)
		};
		\addlegendentry{\PLMech}
		\end{axis}
		\end{tikzpicture}
	\begin{minipage}{0.9\columnwidth} \footnotesize
		\MLMech\ is the protocol with the lowest script complexity in the case of non-reactive while \ALMech\ is the protocol with the lowest script complexity in the case of reactive. We do not consider \CLMech\ and \CPLMech\ since they do not satisfy the security requirements needed to be run on an existing blockchain, and \IMPC\ because of lack of concrete efficiency numbers.
	\end{minipage}
	\caption{Script complexity (input size in bits).}\label{fig:eff_c}
\end{figure}

\subsection{The Ladder Protocol for Trading}\label{sec:negative}

If the 55-party Ladder protocol was to be run on Bitcoin (requiring $\tau = 110$ rounds in total), 
a round would last approximately 60 minutes. 
 Now, assume an optimistic scenario: participants could borrow money from NYFed's SOFR to run the protocol (see \S\ref{sec:escrow_overview}, or alternatively that could be their opportunity costs). In essence they are wealthy, risk neutral and worth essentially cheap credit. Normal humans would require much higher interest rates as the empirical evidence shows \cite{benzion1989discount,ahlbrecht1997empirical,brown2015empirical}.
The discount (minute) rate would be $\rate_{\mathsf{m}} = 0.0005$ for each player.
To represent the \emph{net present cost of participation} of each participant as basis points we simply set  $\deposit = 10\,000USD$. If one is to use the Ladder protocol in a real world use case ({\em e.g.},\ dark pool futures trading),
the base deposit $\deposit$ would be necessarily at least the notional value of one futures contract
({\em i.e.},\ 1 million dollars in case of Eurodollar futures).

Thus, using Eq.~\eqref{eq:npv}, we have (in bps):
\begin{small}
	\allowdisplaybreaks
	\begin{align*}
	\Cost_1(110):=  (\deposit\cdot e^{-\drate_{\mathsf{m}\cdot60}} - \deposit\cdot e^{-\drate_{\mathsf{m}\cdot5\cdot60}})  \approx 0.11 (bps) \\
\Cost_{55}(110):= -(-54\deposit\cdot e^{\add{-}\drate_{\mathsf{m}\cdot2\cdot60}} + 54\deposit\cdot e^{\add{-}\drate_{\mathsf{m}\cdot110\cdot60}})  \approx 49
(bps)	
\end{align*}
\end{small}
In financial terms this is a disaster:
\begin{itemize}
	\item $\party_{55}$ would lose 0.49\%  of the deposit in terms of opportunity costs, which is almost \$5K, just to participate to trade a single contract!
\item Combining the numbers of messages (6000 on average), the last party $\party_{55}$ would spend \$30M just to participate to an average trading day, whereas the first party's cost would only be around \$60K (\$10 per contract, as the cost of $\party_1$ is only 0.1 bps).
\end{itemize}




\subsection{Comparative Analysis of Financial Fairness}\label{sec:comparisons}

We experimentally analyze how the different penalty protocols behave in terms of their 
inter-temporal choices.\footnote{While it seems our results could be derived with pencil and paper, 
the simulation shows non-obvious phenomena as even for a small number of parties (55) the numbers of rounds can be very large (300).}\footnote{\add{We assume the protocols are run in the same notational system for the same final outcome as this is the only way to make a sound design decision. So MPC markets could definietly adjust values 
of $q$s as $q$s are in currency units such as Euro or Dollars. But the same phenomenon is true of any multi-agent trading systems and is not the subject of the present paper and could be applied to any temporal investment strategy.}}
\begin{itemize}
\item \emph{The Multi-Lock Protocol.} \MLMech\ is 
 straightforward in this respect. Every party
	deposits the same amount of coins at the same time, and can withdraw it as 
	soon as s/he has revealed the secret. The same holds for \CMLMech.
	
\item \emph{The Ladder Protocols}. The \LMech, 
	\LLMech, \PLMech, and 
	\ALMech\ protocols have inter-temporal 
	payment schedules which clearly differ in both amount and duration of deposits per 
	party.\footnote{\IMPC\ follows the same deposit/withdrawal blueprint as \MLMech.} To show the difference we implemented a script that simulates the penalty 
	protocol transaction schedule.
\end{itemize}

\begin{figure}[t]
\begin{center}
	\resizebox*{6cm}{!}{
%
%
\begin{tikzpicture}

\begin{axis}[%
width=2.521in,
height=1.566in,
at={(0.758in,0.681in)},
scale only axis,
bar shift auto,
xmin=0.509090909090909,
xmax=5.49090909090909,
xtick={1,2,3,4,5},
xticklabels={{\MLMech},{\LMech},{\ALMech},{\LLMech},{\PLMech}},
xlabel style={font=\color{white!30!black}},
xlabel={Penalty Protocols (55 parties, 2 stages)},
ymin=0,
ymax=350,
ylabel style={align=center,font=\color{white!10!black}},
ylabel={Total Deposit Amount per Party \\ (x times q)},
axis background/.style={fill=white},
legend style={legend cell align=left, align=left, draw=white!30!black, at={(0.95,1.2)}}
]
\addplot[ybar, bar width=0.145, fill=black, draw=black, area legend] table[row sep=crcr] {%
1	54\\
2	1\\
3	55\\
4	110\\
5	168\\
};
\addplot[forget plot, color=white!30!black] table[row sep=crcr] {%
0.509090909090909	0\\
5.49090909090909	0\\
};
\addlegendentry{P1 (min protocol entry requirement)}

%

\addplot[ybar, bar width=0.145, fill=red, draw=black, area legend] table[row sep=crcr] {%
1	54\\
2	54\\
3	108\\
4	216\\
5	327\\
};
\addplot[forget plot, color=white!30!black] table[row sep=crcr] {%
0.509090909090909	0\\
5.49090909090909	0\\
};
\addlegendentry{P55 (max protocol entry requirement)}

\end{axis}
\end{tikzpicture}%
	}
	\end{center}
	\begin{minipage}{\columnwidth}\footnotesize
		This figure shows the total amount of deposit locked by the penalty protocols before the
		withdrawal phase.
		CRYPTO'14 \LMech, CCS'14 \MLMech, and CCS'16 \ALMech\ are non-reactive penalty protocols, while 
		CCS'15 \LLMech\ and CCS'16 \PLMech\ are reactive ones. 
		Among the non-reactive protocols, \LMech\ 
		requires significantly different amounts from the first party (always $q$) and the
		last party which shells 54 times as much and should shell $(n-1)q$ for $n$ parties. 
		The reactive \LLMech\ and 
		\PLMech\ require disproportionate deposits. For each stage of computation, 
		maximum 54 out of 55 parties may be compensated if one party cheats. Therefore one 
		would only expect $118q$ for two stages of computation (to compensate the other 54 honest parties in 2 stages) rather than 216 or 320+.
	\end{minipage}
	\caption{Total Amount of Deposit per Party and Protocol}\label{fig:tda}	
\end{figure}


\begin{figure}[t]
\begin{center}
	\resizebox*{8cm}{!}{
%
%
\begin{tikzpicture}

\begin{axis}[%
width=2.521in,
height=1.566in,
at={(0.758in,0.481in)},
scale only axis,
bar shift auto,
xmin=0.509090909090909,
xmax=5.49090909090909,
xtick={1,2,3,4,5},
xticklabels={{\MLMech},{\LMech},{\ALMech},{\LLMech},{\PLMech}},
xlabel style={font=\color{white!15!black}},
xlabel={Penalty Protocols (55 parties, 2 stages)},
ymin=0,
ymax=548,
ylabel style={font=\color{white!15!black}},
ylabel={Lock Time Window (rounds)},
axis background/.style={fill=white},
legend style={legend cell align=left, align=left, draw=white!15!black, at={(0.25,0.95)}}
]
\addplot[ybar, bar width=0.145, fill=black, draw=black, area legend] table[row sep=crcr] {%
1	1\\
2	55\\
3	1\\
4	437\\
5	275\\
};
\addplot[forget plot, color=white!15!black] table[row sep=crcr] {%
0.509090909090909	0\\
5.49090909090909	0\\
};
\addlegendentry{P1}

\addplot[ybar, bar width=0.145, fill=green, draw=black, area legend] table[row sep=crcr] {%
1	1\\
2	64\\
3	1\\
4	455\\
5	284\\
};
\addplot[forget plot, color=white!15!black] table[row sep=crcr] {%
0.509090909090909	0\\
5.49090909090909	0\\
};
\addlegendentry{P10}

\addplot[ybar, bar width=0.145, fill=blue, draw=black, area legend] table[row sep=crcr] {%
1	1\\
2	79\\
3	1\\
4	485\\
5	299\\
};
\addplot[forget plot, color=white!15!black] table[row sep=crcr] {%
0.509090909090909	0\\
5.49090909090909	0\\
};
\addlegendentry{P25}

\addplot[ybar, bar width=0.145, fill=red, draw=black, area legend] table[row sep=crcr] {%
1	1\\
2	108\\
3	1\\
4	543\\
5	328\\
};
\addplot[forget plot, color=white!15!black] table[row sep=crcr] {%
0.509090909090909	0\\
5.49090909090909	0\\
};
\addlegendentry{P55}

\draw [<-] (axis cs:1,1)-- +(10pt,10pt) node[align=center, above] {All 1 \\Round};
\draw [<-] (axis cs:3,1)-- +(10pt,10pt) node[align=center, above] {All 1 \\Round};

\end{axis}
\end{tikzpicture}%
	}
	\begin{minipage}{\columnwidth}\footnotesize
		Among the non-reactive protocols, CCS'14 \MLMech\ and CCS'16 \ALMech\ both conclude in one round. Moreover, \ALMech\ allows multiple MPCs to be done in the one round period.
		The reactive CCS'15 \LLMech\ and CCS'16 \PLMech\ again require a disproportionate lock time window compared to the non-reactive \LMech, \MLMech, \ALMech.
	\end{minipage}
	\vspace*{-\baselineskip}
	\caption{Maximum lock time window (55 parties, 2 stages)}\label{fig:mltw}
	\end{center}
		\vspace*{-\baselineskip}
\end{figure}

We show only the results for $\party_1$, $\party_{10}$, $\party_{25}$, and $\party_{55}$.\footnote{$\party_1$ and $\party_{55}$ are the first and the last party, which illustrates the maximum difference possible. $\party_{10}$ and $\party_{25}$ are representative of the intermediate parties.}

Fig.~\ref{fig:tda} reports the total deposited amount of  $\party_1$ (the minimum entry requirement of a protocol) 
and $\party_{55}$ (the maximum entry requirement of a protocol).
All protocols except \MLMech\ require a different 
amount of deposit from each party. In terms of total amount, \LMech\ is the 
best protocol... for the first party! The last party has to deposit more than 54x times more. \MLMech\ requires a fixed amount of $(n-1)q$ from each party, 
while \LMech\ requires such amount from only the two last parties $\party_{54}$ and $
\party_{55}$. 
\LLMech, \PLMech\, and \ALMech\ require 
very high 
amounts of deposits (and again largely different): the worst case party $\party_{55}$\footnote{We refer as the “worst case party" the party who has to deposit the most, while as “best case party" the party who has to deposit the least.} has to 
deposit $216q$ in \LLMech\ and $327q$ in \PLMech. Even taking into account the fact 
that \LLMech\ and \PLMech\ consist of 
two stages, such deposits look excessive: one would expect 
to deposit
$118q$ for two stages of computation, since we only need to compensate at most 54
parties per stage when one of the 55 parties aborts.

Fig.~\ref{fig:mltw} shows the maximum time window that a party has to keep his money 
in deposit (starting from the first deposit to the last withdrawal).
\MLMech\ and \ALMech\ have the smallest and 
fairest lock time: only one round. \LMech\ must keep the deposits in more rounds 
and not the same number of rounds: 55 rounds for the best case $\party_1$ and 
108 
rounds for the worst case $\party_{55}$. Again both \LLMech\ and \PLMech\ require 
very high lock windows for the deposits: 543 rounds for \LLMech\ and 328 rounds for 
\PLMech. 

\section{Playing It for Real: Optimistically Unfair.}\label{sec:optunfair}
The observations above refer to the worst 
case but in practice the inter-temporal differences for honest parties might not be 
noticeable. Checking the behavior of a protocol for honest 
parties,  dubbed \emph{Optimistic Computation}~\cite{KumaresanB14},
is important as a protocol can still be fair for all practical purposes. 

\begin{figure}[!t]
	\begin{center}
		\resizebox*{9cm}{!}{
%
%
\begin{tikzpicture}

\begin{axis}[%
width=4.521in,
height=3.566in,
at={(0.758in,0.481in)},
scale only axis,
bar shift auto,
xmin=0.509090909090909,
xmax=5.49090909090909,
xtick={1,2,3,4,5},
xticklabels={{\MLMech},{\LMech},{\ALMech},{\LLMech},{\PLMech}},
xlabel style={font=\color{white!15!black}},
xlabel={Penalty Protocols (55 parties, 2 stages)},
ymin=0,
ymax=2500,
ylabel style={font=\color{white!15!black}},
ylabel={Net present cost of participation (bps)},
axis background/.style={fill=white},
legend style={legend cell align=left, align=left, draw=white!15!black, at={(0.15,0.95)}}
]
\addplot[ybar, bar width=0.145, fill=black, draw=black, area legend] table[row sep=crcr] {%
1	1.4499425203951\\
2	1.47690037588966\\
3	2.95358661574596\\
4	415.227976524088\\
5	830.322117561622\\
};
\addplot[forget plot, color=white!15!black] table[row sep=crcr] {%
0.509090909090909	0\\
5.49090909090909	0\\
};
\addlegendentry{P1}

\addplot[ybar, bar width=0.145, fill=green, draw=black, area legend] table[row sep=crcr] {%
1	1.4499425203951\\
2	6.06905989419815\\
3	3.43690078921099\\
4	626.981141130827\\
5	1003.66691696721\\
};
\addplot[forget plot, color=white!15!black] table[row sep=crcr] {%
0.509090909090909	0\\
5.49090909090909	0\\
};
\addlegendentry{P10}

\addplot[ybar, bar width=0.145, fill=blue, draw=black, area legend] table[row sep=crcr] {%
1	1.4499425203951\\
2	33.0580682915382\\
3	4.24242441170009\\
4	1021.1873718908\\
5	1350.59828130125\\
};
\addplot[forget plot, color=white!15!black] table[row sep=crcr] {%
0.509090909090909	0\\
5.49090909090909	0\\
};
\addlegendentry{P25}

\addplot[ybar, bar width=0.145, fill=red, draw=black, area legend] table[row sep=crcr] {%
1	1.4499425203951\\
2	156.616709770674\\
3	5.79977008158039\\
4	1969.02151928903\\
5	2235.53523360152\\
};
\addplot[forget plot, color=white!15!black] table[row sep=crcr] {%
0.509090909090909	0\\
5.49090909090909	0\\
};
\addlegendentry{P55}

\draw [<-] (axis cs:1,1)-- +(10pt,10pt) node[above] {All 1.45 bps};
\draw [<-] (axis cs:2,200)-- +(0pt,10pt) node[align=center, above] {1.45, 6, 33, 157};
\draw [<-] (axis cs:3,1)-- +(10pt,10pt) node[align=center, above] {Different but all\\ less than 6 bps};

\end{axis}
\end{tikzpicture}%
		}
		\begin{minipage}{\columnwidth}\footnotesize
			Each bar reports the cost of participation for a 
			party in basis point for an optimistic computation (only honest parties and no 
			aborts). 
			CCS'14 \MLMech\ is the fairest protocol with the smallest (and statistically identical) 
			distribution per party. CCS'16 \ALMech\ provides limited financial fairness even though deposits lock 
			in one round (Fig.~\ref{fig:mltw}), because each party deposits a different 
			amount (Fig.~\ref{fig:tda}). 
			CRYPTO'14 \LMech\ is the least fair. 
		\end{minipage}
		\vspace*{-\baselineskip}
	\end{center}
	\caption{Net present costs of participation}\label{fig:ocn}	
	\vspace*{-1.3\baselineskip}
\end{figure}
To check whether this is the case (it is not), we analyze the financial fairness of 
each protocol by simulating the \emph{net present cost of participation} $\Cost_i$ of 
party $\party_i$ (see Eq.~\eqref{eq:npv}) in a large sequence of random executions 
with honest parties:
\begin{itemize}
	\item Use the sequence of deposits $\{\deposit_{i,\round}\}$ and withdraws 
	$\{\return_{i,\round}\}$ of party $i$ obtained from an honest execution of a protocol 
	at each round $\round$; and
	the minute ratio derived from the Secured Overnight Financing 
	rate of the New York Fed (238 bps, Tab.~\ref{tab:overnight}) as this is the going rate 
	among commercial executions, and thus is actually available;
	\item Simulate each protocol execution on Bitcoin. To convert ``rounds'' 
	to ``Bitcoin time'', we use Bitcoin actual network data, {\em i.e.}\ 
	the mean and standard deviation of the block generation 
	time (in minutes) for each day from Dec 29, 2018 to June 26, 2019;
	and consider a round of a protocol to be 6 blocks of the Bitcoin 
	blockchain 
	for a transaction to be confirmed). From the data, a 
	round 
	can take from 47 minutes to 75 minutes.
	\item Compute the net present cost of participation $\Cost_i$ of each $\party_i$ for each 
	of the 180 days using $q = 10000$ (hence $\Cost_i$ is captured as basis points), and plot them in Fig.~\ref{fig:ocn}.\footnote{As the discount rate is small, {\em i.e.} $\drate_{\mathsf{m}} = 0.0005$, the difference due to slight changes (30 minutes) in transaction confirmation times is negligible. Only for a very large number of transactions it becomes significant.}
\end{itemize}

For all cases (both reactive and non-reactive), financial fairness is only achieved in \MLMech, as every party locks and releases the same  deposit at the same time. \MLMech\ is also the best protocol in terms of net present cost of participation.

For non-reactive cases, \LMech\ yields a huge difference in losses between different parties: $\Cost_1$ is around 1.53 bps while $\Cost_4$ is around 162.75 bps. This difference is due to the disparity in both the amount of deposits and the time windows in which the deposits are locked ($\party_1$ deposits only $q$, locked for 55 rounds, while $\party_{55}$ deposits $54q$, locked for 108 rounds). The difference is slighter in \ALMech: all parties' deposits are locked for one round, but differences between amounts of deposits still exists ($110q$ for $\party_1$ and $216q$ for $\party_{55}$).

For reactive cases, in \LLMech\ and \PLMech\, the party $\party_1$ and the party $\party_{55}$ have a large difference in net present costs of participation. Furthermore, the costs for the last party $\party_{55}$ are unacceptable in both protocols: more than 2000 bps in \LLMech, and more than 2300 bps in \PLMech. However, a surprising finding is that \LLMech\ is better than its ``improved'' version \PLMech\ in terms of financial fairness. To explain this phenomenon, let us observe that even though \LLMech\ locks the deposits for a longer time (\LLMech\ concludes in 543 rounds, while \PLMech\ needs 328 rounds), the deposit amount is much less ($\party_{55}$ deposits $216q$ in \LLMech\ but $327q$ in \PLMech).

\section{Conclusions and Open Problems}\label{sec:conclusions}

\subsection{Lesson Learnt}
The main motivation of this work comes from the observation that most 
penalty protocols for cryptographically-fair MPC with penalties might be {\em unfair} when it 
comes to the amount of money each player has to put into escrow in a run of the 
protocol. Hence, the goal is designing penalty protocols that are both cryptographically and financially fair, while at the same time having good efficiency in terms of round complexity, number of transactions, and script complexity.

State-of-the art protocols either achieve low script complexity (heuristically) but not financial fairness~\cite{kumaresan2016improvements}, or achieve financial fairness but either have high script complexity~\cite{KumaresanB14} or require trusted third parties~\cite{kosba2016hawk}.
Alternatively, we also showed that efficiency, cryptographic fairness, and financial fairness are all achievable at the same time and under standard assumptions, so long as one settles for sequential (rather than universal) composability using rewinding-based proofs. The latter might indeed be an option under certain restrictions~\cite{ChoudhuriGJ19}, or using redactable blockchains~\cite{AtenieseMVA17}.

\subsection{Open Problems}\label{subsec:openproblems}
\emph{Beyond the Pro-Rata Compensation.}
We make no assumption on the function used to 
compute the net present value, but in some settings we may want to consider financial fairness only w.r.t.\ specific discount rates. 
An extension would be to drop the assumptions, 
present in the entire literature so far (including our paper) that: (i) all parties 
are compensated equally; 
and (ii) the adversary compensates all honest parties 
who do not receive the result of the computation.\footnote{These two assumptions are apparent, respectively, in the compensation step of $\FuncML$ (where $(\payout,i,j,\coins(\frac{\deposit}{n-1}))$ is sent to every $\party_j \neq 
\party_i$), and in the third step of $\Func_\func^*$ (the $\mathtt{extrapay}$ step).} This approach, known as the ``pro-rata'' approach 
for the restitution of mingled funds~\cite{Burrows-2011-restitution}, however, is not the only possible one. For example,
one could use Clayton's rule where ``first withdrawals from an account
are deemed to be made out of first payments in''~\cite{Chamorro-Courtland-2014-commingled}, and return the funds only
to the first $k$ parties who deposited. Adjusting to these rules
requires simple modifications to the functionalities and the corresponding protocols.

Taking into account that the valuation of coins can change over time is a subject of a paper by itself and will not change the present result of purely deterministic deposits in which the possibility of aborts is studied from a economic game theoretic perspective. At any given time one decide whether to keep participating or not in the protocol based on a model of the rationality of other participants. 
Further complications into the model can be imported from both socio-economical (e.g. mutual trust within the participants, their risk appetites, the popularity of specific MPC protocol, market anticipations for the price of coins/tokens used by each MPC protocol, how early in time the protocol was introduced to the public, etc.) and technical (e.g. validity of assumptions, requirements to the software/hardware running MPC protocol, usability etc.) domains. 
We consider those possible future work.

\emph{Better Efficiency with Reasonable Assumptions.}
The $\CMLMech$ Protocol only achieves standalone security.
We leave it as an open problem to construct a penalty protocol that achieves UC security, with the same efficiency as $\CMLMech$ and while retaining provable security in the plain model.
This question is wide open even in the Random Oracle Model.

\emph{Compiling unfair protocols into fair ones.}
A natural open question is whether a financially unfair protocol like \LMech\ can be compiled into a fair protocol while keeping comparable efficiency and security.

\section{Acknowledgements}

This work was partly supported by the European Union under the H2020 Programme Grant No. 830929 CyberSec4Europe (\url{http:\\cybersec4europe.eu}), the Dutch Government under the Sectorplan Grant at the VU  (H/295212.501) and the University of Rome La Sapienza under the SPECTRA Grant.

\subsection{CRediT author statement}
\emph{Conceptualization} MF; \emph{Methodology} FD (security protocols), MF (fairness, game theory), NCN (security protocols), VD (security proof); \emph{Software Programming} FD, NCN; 
\emph{Formal analysis} FD (security), MF (fairness, game Theory), NCN (security protocols), VD (security  proof);
\emph{Investigation} NCN; \emph{Data Curation}  NCN (Blockchain generation time); \emph{Writing - Original Draft} FD, MF, NCN, VD;
\emph{Writing - Review \& Editing} FD, MF, NCN, VD; 
\emph{Visualization} FD, NCN; \emph{Supervision} MF, DV;
\emph{Project administration} MF; \emph{Funding acquisition} MF, VD.




\bibliographystyle{IEEEtran}
\bibliography{short-names,FinancialFairness}

\ifIEEE

%
%

\appendices
\else
\appendix
\fi
\textbf{Conversion of interest rate.}

\vspace*{-0.5\baselineskip}
\begin{framed}\footnotesize
	\vspace*{-0.5\baselineskip}
	\noindent \textbf{Conversion of Interest Rate}
	
	\noindent To convert the Overnight Rate $\rate_{\mathsf{d}}$ (per annum) to Hourly Rate $\rate_{\mathsf{h}}$ and 
	Minute Rate $\rate_{\mathsf{m}}$, and using those to compute the corresponding payment interest, one needs to follow several intermediate steps:
	\begin{enumerate}
		\item Convert the Overnight Rate $\rate_{\mathsf{d}}$ into \emph{continuous} time, 
		{\em i.e.} $\drate = \ln(1 + \rate_{\mathsf{d}})$ (where $\ln(\cdot)$ is the natural logarithm); 
		using the NYFed's Secured Overnight Financing Overnight Rate $\rate_{\mathsf{d}} = 238 = 2.38\% = 0.0238$,
		$\drate := \ln(1 + 0.0238) = 0.0235$;
		\item Multiply the continuous rate by $\frac{1}{365\cdot 24}$ for the Hourly Rate $\rate_{\mathsf{h}}$ or $\frac{1}{365\cdot 24 \cdot 60}$ for the Minute Rate $\rate_{\mathsf{m}}$;
		then convert back to \emph{discrete} time by taking $e^{\round\drate}$ to obtain the payment interest factor, 
		\item Hence $\npv_i(\round) := e^{-\round\rate_{\mathsf{h}}}$ if we are using the Hourly Rate,
or $\npv_i(\round) := e^{-\round\rate_{\mathsf{m}}}$ if we are using the Minute Rate.
	\end{enumerate}
	\vspace*{-\baselineskip}	
\end{framed}
\vspace*{-0.5\baselineskip}

\textbf{Instances of Escrow.}\label{sec:building_blocks}
\vspace*{-0.5\baselineskip}
\begin{framed}\footnotesize
\vspace*{-0.5\baselineskip}
	\noindent The \textbf{Claim-or-Refund Functionality $\Func_{\sf CR}^*$}  runs with security parameter $1^\secpar$, parties $
	\party_1, \ldots, \party_n$, and ideal adversary $\Sim$.
	\begin{description}
		\item[Deposit Phase:] Upon receiving the tuple $(\mathtt{deposit},sid,ssid,
		\allowbreak i,\allowbreak j,\allowbreak \phi_{i,j},\allowbreak \timeout, \allowbreak\coins(\deposit))$ from $\party_i$, record the message $(\mathtt{deposit},sid,ssid, \allowbreak i,\allowbreak j,\allowbreak \phi_{i,j},\allowbreak \timeout, \deposit)$ and send it to all parties. Ignore any future 
		deposit messages from $\party_i$ to $\party_j$.
		\item[Claim Phase:] After round $\timeout$, upon receiving $(\mathtt{claim},sid,ssid,\allowbreak i, j, \phi_{i,j}, \timeout, \deposit, \wit)$ from
		$\party_j$, check if: (1) a tuple $(\mathtt{deposit},sid,ssid, i, j, \phi_{i,j}, \timeout, 
		\deposit)$ was recorded, and (2) if $\phi_{i,j}(\wit) = 1$. If both checks pass, send $
		(\mathtt{claim},sid,ssid, \allowbreak i, j, \phi_{i,j}, \timeout, \deposit, \wit)$ to all parties, send $
		(\mathtt{claim}, sid,ssid,\allowbreak i, j, \phi_{i,j}, \timeout, \coins(\deposit))$ to $\party_j$, and delete 
		the record $(\mathtt{deposit},sid,ssid,\allowbreak i, j, \phi_{i,j}, \timeout, \deposit)$.
		\item[Refund Phase:] In round $\timeout+1$, if the record $
		(\mathtt{deposit},sid,ssid,\allowbreak i, j, \phi_{i,j}, \timeout, \deposit)$
		was not deleted, then send $(\mathtt{refund},sid,ssid,\allowbreak i, j,  
		\phi_{i,j}, \timeout, \coins(\deposit))$ to $\party_i$, and
		delete the record $(\mathtt{deposit},sid,ssid,\allowbreak i, j, \phi_{i,j}, \timeout, \deposit)$.
	\end{description}
	\vspace*{-\baselineskip}
\end{framed}

\vspace*{-1\baselineskip}

\begin{framed}\footnotesize
	\vspace*{-0.5\baselineskip}
	\noindent The \textbf{Multi-Lock Functionality $\FuncML$} runs with security parameter $1^\lambda$, parties $\party_1,\ldots,\party_n$, and adversary $\Sim$.
	\begin{description}
		\item[Lock Phase:] Wait to receive $(\lock,i,D_i = (\deposit,sid,ssid,
		\allowbreak\phi_1,\allowbreak\ldots,\allowbreak\phi_n,\allowbreak\timeout),\allowbreak\coins(\deposit))$ 
		from each $\party_i$ and record $(\locked,sid,ssid,\allowbreak i,D_i)$. Then, if $\forall i,j: D_i 
		= D_j$ send message $(\locked,sid,ssid)$ to all parties and proceed to the Redeem Phase. 
		Otherwise, for all $i$, if the message $(\locked,sid,ssid,i,D_i)$ was recorded, then delete it, 
		send message $(\abort,sid,ssid,i,\coins(\deposit))$ to $\party_i$ and terminate.
		\item[Redeem Phase:] In round $\timeout$, upon receiving a message $
		(\redeem,sid,ssid,i,\wit_i)$ from $\party_i$, if $\phi(\wit_i) = 1$ then delete $(\locked,sid,ssid,i,D_i)$, 
		send $(\redeem,sid,ssid,\coins(\deposit))$ to $\party_i$ and $(\redeem,sid,ssid,i,\wit_i)$ to all 
		parties.
		\item[Compensation Phase:] In round $\timeout+1$, for all $i \in [n]$, if $
		(\locked,sid,ssid,i,D_i)$ was recorded but not yet deleted, then delete it and send the 
		message $(\payout,sid,ssid,i,j,\coins(\frac{\deposit}{n-1}))$ to every party $\party_j \neq 
		\party_i$.
	\end{description}
	\vspace*{-\baselineskip}
\end{framed}
\clearpage

\ifCONF
{\Huge\centering Supplementary Material}

\clearpage

\section{Cryptographic Primitives}\label{sec:primitives}

\subsection{Notation}
Throughout the paper, we denote the security parameter by $\secpar\in\NN$. A 
function $\nu(\secpar)$ is negligible in $\secpar$ (or just negligible) if it decreases 
faster than the inverse of every polynomial in $\secpar$, {\em i.e.}\ $\nu(\secpar)\in O(1/
p(\secpar))$ for every positive polynomial $p(\cdot)$.

Given an integer $n$, we let $[n] = \{1,\ldots,n\}$. 
If $x$ is a string, $\abs{x}$ denotes its length; if $\cX$ is a set, $\abs{\cX}$ is 
the number of elements in $\cX$. When $x$ is chosen randomly in $\cX$, we write 
$x\getsr\cX$. When $\alg$ is an algorithm, we write $y \getsr \alg(x)$ to denote a 
run of $\alg$ on input $x$ and output $y$; if $\alg$ is randomized, then $y$ is a 
random variable and $\alg(x;\omega)$ denotes a run of $\alg$ on input $x$ and 
random coins $\omega\in\bin^*$. An algorithm is in probabilistic poly time 
(PPT) if it is randomized, and its number of steps is polynomial in its inputs and $1^\secpar$ (in unary).

If $\func:(\bin^*)^n \rightarrow (\bin^*)^n$ is a function, then $\func_i(\inp_1,\allowbreak\ldots,\allowbreak \inp_n)$ is the $i$-th element of $\func(\inp_1,\ldots,\inp_n)$ for $i\in [n]$, and $(\inp_1,\ldots,\inp_n) \mapsto (\out_1,\ldots,\out_n)$ is its input-output behavior.

A probability ensemble $\rv{X} = \{\rv{X}(\secpar)\}_{\secpar\in \NN}$ is an infinite 
sequence of random variables indexed by security parameter $\secpar \in \NN$. Two 
distribution ensembles $\rv{X} = \{\rv{X}(\secpar)\}_{\secpar\in \NN}$ and $\rv{Y} = \{\rv{Y}(\secpar)\}_{\secpar\in \NN}$ are said to be \emph{computationally 
	indistinguishable}, denoted $\rv{X} \cind \rv{Y}$ if for every non-uniform PPT 
algorithm $\advD$ there exists a negligible function $\nu(\cdot)$ such that
$
\left\lvert\Prob{\advD(\rv{X}(\secpar)) = 1} - \Prob{\advD(\rv{Y}(\secpar)) = 1}\right\rvert \leq \nu(\secpar).
$

\subsection{Secret Sharing Schemes} 
An $n$-party secret sharing scheme $(\share,\reconstruct)$ is a pair of poly-time algorithms specified as follows.
(i) The randomized algorithm $\share$ takes as input a message $\msg\in\calM$ and outputs $n$ shares $\shares = (\shares_1,\ldots,\shares_n)\in\calS_1\times\cdots\times\calS_n$;
(ii) The deterministic algorithm $\reconstruct$ takes as input a subset of the shares, say $\shares_\calI$ with $\calI \subseteq [n]$, and outputs a value in $\calM \cup \{\bot\}$. 
\begin{definition}[Threshold secret sharing]
Let $n\in\NN$. For any $t \le n$, we say that $(\share,\reconstruct)$ is an $(t,n)$-threshold secret sharing scheme if it satisfies the following properties.
	\begin{itemize}
		\item \textbf{Correctness}: For any message $\msg\in\calM$, and for any $\calI \subseteq [n]$ such that $|\calI| \geq t$, we have that $\reconstruct(\share(\msg)_\calI)=\msg$ with probability one over the randomness of $\share$.
		\item \textbf{Privacy}: For any pair of messages $\msg_0,\msg_1\in\calM$, and for any $\calU \subset [n]$ such that $|\calU| < t$, we have that $$\{\share(1^\secpar,\msg_0)_\calU\}_{\secpar\in\NN} \cind \{\share(1^\secpar,\msg_1)_\calU\}_{\secpar\in\NN}.$$
	\end{itemize}
\end{definition}

\subsection{Secret-Key Encryption}
A secret-key encryption (SKE) scheme over key space $\calK$ is a pair of polynomial-time algorithms $(\Enc,\Dec)$ specified as follows.
(i) The randomized algorithm $\Enc$ takes as input a key $\key\in\calK$ and a message $\msg\in\calM$, and outputs a ciphertext $c\in\calC$;
(ii) The deterministic algorithm $\Dec$ takes as input a key $\key\in\calK$ and a ciphertext $c\in\calC$, and outputs a value in $\calM \cup \{\bot\}$.
Correctness says that for every key $\key\in\calK$, and every message $\msg\in\calM$, it holds that $\Dec(\key,\Enc(\key,\msg))=\msg$ with probability one over the randomness of $\Enc$.
\begin{definition}[Semantic security]
We say that $(\Enc,\Dec)$ satisfies semantic security if for all pairs of messages $\msg_0,\msg_1\in\calM$ it holds that $$\{\Enc(\key,\msg_0):\key\getsr\calK\}_{\secpar\in\NN} \cind \{\Enc(\key,\msg_1):\key\getsr\calK\}_{\secpar\in\NN}.$$
\end{definition}

\subsection{Non-Interactive Commitments}
A non-interactive commitment is a PPT algorithm $\Commit$ taking as input a message $\msg\in\bin^k$ and outputting a value $\com = \Commit(\msg,\rndcom)\in\bin^l$ where $\rndcom\in\bin^*$ is the randomness used to generate the commitment. The pair $(\msg,\rndcom)$ is also called the {\em opening}. 
A non-interactive commitment typically satisfies two properties known as binding and hiding; we review these properties (in the flavor we need them) below.

\begin{definition}[Perfect binding]
We say that $\Commit$ satisfies perfect binding if for all $\com\in\bin^l$ there do not exist values $(\msg_0,\msg_1,\rndcom_0,\rndcom_1)$, with $\msg_0 \ne \msg_1$, s.t.\ $\Commit(\msg_0,\rndcom_0) = \Commit(\msg_1,\allowbreak\rndcom_1)=\com$.	
\end{definition}	

\begin{definition}[Computational hiding]
We say that $\Commit$ satisfies computational hiding if for all pairs of message $\msg_0,\msg_1\in\bin^k$, it holds that $$\{\Commit(1^\secpar,\msg_0)\}_{\secpar\in\NN} \cind \{\Commit(1^\secpar,\msg_1)\}_{\secpar\in\NN}.$$
\end{definition}

%
%


\section{MPC Definitions}\label{app:def}
\subsection{Sequential Composability}
\paragraph*{The Real Execution.} 
In the real world, protocol $\pi$ is run by a set of parties $\party_1,\ldots,\party_n$ in the presence of a PPT adversary $\advA$ coordinated by a non-uniform distinguisher $\advD = \{\advD_\secpar\}_{\secpar \in \NN}$. 
At the outset, $\advD$ chooses the inputs $(1^\secpar,\inp_i)_{i\in[n]}$ for each player $\party_i$, the set of corrupted parties $\calI \subset [n]$, auxiliary information $\aux\in\bin^*$, and gives $((\inp_i)_{i\in\calI},\aux)$ to $\advA$.
Hence, the  protocol $\pi$ is run with the honest parties following their instructions (using input $\inp_i$), and with the attacker $\advA$ sending all messages of the corrupted players by following any polynomial-time strategy.  Finally, $\advA$ passes an arbitrary function of its view to $\advD$, who is also given the output of the honest parties, and returns a bit.

\paragraph*{The Ideal Execution.}
The ideal process for the computation of $\func$ involves a set of dummy parties $
\party_1,\ldots,\party_n$, an ideal adversary $\Sim$ (a.k.a.\ the simulator), and an 
ideal functionality $\Func_\func$.
At the outset, $\advD$ chooses $(1^\secpar,\inp_i)_{i\in[n]}$, $\calI$, $\aux$ and 
sends $((\inp_i)_{i\in\calI},\aux)$ to $\advS$.
Hence, each player $\party_i$ sends its input $\inp_i'$ to the ideal functionality---where $\inp_i' = \inp_i$ if $\party_i$ is honest, and otherwise $\inp_i'$ is chosen 
arbitrarily by the simulator---which returns to the parties their respective outputs $
\func_i(\inp_1',\ldots,\inp_n')$.  Finally, $\advS$ passes an arbitrary function of its 
view to $\advD$, who is also given the output of the honest parties, and returns a 
bit.

\paragraph*{Securely Realizing an Ideal Functionality.}
Intuitively, an MPC protocol is secure if whatever the attacker can learn in the real 
world can be emulated by the simulator in the ideal execution.
Formally, if we denote by $\REAL_{\pi,\advA,\advD}(\secpar)$ the random variable 
corresponding to the output of $\advD$ in the real execution, and by $
\IDEAL_{\func,\Sim,\advD}(\secpar)$ the random variable corresponding to the 
output of $\advD$ in the ideal execution, the two probability ensembles have to be 
computationally indistinguishable.

Let us write $\IDEAL_{\func_\bot,\Sim,\advD}(\secpar)$ for the output distribution in 
the above ideal execution.
\begin{definition}[Simulation-based security]\label{def:uc}
	Let $n\in\NN$.
	Let $\Func_\func$ be an ideal functionality for $\func:(\bin^*)^n \rightarrow 
(\bin^*)^n$, and let $\pi$ be an $n$-party protocol.
	We say that $\pi$ $s$-securely computes $\Func_\func$ if for any PPT 
adversary $\advA$ there exists a PPT simulator $\Sim$ such that all PPT non-
uniform distinguishers $\advD$ corrupting at most $s$ parties, we have
	\begin{equation*}
	\left\{\IDEAL_{\func,\Sim,\advD}(\secpar)\right\}_{\secpar\in\NN} 
	\cind 
	\left\{\REAL_{\pi,\advA,\advD}(\secpar)\right\}_{\secpar\in\NN}.
	\end{equation*}
\end{definition}
When replacing $\IDEAL_{\func,\Sim,\advD}(\secpar)$ with $\IDEAL_{\func_\bot,
\Sim,\advD}(\secpar)$ in Def.~\ref{def:uc}, we say that $\pi$ $s$-securely computes 
$\func$ with aborts.

\paragraph*{The Hybrid Model.}
Let $\HYBRID^\gunc_{\pi,\advA,\advD}(\secpar)$ denote the random variable 
corresponding to the output of $\advD$ in the $\Func_\gunc$-hybrid model.
We say that a protocol $\pi_f$ for computing $\func$ is secure in the $\Func_\gunc$-hybrid model if the ensembles $\HYBRID^\gunc_{\pi_f,\advA,\advD}(\secpar)$ and $
\IDEAL_{\func,\Sim,\advD}(\secpar)$ are computationally close.

\begin{definition}[MPC in the hybrid model]\label{def:uc_hyb}
	Let $n\in\NN$.
	Let $\Func_\func$, $\Func_\gunc$ be ideal functionalities for $\func,\gunc:
(\bin^*)^n \rightarrow (\bin^*)^n$, and let $\pi$ be an $n$-party protocol.
	We say that $\pi$ $s$-securely realizes $\Func_\func$ in the $\Func_\gunc$-
hybrid model if for all PPT adversaries $\advA$ there exists a PPT simulator $\Sim$ 
such that for all PPT non-uniform distinguishers $\advD$ corrupting at most $s$ 
players, we have
	\begin{equation*}
	\left\{\IDEAL_{\func,\Sim,\advD}(\secpar)\right\}_{\secpar\in\NN} 
	\cind 
	\left\{\HYBRID^\gunc_{\pi,\advA,\advD}(\secpar)\right\}_{\secpar\in\NN}.
	\end{equation*}	
\end{definition}
\subsection{Universal Composability}

\paragraph*{The Real Execution}
In the real world, the protocol $\pi$ is run in the presence of an adversary $\advA$ coordinated by a non-uniform environment $\advZ = \{\advZ_\secpar\}_{\secpar\in\NN}$.
At the outset, $\advZ$ chooses the inputs $(1^\secpar,\inp_i)$ for each player $\party_i$, and gives $\calI$, $\{\inp_i\}_{i\in\calI}$ and $z$ to $\advA$, where $\calI \subset [\num]$ represents the set of corrupted players and $z$ is some auxiliary input. For simplicity, we only consider static corruptions (i.e.,\ the environment decides who is corrupt at the beginning of the protocol).
The parties then start running $\pi$, with the honest players $\party_i$ behaving as prescribed in the protocol (using input $x_i$), and with malicious parties behaving arbitrarily (directed by $\advA$). The attacker may delay sending the messages of the corrupted parties in any given round until after the honest parties send their messages in that round; thus, for every $\round$, the round-$\round$ messages of the corrupted parties may depend on the round-$\round$ messages of the honest parties. $\advZ$ can interact with $\advA$ throughout the course of the protocol execution.\\
 Additionally, $\advZ$ receives the outputs of the honest parties, and must output a bit. We denote by $\REAL_{\pi,\advA,\advZ}(\secpar)$ the random variable corresponding to $\advZ$'s guess.

\paragraph*{The Ideal Execution}
In the ideal world, a trusted third party evaluates the function $\func$ on behalf of a set of dummy players $(\party_i)_{i \in [\num]}$. As in the real setting, $\advZ$ chooses the inputs $(1^\secpar,\inp_i)$ for each honest player $\party_i$, and gives $\calI$, $\{\inp_i\}_{i\in\calI}$ and $z$ to the ideal adversary $\advS$, corrupting the dummy parties $(\party_i)_{i \in \calI}$. Hence, honest parties send their input $x_i' = x_i$ to the trusted party, whereas the parties controlled by $\advS$ might send an arbitrary input $x_i'$. The trusted party computes $(y_1,\ldots,y_\num) = \func(x_1',\ldots,x_\num')$, and sends $y_i$ to $\party_i$. During the simulation, $\advS$ and $\advA$ can interact with $\advZ$ throughout the course of the protocol execution.
 Additionally, $\advZ$ receives the outputs of the honest parties and must output a bit. We denote by $\IDEAL_{\func,\advS,\advZ}(\secpar)$ the random variable corresponding to $\advZ$'s guess.

\begin{definition}[UC-Secure MPC]\label{def:mpc}
Let $\pi$ be an $\num$-party protocol for computing a function $\func:(\bin^*)^\num\allowbreak\rightarrow(\bin^*)^\num$. We say that $\pi$ $\thr$-securely UC-realizes $\func$ in the presence of malicious adversaries if such that for every PPT adversary $\advA$ there exists a PPT simulator $\advS$ such that for every non-uniform PPT environment $\advZ$ corrupting at most $t$ parties the following holds:
\[
\left\{\REAL_{\pi,\advA,\advZ}(\secpar)\right\}_{\secpar\in\NN} \cind \left\{\IDEAL_{\func,\advS,\advZ}(\secpar)\right\}_{\secpar\in\NN}.
\]
When replacing $\IDEAL_{\func,\advS,\advZ}(\secpar)$ with $\IDEAL_{\func_\bot,\advS,\advZ}(\secpar)$ we say that $\pi$ $\thr$-securely computes $\func$ {\em with aborts} in the presence of malicious adversaries.
\end{definition}

Hence, we also consider a weakening of the ideal process involving a functionality $\Func_{\func}^\bot$ which behaves identically to $\Func_\func$ except that:
(i) When the trusted party computes the outputs $(\out_1,\ldots,\out_n)$, it first gives $(\out_i)_{i\in\calI}$ of the corrupted players 
to the simulator; and
(ii) The simulator may send either a message $\mathtt{continue}$ or $\mathtt{abort}$ to the trusted party: in the former case, the trusted party sends the 
output $\out_i$ to all honest players, in the latter case they
receive an abort symbol $\bot$. All of the provably secure constructions that we will present we assume that the adversaries are static, i.e. that are fixed at the beginning of the execution.

An even stronger composability guarantee is that of \emph{generalized} UC~\cite{CanettiDPW07}. Roughly, in this setting, entities have black-box access to one or more global functionalities. A global functionality is an external functionality accessible by both the real and ideal world entities. The same instance of a global functionality can be accessed by multiple protocols (either simultaneously or sequentially depending on the model).

\paragraph*{The MPC Hybrid Model.}
Let $\HYBRID^\gunc_{\pi,\advA,\advZ}(\secpar)$ denote the random variable 
corresponding to the output of $\advZ$ in the $\Func_\gunc$-hybrid model.
We say that a protocol $\pi_f$ for computing $\func$ is secure in the $\Func_\gunc$-hybrid model if the ensembles $\HYBRID^\gunc_{\pi_f,\advA,\advZ}(\secpar)$ and $
\IDEAL_{\func,\Sim,\advZ}(\secpar)$ are computationally close.
\begin{definition}[UC-Secure MPC in the hybrid model]
	Let $n\in\NN$.
	Let $\Func_\func$, $\Func_\gunc$ be ideal functionalities for $\func,\gunc:
(\bin^*)^n \rightarrow (\bin^*)^n$, and let $\pi$ be an $n$-party protocol.
	We say that $\pi$ $\thr$-securely realizes $\Func_\func$ in the $\Func_\gunc$-
hybrid model if for all PPT adversaries $\advA$ there exists a PPT simulator $\Sim$ 
such that for all PPT non-uniform environments $\advZ$ corrupting at most $\thr$ 
players, we have
	\begin{equation*}
	\left\{\IDEAL_{\func,\Sim,\advZ}(\secpar)\right\}_{\secpar\in\NN} 
	\cind 
	\left\{\HYBRID^\gunc_{\pi,\advA,\advZ}(\secpar)\right\}_{\secpar\in\NN}.
	\end{equation*}	
\end{definition}
The same definitions are applicable to the stand-alone setting by replacing the environment with a distinguisher $\advD$. He can only handle inputs to the parties and receive back an arbitrary function of $\advA$'s view, and then outputs a bit.

\subsection{Security with Penalties}
The formal definition of secure MPC with penalties (in the hybrid model) is identical to the above, except that $\Func_\func^\bot$ is replaced by the ideal functionality $\Func_\func^*$ depicted in Fig.~\ref{fig:fair}. 
At the outset, the distinguisher $\advD$ (or the environment $\advZ$) initializes each 
party's wallet with some number of coins.
$\advD$ (or $\advZ$) may read or modify ({\em i.e.},\ add coins to or retrieve coins from) 
the wallet (but not the safe) of each honest party, whereas in the 
hybrid (resp.\ ideal) process, the adversary $\advA$ (resp.\ $\advS$) has complete 
access to both wallets and safes of corrupt parties.

At the end of the protocol, honest parties release the coins locked in the protocol to 
their wallet, the distinguisher  is 
given the distribution of coins and outputs its final output.

\begin{figure*}[!h]
		\begin{framed}\footnotesize
		\begin{center}
			\textbf{Functionality $\Func_\func^*$}
		\end{center}
		The functionality runs with security parameter $1^\secpar$, minimum penalty amount $q$, parties $\party_1,\ldots,\party_n$, and adversary $\Sim$ that corrupts parties $\{\party_i\}_{i\in\calI}$ for some $\calI \subseteq [n]$. Let $\calH = [n] \setminus \calI$ with $h = |\calH|$, and $d$ be the amount of the safety deposit.
		\begin{description}
			\item[Input Phase:]
			Wait to receive a message $(\mathtt{input},sid,ssid,i,x_i,\coins(d))$ from $\party_i$ for all $i \in \calH$, and message $(\mathtt{input},\allowbreak\{x_i\}_{i \in \calI},\coins(hq))$ from $\Sim$.  
			\item[Output Phase:] Proceed as follows.
			\begin{itemize}
				\item Send $(\mathtt{return},sid,ssid,\coins(d))$ to $\party_i$ for all $i\in\calH$, and let $(y_1,\ldots,y_n) = \func(x_1,\ldots,x_n)$.
				\item Send $(\mathtt{output},sid,ssid,\{y_i\}_{i \in \calI})$ to $\Sim$.
				\item If $\Sim$ returns $(\mathtt{continue},sid,ssid,\calH_{\sf out})$, then send $(\mathtt{output},sid,ssid,y_i)$ to $\party_i$ for each $i\in\calH$, send $(\mathtt{payback},sid,ssid,\coins((h-|\calH_{\sf out}|)q)$ to $\Sim$, and send $(\mathtt{extrapay},sid,ssid,\coins(q))$ to $\party_i$ for each $i\in\calH_{\sf out}$.
				\item Else, if $\Sim$ returns $(\mathtt{abort},sid,ssid)$ send $(\mathtt{penalty},sid,ssid,\coins(q))$ to $\party_i$ for each $i\in\calH$.
			\end{itemize}
		\end{description}	
		\vspace*{-\baselineskip}
	\end{framed}
	\vspace*{-\baselineskip}
	\caption{The functionality $\Func^*_\func$ for secure computation with penalties~\cite{BentovK14}.}\label{fig:fair}
	\vspace*{-\baselineskip}
\end{figure*}

\section{Proof of Theorem \ref{thm:multilock_security}}\label{app:multilock_security}
In this section we prove Thm.~\ref{thm:multilock_security}
A more detailed description of the protocol can be found in Fig. \ref{fig:newseesaw}.
\begin{figure}[!t]\footnotesize
	\begin{framed}
		\begin{center}
			\textbf{The Compact Multi-Lock ($\CMLMech$) Penalty Protocol}
		\end{center}
		\vspace*{-0.2\baselineskip}
		Let $(\share,\reconstruct)$ be an $n$-party secret sharing scheme, $(\Enc,\Dec)$ be a secret-key encryption scheme, $\Commit$ be a non-interactive commitment, and $\func:(\bin^*)^n\rightarrow(\bin^*)^n$ be the function being computed. 
		The protocol runs with security parameter $1^\secpar$, parties $\party_1,\ldots,\party_n$ holding inputs $\inp_1,\allowbreak\ldots,\allowbreak\inp_n$, penalty amount $q$, timeout $\timeout$, and hybrid functionalities $\FuncML$ and $\Func^\bot_{\tilde f}$ for the following derived function $\tilde\func:(\bin^*)^n\rightarrow(\bin^*)^n$:
		\[
		\tilde\func(\inp_1,\ldots,\inp_n) \mapsto (((\com_i)_{i\in[n]},\ctx,\key_1,\rndcom_1), \ldots, ((\com_i)_{i\in[n]},\ctx,\key_n,\rndcom_n)),
		\]
		where $\key\getsr\bin^\secpar$, $(\key_1,\ldots,\key_n) \getsr \share(\key)$, $\ctx \getsr \Enc(\key,\allowbreak\func(\inp_1,\allowbreak\ldots,\allowbreak\inp_n))$, and $\com_i = \Commit(\key_i;\allowbreak\rndcom_i)$ for uniformly random $\rndcom_i \in \bin^*$.
		\begin{description}
			\item[Shares Distribution Phase:]
			Each player $\party_i$ hands its input $\inp_i$ to the ideal functionality $\Func^\bot_{\tilde f}$ and receives back $((\com_j)_{j \in [n]},\ctx,\key_i,\rndcom_i)$. If an $\abort$ message is received, then abort the protocol.
			
			\item[Fair Reconstruction Phase:] 
			Let $\phi(\com_i; \wit_i = (\key_i,\rndcom_i))=1$ iff $\Commit(\key_i;\allowbreak\rndcom_i) = \com_i$, and set $D_i = (\deposit,\allowbreak(\phi(\com_i;\cdot),\allowbreak\ldots,\phi(\com_n;\cdot)),\timeout)$ with $\deposit = q(n-1)$. Player $\party_i$ proceeds as follows:
			\begin{enumerate}[leftmargin=0.4cm]
				\item Send $(\lock,i,D_i,\coins(\deposit))$ to $\FuncML$.
				\item\label{step:lock} Wait to receive one of two possible messages from $\FuncML$:
				\begin{itemize}[leftmargin=0.2cm]
					\item Upon  $(\abort,i,\coins(\deposit))$, output $\bot$ and terminate.
					\item Upon  $(\locked)$, send $(\redeem,i,(\key_i,\rndcom_i))$ to $\FuncML$ receiving back $(\redeem,\coins(\deposit))$.
				\end{itemize}
				\item Initialize a counter $\ell = 0$ and wait until round $\timeout$ for the other players to reveal their witness:
				\begin{itemize}[leftmargin=0.2cm]
					\item Upon  $(\redeem,j,\key'_j)$ from $\FuncML$, set $\ell \gets \ell + 1$.
					\item In case $\ell = n-1$, compute $\key' = \reconstruct(\key'_1,\ldots,\key'_n)$ and $\out' = \Dec(\key',\ctx)$, and terminate.
				\end{itemize}
				\item At round $\timeout+1$, wait to receive $n-1-\ell$ messages from $\FuncML$:
				\begin{itemize}[leftmargin=0.2cm]
					\item Upon  $(\payout,i,j,\coins(\frac{\deposit}{n-1}))$ for each unrevealed witness $\key_j'$, let $\ell \gets \ell+1$.
					\item When $\ell = n-1$, then stop the execution.
				\end{itemize}
			\end{enumerate}
		\end{description}
		\vspace*{-\baselineskip}
	\end{framed}
	\vspace*{-\baselineskip}
	\caption{Our alternative financially and cryptographically fair protocol for general-purpose MPC with penalties.}
	\label{fig:newseesaw}
	\vspace*{-\baselineskip}
\end{figure}

	We begin by describing the simulator $\Sim$ in the 
ideal world. Let $\calI$ be the set of corrupted parties (recall that this set is fixed 
before the protocol starts), and $\calH = [n] \setminus \calI$ be the set of honest 
parties (with $h = |\calH|$).
\begin{enumerate}
	\item\label{step:simsetup} Acting as $\Func^\bot_{\tilde\func}$, wait to receive the inputs $\{\inp_i\}_{i \in \calI}$ from $\advA$.
	Hence, sample $\key,\tilde\key\getsr\bin^\secpar$, let $(\key_1,\allowbreak\ldots,\key_n) \getsr \share(\tilde\key)$, $\ctx \getsr \Enc(\key,0^m)$ and $\com_i \allowbreak\getsr\allowbreak \Commit(0^k)$ for all $i \in \calH$, and send $((\com_j)_{j\in[n]},\allowbreak\ctx,\allowbreak \key_i,\rndcom_i)_{i\in\calI}$ to $\advA$.
	\item\label{step:unfair} Acting as $\FuncML$, wait to receive from $\advA$ the message $(\lock,\allowbreak i,D_i)$ for each $i \in \calI$, where $D_i = (\deposit,\allowbreak(\phi(\com_1;\cdot),\allowbreak\ldots,\phi(\com_n;\cdot)),\timeout)$.
	\begin{itemize} 
		\item If for some $j \in \calI$ the message is not received, or if $\exists i,j\in\calI$ 
		such that $D_i \ne D_j$, then return message $(\abort,\allowbreak i,\coins(\deposit))$ to $\advA$ for each 
		corrupted $\party_i$ that deposited, and terminate the simulation.
		\item Else, send $(\locked)$ to $\advA$ for each corrupted party.
	\end{itemize}
	\item Send $(\cominput,\{\inp_i\}_{i \in\calI},\coins(hq))$ to $\Func_f^*$, receiving $(\comoutput,\allowbreak\out)$ back. Hence, rewind the execution of $\advA$ to step~\ref{step:simsetup}, change the distribution of $((\com_j)_{j\in[n]},$ $\ctx,\key_i,\rndcom_i)_{i\in\calI}$ to that of the real protocol $\pi$, and repeat step~\ref{step:unfair} of the simulation, except that, in case $\advA$ now aborts, the rewinding is repeated with fresh randomness and step~\ref{step:unfair} is run again.
	\item At round $\timeout$, acting as $\FuncML$, send the message $(\redeem,\allowbreak i,(\key_i,\rndcom_i))$ to $\advA$ for each $i\in\calH$. Set $\ell = 0$. Hence,
	upon receiving message $(\redeem,i,(\key'_i,\rndcom_i'))$ from $\advA$ (on behalf of each corrupted $\party_i$):
	\begin{itemize}
		\item If $\ell < n-h$, check that $\phi(\com_i;(\key'_i,\rndcom'_i))=1$; if the check 
		passes and also $\key'_i = \key_i$, send $(\redeem,\allowbreak i,\coins(d))$ to $\party_i$ and $(\redeem,i,(\key'_i,\rndcom_i'))$ to every other corrupted party, and 
		update $\ell \gets \ell + 1$.
		\item If $\ell = n-h$, send $(\cont,\emptyset)$ to $\Func^*_f$, receive back $(\payback,\coins(hq))$ as answer, and terminate the simulation.
	\end{itemize}
	\item At round $\timeout+1$, if $\ell < n-h$, send message $(\payout,\allowbreak i,j,
	\coins(\tfrac{\deposit}{n-1}))$ to each corrupted $\party_j \ne \party_i$ on behalf of each 
	corrupted player $\party_i$ that did not redeem its witness in the previous step of 
	the simulation. Hence, send $(\abort)$ to $\Func^*_\func$ and terminate the 
	simulation.
\end{enumerate}
We remark that it is not immediate that the simulator runs in (expected) polynomial time. This is because in case $\advA$ does not abort during the {\em first run} before the rewinding, the simulator needs to make sure that $\advA$ does not abort during the {\em second run} either, as otherwise the probability of abort would be different in the ideal and real world.
This issue can be solved by slightly adjusting the way the simulator deals with aborts, using a technique originally described in~\cite{GoldreichK96} (see also~\cite[\S 5.4]{Lindell16}).

To conclude the proof, we consider a sequence of hybrid experiments and show that each pair of hybrid distributions derived from the experiments are computationally close.
\begin{description}
	\item[$\hyb_0(\secpar)$:] Identical to $\HYBRID^{\FuncML,\Func^\bot_{\tilde\func}}_{\pi,\advA,\advD}(\secpar)$, except that during the fair 
	reconstruction phase, in case the attacker does not provoke an abort during 
	step~\ref{step:lock} in Fig.~\ref{fig:newseesaw}, we rewind the adversary to the 
	share distribution phase and re-run the entire protocol.
	Clearly, such an artificial rewind does not change the distribution of the hybrid 
	execution in which protocol $\pi$ is originally run.
	\item[$\hyb_1(\secpar)$:] As above except that during 
	the {\em first run} of the share reconstruction phase (before the rewinding) we 
	switch the distribution of the commitments to $\com_i \getsr \Commit(0^k)$ for each $i \in \calH$. 
	During the {\em second run} (after the rewind) of the share distribution phase, if any, 
	the honest commitments are reset to the original distribution $\com_i \getsr \Commit(\key_i)$.
	\item[$\hyb_2(\secpar):$] As above, except that during the {\em first run} of the share reconstruction phase (before the rewinding) we switch the distribution of the shares to $(\key_1,\ldots,\key_n) \getsr \share(\tilde\key)$ for random $\tilde\key\getsr\bin^\secpar$ independent from $\key\getsr \allowbreak \bin^\secpar$.
	\item[$\hyb_3(\secpar):$] As above, except that during the {\em first run} of the share reconstruction phase (before the rewinding) we switch the distribution of the ciphertext to $\ctx \getsr \Enc(\key,0^m)$.
	\item[$\hyb_4(\secpar):$] Identical to $\IDEAL_{\func^*,\Sim,\advD}(\secpar)$ for the above defined simulator $\Sim$.
\end{description}
\begin{lemma} \label{lem:hyb01}
	$\{\hyb_0(\secpar)\}_{\secpar \in \setN} \cind \{\hyb_1(\secpar)\}_{\secpar \in \setN}$:
\end{lemma}
\begin{proof}
Let $h = |\calH|$. For the rest of the proof, we will refer to the keys and the commitments of the honest players as $\key_H^j$ and $\com_H^j$ for each $j \in [h]$, meaning that we are considering the key/commitment of the $j$-th honest player in $\calH$. When we write $\key_j$ or $\com_j$ we are referring to the key and the commitment of the player $\party_j$ instead.
 For $i \in [0,h]$, consider a hybrid argument in which $\hyb_{0,i}(\secpar)$ is modified as in $\hyb_1(\secpar)$ for the  first $i$ honest parties, meaning that except that during the {\em first run} of the share distribution phase we set $\com_H^j \getsr \Commit(0^k)$ for all $j\le i$, whereas $\com_H^j \getsr \Commit(\key_H^j)$ for all $j > i$. Then, it suffices to show that for every $i\in[0,h-1]$: 
\[
\{\hyb_{0,i}(\secpar)\}_{\secpar \in \setN} \cind \{\hyb_{0,i+1}(\secpar)\}_{\secpar \in \setN}.
\]
The proof is by contradiction: assume that there exists a PPT attacker $\advA$ and a non-uniform PPT 
distinguisher  $\advD$ that can distinguish $\hyb_{0,i}(\secpar)$ and $\hyb_{0,i+1}\allowbreak(\secpar)$ with non-negligible probability. Then we can build a non-uniform PPT distinguisher $\hat\advD$ breaking the hiding 
property of the commitment. 
%
Fix $i\in[0,h-1]$. 
We  build a non-uniform PPT distinguisher $\hat\advD$ breaking the hiding property of the 
commitment. A description of $\hat\advD$ follows.
\begin{itemize}
	\item Run $\dist$, obtaining $\calI$, $\{x_i\}_{i \in [n]}$, and $z$; forward $\{\inp_i\}_{i\in\calI}$ and $\aux$ 
	to $\advA$.
	\item Sample $\key \getsr \bin^\secpar$ and $(\key_1,\ldots,\key_n) \getsr \share(\key)$.
	  For each $j \le i$
	define $\com_H^j \getsr \Commit(0^k)$, whereas for each $j > i+1$ let $\com_H^{j} \getsr \Commit(\key_H^j)$.
	\item Forward $(0^k,\key_{H}^{i+1})$ to the challenger. Upon receiving a challenge commitment $\hat\com$, 
	let $\com_{i+1} = \hat\com$. 
	\item Let $\ctx \getsr \Enc(\key;\func(\inp_1,\ldots,\inp_n))$, and forward $((\com_j)_{j\in[n]},
	\allowbreak\ctx,\key_i,\allowbreak\rndcom_i)_{i\in\calI}$ to $\advA$.
	\item Continue the simulation with $\advA$ exactly as described in $\hyb_{0,i}(\secpar)$ or $\hyb_{0,i+1}\allowbreak
	(\secpar)$ (in fact, note that the two experiments are identical after the {\em first run} of the share 
	distribution phase); when $\advA$ terminates pass its output to $\advD$ together with the output of the 
	honest parties, and return the same guess as that of $\advD$.
\end{itemize}
By inspection, the simulation done by $\hat\advD$ is perfect, in the sense that when $\hat\com$ is a 
commitment to $0^k$ the view of $\advA$ is identical to that in a run of $\hyb_{0,i}(\secpar)$, whereas 
when $\hat\com$ is a commitment to $\key_H^{i+1}$ the view of $\advA$ is identical to that in a run of $
\hyb_{0,i+1}(\secpar)$. Hence, the view of $\advD$ is also perfectly simulated, so that $\hat\advD$ 
breaks the hiding property with non-negligible probability.
\end{proof}

\begin{lemma}\label{lem:hyb12}
	$\{\hyb_1(\secpar)\}_{\secpar \in \setN} \cind \{\hyb_2(\secpar)\}_{\secpar \in \setN}$.
\end{lemma}
\begin{proof}
We prove this claim by contradiction: if there exists a PPT attacker $\advA$ and a 
non-uniform PPT distinguisher $\advD$ that can tell apart the two experiments with 
better than negligible probability, then we can build a PPT distinguisher $\hat\advD$ 
attacking privacy of the underlying threshold secret sharing scheme. 
A description of $\hat\advD$ follows.
\begin{itemize}
	\item Run $\dist$, obtaining $\calI$, $\{x_i\}_{i \in [n]}$, and $z$; forward $
	\{x_i\}_{i\in\calI}$ and $z$ to $\advA$.
	\item Sample $\key,\tilde\key \getsr \bin^\secpar$ and forward $(\calU,
	\key,\tilde\key)$ to the challenger, where $\calU = \calI$.
	\item Upon receiving $(\hat\key_i)_{i\in\calI}$ from the challenger, for each $j \in \calH$ let $\com_j \getsr \allowbreak \Commit(0^k)$, compute $\ctx \getsr 
	\allowbreak\Enc(\key,\allowbreak f(x_1,\ldots,\allowbreak x_n))$, and pass $((\com_j)_{j\in[n]},
	\ctx,\hat\key_i,\rndcom_i)_{i\in\calI}$ to $\advA$.
	\item Continue the simulation with $\advA$ exactly as described in $
	\hyb_1(\secpar)$ or $\hyb_2(\secpar)$ (in fact, note that the two experiments are 
	identical after the {\em first run} of the share distribution phase); when $\advA$ 
	terminates pass its output to $\advD$ together with the output of the honest parties, 
	and return the same guess as that of $\advD$.
\end{itemize}
By inspection, the simulation done by $\hat\advD$ is perfect, in the sense that 
when 
$(\hat\key_i)_{i\in\calI}$ are taken from a secret sharing of $\key$ the view of 
$\advA$ is identical to that in a run of $\hyb_{1}(\secpar)$, whereas when 
$(\hat\key_i)_{i\in\calI}$ are taken from a secret sharing of $\tilde\key$ the view 
of 
$\advA$ is identical to that in a run of $\hyb_{2}(\secpar)$. Hence, the view of 
$\advD$ is also perfectly simulated, so that $\hat\advD$ breaks the privacy 
property 
with non-negligible probability.
\end{proof}

\begin{lemma}\label{lem:hyb23}
	$\{\hyb_2(\secpar)\}_{\secpar \in \setN} \cind \{\hyb_3(\secpar)\}_{\secpar \in \setN}$:
\end{lemma}
\begin{proof}
Once again, we prove this claim by assuming that there exists a PPT attacker 
$\advA$ and a non-uniform PPT distinguisher $\advD$ that can tell apart the two 
experiments with better than negligible probability, from which we can  build a PPT distinguisher 
$\hat\advD$ attacking the semantic security of the underlying secret key encryption 
scheme. 
A description of $\hat\advD$ follows.
\begin{itemize}
	\item Run $\dist$, obtaining $\calI$, $\{x_i\}_{i \in [n]}$, and $z$; forward $
	\{x_i\}_{i\in\calI}$ and $z$ to $\advA$.
	\item Sample $\key,\tilde\key \getsr \bin^\secpar$ and $(\key_1,\ldots,
	\key_n) \getsr \share(\tilde\key)$
	\item Forward $(\hat\msg_0 = f(x_1,\ldots,x_n),\hat\msg_1 = 0^m)$ to the 
	challenger.
	\item Upon receiving a challenge ciphertext $\hat\ctx$ from the 
	challenger, for each $j \in \calH$ let $\com_j \getsr \Commit(0^k)$ and pass $
	((\com_j)_{j\in[n]},\ctx,\key_i,\rndcom_i)_{i\in\calI}$ to $\advA$.
	\item Continue the simulation with $\advA$ exactly as described in $
	\hyb_2(\secpar)$ or $\hyb_3(\secpar)$ (in fact, note that the two experiments are 
	identical after the {\em first run} of the share distribution phase); when $\advA$ 
	terminates pass its output to $\advD$ together with the output of the honest parties, 
	and return the same guess as that of $\advD$.
\end{itemize}
By inspection, the simulation done by $\hat\advD$ is perfect, in the sense that 
when $\hat\ctx$ is an encryption of $f(x_1,\ldots,x_n)$ the view of $\advA$ is 
identical to that in a run of $\hyb_{2}(\secpar)$, whereas when $\hat\ctx$ is an 
encryption of $0^m$ the view of $\advA$ is identical to that in a run of $\hyb_{3}
(\secpar)$. Hence, the view of $\advD$ is also perfectly simulated, so that $
\hat\advD$ breaks the privacy property with non-negligible probability.
\end{proof}

\begin{lemma}
	$\{\hyb_3(\secpar)\}_{\secpar \in \setN} \equiv \{\hyb_4(\secpar)\}_{\secpar \in \setN}$.
\end{lemma}
\begin{proof}
	We claim that the distribution of $\hyb_3(\secpar)$ and $\hyb_4(\secpar)$ are identical.
	First, note that the distribution of the honest parties' output is the same in the two experiments.
	This follows by perfect binding of the commitments and perfect correctness of the encryption scheme. In fact, in $\hyb_3(\secpar)$, the commitments $(\com_j)_{j\in[n]}$ can be opened only to a valid share $\key_j$ of the secret key $\key$ under which the value $f(x_1,\ldots,x_n)$ is encrypted, and this invariant is maintained by the simulator.
	
	Second, the simulator $\advS$ perfectly emulates the functionality $\FuncML$.
	This is because, $\advS$ sends $\Func_f^*$ a penalty amount of $\coins(hq)$ that are returned to $\Sim$ if and only if all corrupted parties exhibit a valid opening during the redeem phase of the protocol. In particular, in case all corrupted parties reveal a valid witness, $\Sim$ sends $\cont$ to $\Func^*_f$, receiving $\coins(hq)$ back, which shows that the coins are correctly redistributed among the parties.
\end{proof}
\noindent The theorem  follows by combining the above lemmas.

\section{Further Examples of Ladder Protocols}\label{apdx:ladderexamples}
Below, 
we illustrate the 3-party case of the \LLMech\ protocol (without considering Bootstrap Deposits).

\begin{center}\footnotesize
	\textbf{ROOF}: $ P_j \xrightarrow[q,\tau_4]{\hspace*{1cm} TT_3 \hspace*{1cm}} P_3$ (for $j \in \{1,2\}$)\\
	\hspace*{0.1cm}\textbf{LADDER}:
	$ P_3 \xrightarrow[q,\tau_3]{\hspace*{1cm} TT_2 \land U_{3,2} \hspace*{1cm}} P_2$ (Rung Unlock)\\
	\hspace*{0.7cm} $ P_3 \xrightarrow[2q,\tau_2]{\hspace*{1cm} TT_2 \hspace*{1cm}} P_2$ (Rung Climb)\\
	\hspace*{1.1cm} $ P_2 \xrightarrow[q,\tau_2]{\hspace*{1cm} TT_1 \land U_{2,3} \hspace*{1cm}} P_3$ (Rung Lock)\\
	\hspace*{-1.8cm}\textbf{FOOT}:
	$ P_2 \xrightarrow[q,\tau_1]{\hspace*{1cm} TT_1 \hspace*{1cm}} P_1$ 
\end{center}

The tokens $U_{i,j}$ are necessary to avoid specific attacks, for further details look at \cite{kumaresan2015use}.

Fig.~\ref{fig:lockedladder4} and Fig.~\ref{fig:plantedladder4} show the amount of coins locked during the Deposit and Claim Phase of the 4-party $\LMech$ and $\PLMech$ Protocols. Since $\LMech$ and $\PLMech$ are for reactive functionalities, in the figures we assume a 2-stages functionality for each protocol. Fig.~\ref{fig:amortizedladder4} illustrates the amount of coins locked for the $\ALMech$ Protocol.

\begin{figure}[!t]
\centering
	\resizebox*{7cm}{!}{\begin{tikzpicture}
\begin{groupplot}[
group style={group size=4 by 1,vertical sep = 1cm},
height=5cm,width=5cm, ymin = -13, 
ymax = 0]
\nextgroupplot[title=$P_1$,xlabel=Time (Prot.\ Steps),ylabel=(Multiples of $\deposit$),legend to name = curves2]
\addplot[mark = *,red] coordinates { 
	(0,0)
	(1,-1)
	(2,-1)
	(3,-1)
	(4,-1)
	(5,-1)
	(6,-1)
	(7,-1)
	(8,-1)
	(9,-1)
	(10,-4)
	(11,-5)
	(12,-5)
	(13,-5)
	(14,-5)
	(15,-5)
	(16,-5)
	(17,-5)
	(18,-5)
	(19,-5)
	(20,-8)
};
\addlegendentry{Deposits}
\addplot[mark = *,blue] coordinates { 
	(20,-8)
	(21,-5)
	(22,-5)
	(23,-4)
	(24,-4)
	(25,-4)
	(26,-4)
	(27,-4)
	(28,-4)
	(29,-1)
	(30,0)
	(31,0)
	(32,0)
	(33,0)
	(34,0)
	(35,0)
};
\addlegendentry{Withdrawals}
\coordinate (top1) at (axis cs:20,\pgfkeysvalueof{/pgfplots/ymax});
\coordinate (bot1) at (axis cs:20,\pgfkeysvalueof{/pgfplots/ymin});
%
%
\nextgroupplot[title=$P_4$](b)
\addplot[mark = *,red] coordinates { 
	(0,0)
	(1,0)
	(2,-1)
	(3,-4)
	(4,-4)
	(5,-5)
	(6,-5)
	(7,-5)
	(8,-5)
	(9,-6)
	(10,-6)
	(11,-6)
	(12,-7)
	(13,-10)
	(14,-10)
	(15,-11)
	(16,-11)
	(17,-11)
	(18,-11)
	(19,-12)
	(20,-12)
};
\addplot[mark = *,blue] coordinates {
	(20,-12)
	(21,-12)
	(22,-11)
	(23,-11)
	(24,-10)
	(25,-10)
	(26,-9)
	(27,-9)
	(28,-6)
	(29,-5)
	(30,-5)
	(31,-4)
	(32,-4)
	(33,-3)
	(34,-3)
	(35,0)
};
\coordinate (bot) at (rel axis cs:1,0);
\coordinate (top4) at (axis cs:20,\pgfkeysvalueof{/pgfplots/ymax});
\coordinate (bot4) at (axis cs:20,\pgfkeysvalueof{/pgfplots/ymin});
\end{groupplot}
\path (top)--(bot) coordinate[midway] (group center);
\node[above,rotate=90] at (group center -| current bounding box.west) {Coins};
\draw [dashed] (bot1) -- (top1);
\draw [dashed] (bot4) -- (top4);
\end{tikzpicture}}
		\vspace*{-0.8\baselineskip}
	\caption{Coins locked in a run of the 4-party 2-stages Locked Ladder Protocol during the Deposit Phase (in red) and the Claim Phase (in blue).}\label{fig:lockedladder4}
\end{figure}
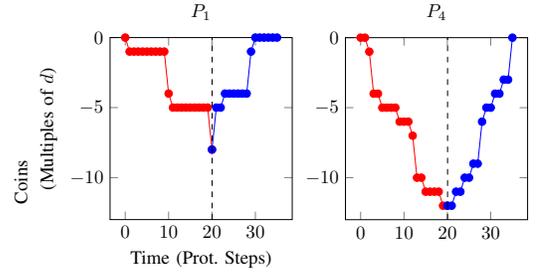

\begin{figure}[!t]
\centering
\resizebox*{7cm}{!}{\begin{tikzpicture}
\begin{groupplot}[
group style={group size=4 by 1,vertical sep = 1cm},
height=5cm,width=5cm, ymin = -22, 
ymax = 0]
\nextgroupplot[title=$P_1$,xlabel=Time (Prot.\ Steps),ylabel=(Multiples of $\deposit$),legend to name = curves2]
\addplot[mark = *,red] coordinates { 
	(0,0)
	(1,-3)
	(2,-3)
	(3,-3)
	(4,-3)
	(5,-11)
	(6,-11)
	(7,-11)
	(8,-11)
	(9,-15)
	(10,-15)
	(11,-15)
	(12,-15)
	
};
\addlegendentry{Deposits}
\addplot[mark = *,blue] coordinates { 
	(12,-15)
	(13,-14)
	(14,-14)
	(15,-14)
	(16,-14)
	(17,-9)
	(18,-9)
	(19,-9)
	(20,-9)
	(21,0)
	(22,0)
	(23,0)
	(24,0)
	
};
\addlegendentry{Withdrawals}
\coordinate (top1) at (axis cs:12,\pgfkeysvalueof{/pgfplots/ymax});
\coordinate (bot1) at (axis cs:12,\pgfkeysvalueof{/pgfplots/ymin});
%
%
%
%
%
%
\nextgroupplot[title=$P_4$](b)
\addplot[mark = *,red] coordinates { 
	(0,0)
	(1,0)
	(2,-11)
	(3,-11)
	(4,-11)
	(5,-11)
	(6,-18)
	(7,-18)
	(8,-18)
	(9,-18)
	(10,-21)
	(11,-21)
	(12,-21)
	
};
\addplot[mark = *,blue] coordinates {
	(12,-21)
	(13,-21)
	(14,-21)
	(15,-21)
	(16,-17)
	(17,-17)
	(18,-17)
	(19,-17)
	(20,-9)
	(21,-9)
	(22,-9)
	(23,-9)
	(24,0)
	
};
\coordinate (bot) at (rel axis cs:1,0);
\coordinate (top4) at (axis cs:12,\pgfkeysvalueof{/pgfplots/ymax});
\coordinate (bot4) at (axis cs:12,\pgfkeysvalueof{/pgfplots/ymin});
\end{groupplot}
\path (top)--(bot) coordinate[midway] (group center);
\node[above,rotate=90] at (group center -| current bounding box.west) {Coins};
\draw [dashed] (bot1) -- (top1);
\draw [dashed] (bot4) -- (top4);
\end{tikzpicture}}
		\vspace*{-0.8\baselineskip}
	\caption{Coins locked in a run of the 4-party 2-stages Planted Ladder Protocol during the Deposit Phase (in red) and the Claim Phase (in blue).}\label{fig:plantedladder4}
\end{figure}

\begin{figure}[!t]
\centering
	\resizebox*{7cm}{!}{\begin{tikzpicture}
\begin{groupplot}[
group style={group size=4 by 1,vertical sep = 1cm},
height=5cm,width=5cm, ymin = -13, 
ymax = 0]
\nextgroupplot[title=$P_1$,xlabel=Time (Prot.\ Steps),ylabel=(Multiples of $\deposit$),legend to name = curves2]
\addplot[mark = *,red] coordinates { 
	(0,0)
	(1,-8)

};
\addlegendentry{Deposits}
\addplot[mark = *,blue] coordinates { 
	(1,-8)
	(2,0)
	
};
\addlegendentry{Withdrawals}
\coordinate (top1) at (axis cs:1,\pgfkeysvalueof{/pgfplots/ymax});
\coordinate (bot1) at (axis cs:1,\pgfkeysvalueof{/pgfplots/ymin});
%
%
%
%
%
%
\nextgroupplot[title=$P_4$](b)
\addplot[mark = *,red] coordinates { 
	(0,0)
	(1,-12)
	
};
\addplot[mark = *,blue] coordinates {
	(1,-12)
	(2,0)
	
};
\coordinate (bot) at (rel axis cs:1,0);
\coordinate (top4) at (axis cs:1,\pgfkeysvalueof{/pgfplots/ymax});
\coordinate (bot4) at (axis cs:1,\pgfkeysvalueof{/pgfplots/ymin});
\end{groupplot}
\path (top)--(bot) coordinate[midway] (group center);
\node[above,rotate=90] at (group center -| current bounding box.west) {Coins};
\draw [dashed] (bot1) -- (top1);
\draw [dashed] (bot4) -- (top4);
\end{tikzpicture}}
		\vspace*{-0.8\baselineskip}
	\caption{Coins locked in a run of the 4-party Amortized Ladder Protocol during the Deposit Phase (in red) and the Claim Phase (in blue).}\label{fig:amortizedladder4}
\end{figure}

\section{Bitcoin}\label{sec:bitcoin-apdx}
\subsection{Bitcoin Transaction System and Its Malleability Problem}\label{sec:bitcointx}
A coin $\btc$ in Bitcoin is represented by a tuple
$(\btchtx,\btcidx,v,\btcscript(\cdot))$, where 
$\btchtx = \crh{\btctx}$ is the \emph{hash}
of an existing transaction $\btctx$,
$\btcidx$ is the index of the coin in the list of $\btctx$'s outputs,
$v$ is the value of $\btc$, 
and $\btcscript(\cdot)$ is the circuit in which a user, to spend $\btc$, must specify a \emph{witness} $\btcinscript$ s.t.\ $\btcscript(\btcinscript) = 1$.
The coin $\btc$ is called \emph{unspent} if there is no \emph{confirmed transaction}
that refers to $\btc$ as input.  
A multiple inputs/outputs transaction $\btctx'$ that spends $N$ input coins $\{\btc_i = (\btchtx_i,\btcidx_i,v_i,\btcscript_i(\cdot))\}_{i=1}^N$ to $M$ output coins $\{\btc'_j = (\btchtx',j,v'_j,\btcscript'_j(\cdot))\}_{j=1}^M$ is  of the form
$\btctx' = (\{(\btchtx_i,\btcidx_i,v_i,\btcinscript_i)\}_{i=1}^N,\allowbreak\{j,v'_j,\allowbreak\btcscript'_j(\cdot)\}_{j=1}^M,\kt)$,
where $\sum_i v_i \ge \sum_j v'_j$; $\forall i: \btcscript_i(\btcinscript_i) = 1$; and $\kt$ refers to the \texttt{lock\_time} field of a  transaction, which tells the miners \emph{not to confirm} $\btctx'$ until time $\kt$.

The \emph{simplified} form of a transaction that \emph{excludes} the {witness} $\btcinscript_i$ is of the form
	$\btctxsimp' = (\{(\btchtx_i,\btcidx_i,v_i)\}_{i=1}^N,\{j,v'_j,\btcscript'_j(\cdot)\}_{j=1}^M,\kt)$.
This is useful as typically $\btcscript(\btcinscript)$ is simply a signature verification, 
i.e. $\btcinscript = \sigma$ and $\btcscript(\btcinscript) = 1$ iff $\ver_S(\btctxsimp',\sigma) = 1$ where 
$\sigma$ is the signature of the sender on $\btctxsimp'$. 


To briefly recall, the malleability problem happens when a new transaction refers to a coin in a \emph{unconfirmed} transaction as an input, e.g.
\begin{itemize}
	\item A transaction $\btctx_1$ is \emph{created} (but unconfirmed), and spends the input coin $\btc_0$ to a new coin $\btc_1$.
	\item  A new transaction $\btctx_2$ can be created \emph{immediately} to spend the new coin $\btc_1$ to another coin $\btc_2$.
\end{itemize}

Since the validity of $\btctx_2$ now depends on the existence of $\btctx_1$, there is a chance that $\btc_2$ is \emph{invalid}. 
	
	(1) The attacker can simply eavesdrop the blockchain, and upon seeing\footnote{As the Bitcoin blockchain is public, anyone can see a transaction before it is confirmed.} $\btctx_1$, slightly modify $\btctx_1$'s witness $\btcinscript_0$ to obtain a modified $\btctx'_1$ with the new witness $\btcinscript'_0$ by (1) keeping the valid $\sigma_0$ (on \emph{only} the simplified form that does not contains witness $\btcinscript_0$) and (2) adding into $\btcinscript_0$ some Bitcoin op codes that do not affect the circuit evaluation, e.g. adding some additional data to be pushed on the stack prior to the required signatures following with \texttt{OP\_DROP} to leave the stack exactly as before afterwards.
	
	(2) The attacker can now broadcast the $\btctx'_1$ and the two transactions ($\btctx_1$ and $\btctx'_1$) enter a race to be confirmed on the blockchain. If $\btctx'_1$ gets to the blockchain then $\btctx_1$ will not be confirmed and $\btctx_2$ is not valid as it refers to the \emph{unconfirmed and now invalid}
	$\btchtx_1 = \crh{\btctx_1}$ but not the confirmed $\btchtx'_1 = \crh{\btctx'_1}$, and $\btchtx_1 \neq \btchtx'_1$ due to $\btcinscript_0 \neq \btcinscript'_0$ (e.g., $\btcinscript'_0$ contains the harmless additional data and \texttt{OP\_DROP}).

\paragraph*{The Malleability Problem in Bitcoin-based $\FuncCR$.}
Let us recall that $\FuncCR$'s goal is to let a sender $\btcs$ to \emph{conditionally} send some money to a receiver $\btcr$: if $\btcr$ reveals a pre-image
$\kk$ of a commitment $\hk$ (i.e., $\crh{\kk} = \hk$) before a time $\kt$, the receiver $\btcr$ can claim the money otherwise it is refunded to the sender $\btcs$ after time $\kt$. 

To do so, the sender ${\btcs}$ creates a transaction $\btctx_{\mathsf{CR}} = (\btchtx_0,\btcidx_0,v,\allowbreak\sigma_0,1,v,\btcscript_{\mathsf{CR}}(\cdot))$
that spends a coin $\btc_0 = (\btchtx_0, \btcidx_0, v, \btcscript_0(\cdot))$ 
to a new coin $\btc_{\mathsf{CR}}$ with the circuit $\btcscript_{\mathsf{CR}}(\cdot)$ such that $\btcscript_{\mathsf{CR}}(\btcinscript_{\mathsf{CR}}) = 1$ iff
either 
(1) the receiver claims the money with a transaction $\btctx_{\mathsf{C}}$ with the witness $\btcinscript_{\mathsf{CR}} = (\sigma_{\btcr}, \kk)$ such that $(\ver_{\btcr}(\sigma_{\btcr})) \land (\crh{\kk} = \hk)$;
or (2) the sender gets refunded with a transaction $\btctx_{\mathsf{R}}$ with the witness $\btcinscript_{\mathsf{CR}} = (\sigma_{\btcs}, \sigma_{\btcr})$ such that $(\ver_{\btcs}(\sigma_{\btcs})) \land (\ver_{\btcr}(\sigma_{\btcr}))$. 

Indeed the refund transaction $\btctx_{\mathsf{R}}$ \emph{should only need} the signature $\sigma_{\btcs}$ of the sender $\btcs$ but to enforce the time lock $\kt$, the signature of both parties must be involved, i.e. if the $\btctx_{\mathsf{R}}$ only asks for the signature $\sigma_{\btcs}$, the sender $\btcs$ can just create a transaction with time lock $\kt' < \kt$, and thus get refund earlier. Thus, concretely:
\[\btctx_{\mathsf{R}} = (\btchtx_{\mathsf{CR}},1,v,(\sigma_{\btcs}, \sigma_{\btcr}),\btcscript_{\mathsf{R}}(\cdot),\kt).\]
At the same time, if the sender $\btcs$ broadcasts the transaction $\btctx_{\mathsf{CR}}$ immediately and it gets confirmed, s/he suffers from an attack where the receiver $\btcr$ neither reveals the pre-image $\kk$ to claim the money nor allows the sender to get the refund. This attack is possible because the refund transaction $\btctx_{\mathsf{R}}$ asks also for the signature of $\btcr$. If the receiver refuses to sign $\btctx_{\mathsf{R}}$ the money will be locked. As a result the refund transaction $\btctx_{\mathsf{R}}$, which \emph{refers to $\btctx_{\mathsf{CR}}$ as input}, has to be \emph{created before the transaction $\btctx_{\mathsf{CR}}$ is confirmed}. And this is the root of the malleability problem: $\btctx_{\mathsf{R}}$ can be invalidated using the Bitcoin Malleability attack described above.

\subsection{Solutions to Malleability in Protocols Implementations}\label{sec:bitcoin-mall}
A general solution to this problem was proposed in a Bitcoin fork that introduced the Segregated Witness 
transaction, which separates the malleable witness $\btcinscript$ from the computation of the 
transaction id $\btchtx$, {\em i.e.}\ $\btchtx = \crh{\btctxsimp}$ in a \emph{segwit} transaction.\footnote{See 
	\url{https://bitcoincore.org/en/segwit_wallet_dev/}} Thus, a new transaction can ``safely'' 
refer to a previous (possibly unconfirmed) transaction. Yet, a segwit transaction only avoids 
malleability if all inputs are segwit spends, and it is only applicable to the original Bitcoin but not Bitcoin Cash.\footnote{See \url{https://coincentral.com/what-is-bitcoin-cash-vs-bitcoin/}. According to 
	CoinMarketCap (\url{https://coinmarketcap.com/}), as of Aug 2020, Bitcoin Cash is still among the top 10 cryptocurrencies.} Further, referring to an 
unconfirmed transaction is a bad protocol design practice as unconfirmed transactions are 
not guaranteed to be valid later. Hence, a more generally applicable solution would be preferable.

An alternative solution would be to make use of Bitcoin's new op code 
\texttt{OP\_CHECKLOCKTIMEVERIFY},\footnote{See~\url{https://github.com/bitcoin/bips/blob/master/bip-0065.mediawiki}.} which allows a 
transaction output to be made unspendable until some point in the future~\cite{KumaresanB16}. This allows to 
circumvent the malleability problem by letting $\btctx_{\mathsf{R}}$ be generated \emph{after} $\btctx_{\mathsf{CR}}$ is confirmed. 
We refer the reader to \S\ref{sec:bitcoin-fcrfix} for 
further details.
However, this solution only works for $\FuncCR$ but not for $\FuncML$.
In fact, in the implementation from~\cite{KumaresanB14}, all $n$ parties agree on a 
\emph{simplified} lock transaction $\btctxsimp_{\mathsf{lock}} = (\{\btchtx_i,\btcidx_i,x\},\{x,
\btcscript_{\mathsf{lock},i}(\cdot)\})$. The circuit $\btcscript_{\mathsf{lock},i}(\cdot)$ intuitively ensures that
either $\party_i$ reveals its share and claims back its 
deposit (of value $\deposit$) before time $\kt$, or the deposit will be used to compensate all remaining parties 
$\party_{j \neq i}$ (with $\tfrac{\deposit}{n-1}$ coins) after time $\kt$. 
All parties then compute a global $\btchtx_{\mathsf{lock}} = \crh{\btctx_{\mathsf{lock}}}$ using 
MPC,\footnote{To prevent the parties to learn the full form $\btctx_{\mathsf{lock}}$ thus the 
	transaction cannot be posted and lock the money until a later step where all parties agree on all 
	transactions required for $\FuncML$.} where the secret input to each player consists of the witness $\btcinscript_i$.

Unfortunately, to enforce the time lock $\kt$, the compensation transaction $\btctx_{\mathsf{pay},i}$ must be signed by all parties, and must refer to the unconfirmed transaction $\btctx_{\mathsf{lock}}$, which makes the malleability attack possible.
The source of the problem is not just the technique ({\em i.e.}\ the \texttt{lock\_time} field), but also the 
protocol itself: it requires that the lock transaction $\btchtx_{\mathsf{lock}}$ is only 
completed and posted \emph{after} all parties have agreed on the set of 
transactions $\{\btctx_{\mathsf{pay},i}\}_{i \in [n]}$, otherwise a malicious party $\party_j$ 
can just refuse to sign a $\btctx_{\mathsf{pay},i}$ for $\party_i$, while the honest but 
unfortunate $\party_i$ has signed $\btctx_{\mathsf{pay},j}$, which means that $\party_j$ can claim 
back the money but $\party_i$ cannot. 
\vspace{-0.3\baselineskip}
\subsection{Fixing the Bitcoin-based $\FuncCR$}\label{sec:bitcoin-fcrfix}
Using \texttt{OP\_CHECKLOCKTIMEVERIFY}, the lock time $\kt$ can now be embedded directly into the circuit $\btcscript_{\mathsf{CR}}(\cdot)$: when verifying the refund transaction $\btctx_{\mathsf{R}}$ with the witness $\btcinscript_{\mathsf{CR}} = (\sigma_{\btcs})$, all miners (automatically) take an additional input $t$ that is the current block height, and verify that $\ver_{\btcs}(\sigma_{\btcs})$ and $t \ge \kt$.\footnote{This also technically enforces the $lock\_time$ field of the transaction $\btctx_{\mathsf{R}}$ to have a value at least $\kt$.}

Concretely, a high-level view of the Bitcoin script for $\btcscript_{\mathsf{CR}}(\cdot)$ is as follows~\cite{KumaresanB16}.

\begin{framed}\footnotesize

	\noindent\texttt{<$\kt$> CHECKLOCKTIMEVERIFY 
	\\\noindent IF HASH256 <$\hk$> EQUALVERIFY 
	 <$\pk_{\btcr}$> CHECKSIGVERIFY 
	\\\noindent ELSE <$\pk_{\btcs}$> CHECKSIGVERIFY 
	\\\noindent ENDIF}
\end{framed}
Since the lock time is now directly in the circuit $\btcscript_{\mathsf{CR}}(\cdot)$ it allows the refund transaction $\btctx_{\mathsf{R}}$ to be created \emph{after} $\btctx_{\mathsf{CR}}$ is \emph{confirmed} as there is no longer a need for the signature of $\btcr$ on $\btctx_{\mathsf{R}}$.

\paragraph*{Bitcoin Realization of $\FuncML$}\label{sec:bitcoin-impl}
We show how to implement $\FuncML$ using the new op 
code \texttt{OP\_CHECKLOCKTIMEVERIFY}.

	\begin{framed}\footnotesize
		\begin{center}
			\textbf{Bitcoin Realization of $\FuncML$}
		\end{center}
		To transact in Bitcoin, a sender $\btcs$ generates a pair of public and private keys $\pk_{\btcs}$ and $\sk_{\btcs}$ (the so-called \emph{addresses}). Similarly, we write $(\pk_{\btcr},\sk_{\btcr})$ for the key of the receiver $\btcr$.
		Let us denote by $\sigma \getsr \sig_{\btcs}(\msg)$  the signature of $\btcs$ on $\msg$ (using $\sk_{\btcs}$), and by $\ver_{\btcs}(\msg,\sigma)$  the signature verification function (using $\pk_{\btcs}$).
		Since the witness $\btcinscript$ contains the signature $\sigma$, it is not possible to include $\btcinscript$ in $\msg$, and thus signing a transaction $\btctx$ means signing only its simplified form ({\em i.e.},\ $\msg = \btctxsimp$).
		We assume that each $\party_i$ owns a coin $\btc_i = (\btchtx_i, \btcidx_i, (n-1)q, \btcscript_i(\cdot))$ that is spendable only by $\party_i$ signing the spending transaction.
		\begin{description}
			\item[Lock Phase:] The players proceed as follows.
			\begin{enumerate}[leftmargin=-0.3cm]
				\item Each $\party_i$ ($i \in [n]$) creates $\btcidx_{\mathsf{lock}^i_j} = (i(n-1)+j$
				and
								\begin{align*}
				\btctxsimp_{\mathsf{lock}} = & (\{\btchtx_i, \btcidx_i, (n-1)q\}_{i \in [n]},\\ 
				& 
				 \{\btcidx_{\mathsf{lock}^i_j}, q,
				\btcscript^i_j(\cdot)\}_{i,j \in [n] \land i \neq j})
								\end{align*}
				where $\btcscript^i_j(\cdot)$ takes as input either (i) before time $\kt$ the witness 
				$\btcinscript^i_j = (\kk_i, \sigma_i)$
				and returns 1 iff $\crh{\kk_i} = \hk_i$ and $\ver_i(\sigma_i) = 1$; or (ii) after time $\kt$ the witness 
				$\btcinscript^i_j = \sigma_j$ 
				and returns 1 iff $\ver_j(\sigma_j)$.
				\item Each $\party_i$ ($i \in [n]$) generates $\btcinscript_i \getsr \sig_{i}(\btctxsimp_{\mathsf{lock}})$.
				\item Each $\party_i$ ($i \in [n)$) sends $\btcinscript_i$ to $\party_n$.
				\item $\party_n$ creates 
				$\btctx_{\mathsf{lock}} = (\{\btchtx_i, \btcidx_i, \allowbreak(n-1)q, \btcinscript_i\}_{i \in [n]},\allowbreak \{\btcidx_{\mathsf{lock}^i_j},q,\allowbreak\btcscript^i_j\}_{i,j \in [n] \land i \neq j})$ and sends it to the ledger.
			\end{enumerate}
			\item[Redeem Phase:] Before time $\kt$:
			\begin{enumerate}[leftmargin=-0.3cm]
				\item \noindent  Each $\party_i$ ($i \in [n]$) obtains $\btchtx_{\mathsf{lock}}$ from the ledger.
				\item Each $\party_i$ ($i \in [n]$) creates and sends to the ledger
								\begin{align*}
				\btctx_{\mathsf{redeem},i} = & (\{\btchtx_{\mathsf{lock}}, \btcidx_{\mathsf{lock}^i_j}, q,\\ & (\kk_i, \sigma_i)\}_{j \in [n] \land j \neq i},
				(1, (n-1)q, \btcscript_\mathsf{new}^i(\cdot))),
								\end{align*}
				where $\sigma_i \getsr \sig_i(\btctxsimp_{\mathsf{redeem},i})$ and $\btcscript_\mathsf{new}^i(\cdot)$ is any Bitcoin script (typically just a signature verification).
			\end{enumerate}
			\item[Compensation Phase:] After time $\kt$:
			\begin{enumerate}[leftmargin=-0.3cm]
				\item Each $\party_j$ ($j \in [n]$) obtains $\btchtx_{\mathsf{lock}}$ from the ledger.
				\item For the unclaimed $\{\btc_{\mathsf{lock}^i_j}\}$ by $\party_i$ in the redeem phase, each $\party_j$ creates and send to the ledger
				\begin{align*}
				\btctx_{\mathsf{compensate}^j} &= (\{\btchtx_{\mathsf{lock}},\btcidx_{\mathsf{lock}^i_j}, q, (\sigma_j)\}, \\
				&\qquad(1, q, \btcscript_\mathsf{new'}^j(\cdot))),
				\end{align*}
				where $\sigma_j \getsr \sig_j(\btctxsimp_{\mathsf{compensate}^j})$ and $\btcscript_\mathsf{new'}^j(\cdot)$ is any Bitcoin script (typically just a signature verification).
			\end{enumerate}
		\end{description}
		\vspace*{-\baselineskip}
	\end{framed}

Our implementation is inspired by both the Bitcoin-based Timed Commitment 
scheme~\cite{AndrychowiczDMM14} and the original Bitcoin-based $\FuncML$.\footnote{We also ignore transaction fees. 
	However one can replace $q$ with $q' < q$ in our protocol to include 
	any reasonable transaction fee $q - q'$.}
, but 
we avoid the malleability problem by \emph{not referring to any unconfirmed 
transaction} in our protocol design.
Note that, in our implementation, \emph{all transactions refers only to the confirmed transactions}, hence our Bitcoin realization of $\FuncML$ does not suffer from the malleability problem as in~\cite{KumaresanB14}.

It is easy to see that our construction correctly implements $\FuncML$. In fact, a total of $n(n-1)q$ coins are locked during the Lock Phase. An honest party can redeem his own $(n-1)q$ coins during the Redeem Phase by revealing the committed secret $\kk_i$ before time $\kt$. Otherwise, after time $\kt$, in the Compensation Phase, an honest party gets $q$ coins as compensation (for each dishonest player that did not open).
Security relies on the honest majority of the Bitcoin network to enforce script validation and resist double spending~\cite{nakamoto2008bitcoin}. 
For simplicity, we assume a party is rational in the sense that it always redeems a coin as soon as it becomes possible.

In the Lock Phase, transaction $\btctx_{\mathsf{lock}}$ is co-created by all parties (who also make sure that they sign the correct transaction $\btctxsimp_{\mathsf{lock}}$), and takes the spendable coins $\{\btc_i\}_{i \in [n]}$ as input. Once the $\btctx_{\mathsf{lock}}$ is confirmed by the miners, the coins $\{\btc_i\}_{i \in [n]}$ will be locked and not spendable.
A high-level Bitcoin script for $\btcscript^i_j(\cdot)$ would be (see the appendix~\ref{apdx:demo} for more details on the code):

	\texttt{IF HASH256 <$\hk_i$> EQUALVERIFY <$\pk_{i}$> CHECKSIG
		\\\indent ELSE <$\kt$> CHECKLOCKTIMEVERIFY DROP <$\pk_{j}$> CHECKSIG 
		ENDIF}

In the Redeem Phase, only a $\party_i$ that holds the committed secret $\kk_i$ and the corresponding signing key $\sk_i$ can redeem the $n-1$ outputs (at indices $\btcidx_{\mathsf{lock},i,j}$) of $\btctx_{\mathsf{lock}}$ (for a total of $q(n-1)$ coins). This is done by following the IF path of $\btcscript_{i,j}(\cdot)$, {\em i.e.}\ revealing $\kk_i$ and signing the transaction $\btctx_{\mathsf{redeem},i}$ as specified in each circuit $\btcscript_{i,j}$. 
Now, if another party $\party_j$ wants to redeem the same coins as $\party_i$ (at indices $\btcidx_{\mathsf{lock},i,j}$ of $\btctx_{\mathsf{lock}}$), it can only follow the ELSE path. Even though $\party_j$ can create the $\btctx_{\mathsf{compensate},j}$ by itself, the transaction will fail at \texttt{CHECKLOCKTIMEVERIFY} as $\kt$ has not been reached.

In the Compensation Phase, once time $\kt$ has passed, each $\party_j$ can claim the compensation (index $\btcidx_{\mathsf{lock},i,j}$ of $\btctx_{\mathsf{lock}}$) from each dishonest $\party_i$, by signing  $\btctx_{\mathsf{compensate},j}$. (This transaction will pass \texttt{CHECKLOCKTIMEVERIFY}, as $\kt$ is reached.)
$\party_j$ cannot claim the coins from the honest $\party_i$, as $\party_i$ has already redeemed them.

\subsection{Snapshots from the Demo}\label{apdx:demo}
In what follow we show some transactions obtained from our proof of concept implementation of our Bitcoin-based $\FuncML$ for the two party case, running in our own Bitcoin test net. 

Alice, whose receiving address\footnote{Ideally, Alice should use two separate addresses for the two receiving cases, i.e. redeem and compensate, but for simplicity of exposition we use the same address in this demo. The same applies to Bob.} is \texttt{2Mxb\ldots eXuq} and Bob, whose receiving address is \texttt{mrsU\ldots Ke8T}, will together create the Pay-To-Script-Hash (P2SH) lock transaction $\btctx_{\mathsf{lock}}$ that consists of two script-hashes as outputs: \texttt{2NGY\ldots WypP} and \texttt{2MvN\ldots DoBd}.

\texttt{2NGY\ldots WypP} is the hash of the first output script
\begin{footnotesize}
\begin{verbatim}
OP_IF 
OP_HASH256 
32 0xbcdc...f514 OP_EQUALVERIFY 
33 0x0399...03be OP_CHECKSIG 
OP_ELSE 
2 0x9600 OP_NOP2 OP_DROP 
33 0x024e...3c10 OP_CHECKSIG 
OP_ENDIF
\end{verbatim}
\end{footnotesize}
while \texttt{2MvN\ldots DoBd} is the hash of the second output script
\begin{footnotesize}
\begin{verbatim}
OP_IF 
OP_HASH256 
32 0xbba2...76e4 OP_EQUALVERIFY 
33 0x024e...3c10 OP_CHECKSIG 
OP_ELSE 
2 0x9600 OP_NOP2 OP_DROP 
33 0x0399...03be OP_CHECKSIG 
OP_ENDIF
\end{verbatim}
\end{footnotesize}

Notice that \texttt{CHECKLOCKTIMEVERIFY} is implemented using \texttt{OP\_NOP2}.\footnote{See \url{https://github.com/bitcoin/bips/blob/master/bip-0065.mediawiki\#Detailed_Specification}.}
As a result of posting $\btctx_{\mathsf{lock}}$ both parties will obtain the transaction id \texttt{a21e\ldots 3a84}.

To redeem \texttt{2NGY\ldots WypP}, Alice can simply post a signed redeem transaction $\btctx_{\mathsf{redeem},A}$ that spends the first output of \texttt{a21e\ldots 3a84} and includes the secret `alicesSecret' (whose double-SHA256 is \texttt{0xbcdc...f514}) in the witness. To redeem \texttt{2MvN\ldots DoBd} Bob can similarly post (the signed transaction) $\btctx_{\mathsf{redeem},B}$ that spends the second output of \texttt{a21e\ldots 3a84} and includes the secret `bobsSecret' (whose double-SHA256 is \texttt{0xbba2...76e4}) in the witness.

If either Alice or Bob (or both) is dishonest, the remaining party can simply send a signed transaction after the specified lock time (as an example we specify it as the 150-th block, \texttt{0x9600}) to get compensated.
\paragraph*{Practical Limitations.}
In our realization of $\FuncML$ using Bitcoin, only the  
transaction $\btctx_{\mathsf{lock}}$ is a \emph{standard} P2SH transaction, {\em i.e.}\ 
the small set of believed-to-be-safe templates.\footnote{\url{https://bitcoin.org/en/transactions-guide\#standard-transactions}.} The 
other two transactions, {\em i.e.}\ $\btctx_{\mathsf{redeem},*}$ and 
$\btctx_{\mathsf{compensate},i}$, are \emph{non-standard} 
and thus might not be accepted by Bitcoin miners. 

Another practical limitation is w.r.t.\ the number of parties running $
\FuncML$. As of Bitcoin core 0.9.3: ``\ldots The transaction must be smaller than 
100,000 bytes. That is around 200 times larger than a typical single-input, single-output P2PKH transaction \ldots''.
Hence, our Bitcoin-based $\FuncML$ 
can support up to 14 parties in an MPC. This is sufficient for poker, but 
barely enough for more challenging applications such as
secure distributed trading \cite{massacci2018fmex}. Yet, practical MPC with a large number 
of traders was also unreachable for the distributed trading protocol itself. Scaling up is a challenge for the 
community.

\section{Protocol Extensions}\label{app:ext}
\subsection{Reactive \CMLMech}
In the reactive setting, we have to securely 
realize ideal functionalities whose computation proceeds in stages. 
After each phase, the players receive a 
partial output and a state, which both influence the choice of the inputs for the next 
stage. Normally, simulation-based security for non-reactive 
functionalities implies the same flavor of security for reactive ones~\cite{HazayL10}.
However, for secure computation with penalties the naive approach ({\em i.e.}\ run 
an independent instance of the penalty protocol for each stage of the 
functionality) fails in case of aborts after a given stage is concluded 
(and before the next stage starts) 
as one needs a mechanism to force the parties to continue to the next stage.

Let $\pi$ be a protocol for securely realizing (with aborts) a reactive functionality:
during each 
stage, every player independently computes its next message as a function of its 
current input and transcript so far, and broadcasts this message to all other parties.  
Kumaresan {\em et al.}~\cite{kumaresan2016improvements} show how to deal with 
such protocols
using \PLMech. We can do the same for \CLMech\ by leveraging 
the 
power of $\FuncML$. In particular, we consider an invocation of $\FuncML$ 
for each stage of $\pi$, where during the Lock Phase each player $\party_i$ specifies a circuit $\phi_i$ that checks the correctness of $\party_i$'s next 
message w.r.t.\ the protocol transcript so far. 
This is possible so long as the 
underlying MPC is publicly verifiable, a property also used 
in~\cite{kumaresan2016improvements}. The price to pay  is 
a larger communication/round complexity.

The above requires to augment the $\FuncML$ functionality so that each player 
can deposit $\coins(d\cdot \bar\tau)$, where $\bar\tau$ is 
the maximum number of stages in a run of $\pi$, and claim back at most $\coins(d)$ 
for each 
stage ({\em i.e.},\ after revealing its next message).

\subsection{BoBW \CMLMech}
A drawback of cryptographic fairness with penalties is that a {\em single} corrupted player 
can cause the protocol to abort (at the price of compensating the other players).
Ideally, we would like to have a protocol such that when $s_1 < n/2$ players 
are corrupted the protocol achieves full security, whereas in case the number of 
corrupted parties is $s_2 \ge n/2$, the protocol achieves security with aborts or, even 
better, fairness with penalties.
This yields so-called best-of-both-worlds (BoBW) security~\cite{IshaiKKLP11}, which is known to be impossible in the parameters regime $s_1 + s_2 \ge n$.

Kumaresan {\em et al.}~\cite{kumaresan2016improvements} provide a dual mode 
protocol achieving BoBW security for any  $s_1 + s_2 < n$. We can 
easily adapt their approach---and, in fact, even simplify it thanks to the power of $\FuncML$---to our \CMLMech\ protocol, by using the following 
modifications:
(i) The function $\tilde\func$ now computes an $(s_2+1,n)$-threshold secret sharing 
of the symmetric key $\key$;
(ii) At the end of the Reconstruction Phase, all honest parties broadcast their share.
Modification (i) ensures that, at the end of the Share Distribution Phase, an attacker 
controlling at most $s_2$ parties has no information on the output. The dual-mode protocol of~\cite{kumaresan2016improvements} relies on \PLMech,\ as \LMech\ does not guarantee compensation to honest 
parties that did not reveal their share (say, due to an abort during the protocol); this eventuality is not possible given the atomic nature of $
\FuncML$, and thus no further change to $\CMLMech$ is required 
to obtain fairness with penalties in the presence of $s_2$ corrupted parties.
Modification (ii) is needed because when $s_1 < n - s_2$ parties are corrupted, 
there are at least $n - s_2 \ge s_1 + 1$ honest parties, which allows everyone to 
 obtain the output by correctness of the secret sharing scheme.

\clearpage
\fi

\end{document}